\documentclass{article}

\usepackage{geometry}
\usepackage{bm}
\usepackage{epstopdf}
\usepackage{easybmat}
\usepackage{dsfont}
\usepackage{tabularx}
\usepackage{etex}
\usepackage{stmaryrd}
\usepackage{amsthm}
\usepackage{amsmath}
\usepackage{amsfonts}
\usepackage{amssymb}
\usepackage{graphicx}
\usepackage{array}
\usepackage{booktabs}
\usepackage[usenames]{color}
\usepackage{epstopdf}
\usepackage{caption}
\usepackage{subcaption}
\usepackage{float}
\usepackage{algorithm}
\usepackage{algpseudocode}
\usepackage{nomencl}
\usepackage[font={small}]{caption}
\usepackage{url}

\newtheorem{theorem}{Theorem}[section]

\newtheorem{corollary}[theorem]{Corollary}

\newtheorem{definition}[theorem]{Definition}

\newtheorem{lemma}[theorem]{Lemma}

\newtheorem{proposition}[theorem]{Proposition}
\newtheorem{remark}[theorem]{Remark}

\algrenewcommand\alglinenumber[1]{\tiny #1:}
\newcommand{\T}[1]{^\mathsf{#1}}
\newcommand{\argmax}[1]{\underset{#1}{\operatorname{arg}\,\operatorname{max}}\;}
\newcommand{\argmin}[1]{\underset{#1}{\operatorname{arg}\,\operatorname{min}}\;}
\newcommand{\minimize}[1]{\underset{#1}{\operatorname{minimize}}\;}
\newcommand{\maximize}[1]{\underset{#1}{\operatorname{maximize}}\;}
\newcommand{\I}[1]{\llbracket {#1}\rrbracket_I}
\newcommand{\iI}[1]{\llbracket {#1}\rrbracket_{\mathbb{I}}}
\newcommand{\XI}[1]{\llbracket {#1}\rrbracket_{X,I}}
\newcommand{\II}{\mathbb{I}}
\newcommand{\y}{\widetilde{y}}
\newcommand{\X}{\widetilde{X}}
\newcommand{\Beta}{\widetilde{\beta}}

\newcommand{\ES}[1]{\textnormal{\tiny {#1}}}
\renewcommand{\arraystretch}{1.2}
\newcommand{\Chi}{\raisebox{2pt}{$\chi$}}

\numberwithin{equation}{section}
\newgeometry{tmargin=3cm, bmargin=3cm} 

\makenomenclature

\floatname{algorithm}{Procedure}

\begin{document}
\begin{center}{\Large\bf  Group SLOPE - adaptive selection of groups of predictors}
\footnote{An earlier version of the paper appeared on arXiv.org in November 2015: arXiv:1511.09078 }
\vspace{0.2in}

\noindent{Damian Brzyski$^{a,b}$,  Alexej Gossmann$^{c}$, Weijie Su$^{d}$, Ma\l{}gorzata Bogdan$^{e}$}
\vspace{0.2in}

\noindent $\mbox{}^{a}$ {\it \footnotesize Department of Epidemiology and Biostatistics,  Indiana University, Bloomington, IN 47405, USA}

\noindent $\mbox{}^{b}$ {\it \footnotesize  Institute of Mathematics, Jagiellonian University, 30-348 Cracow, Poland}

\noindent $\mbox{}^{c}$ {\it \footnotesize  Department of Mathematics, Tulane University, New Orleans, LA 70118, USA}
 
\noindent  $\mbox{}^{d}$ {\it \footnotesize  Department of Statistics, University of Pennsylvania, Philadelphia, PA 19104, USA}

\noindent $\mbox{}^{e}$ {\it \footnotesize  Institute of Mathematics, University of Wroclaw, 50-384 Wroclaw, Poland}

\vspace{0.2in}


\end{center}

\begin{abstract}
Sorted L-One Penalized Estimation (SLOPE, \cite{SLOPE}) is a relatively new convex optimization procedure which allows for adaptive selection of regressors under sparse high dimensional designs. Here we extend the idea of SLOPE to deal with the situation when one aims at selecting whole groups of explanatory variables instead of single regressors. Such groups can be formed by clustering strongly correlated predictors or groups of dummy variables corresponding to different levels of the same qualitative predictor. We formulate the respective convex optimization problem, gSLOPE (group SLOPE), and propose an efficient algorithm for its solution. We also define a notion of the group false discovery rate (gFDR) and provide a choice of the sequence of tuning parameters for gSLOPE so that gFDR is provably controlled at a prespecified level if the groups of variables are orthogonal to each other. Moreover, we prove that the resulting procedure adapts to unknown sparsity and is asymptotically minimax with respect to the estimation of the proportions of variance of the response variable explained by regressors from different groups. We also provide a method for the choice of the regularizing sequence  when variables in different groups are not orthogonal but statistically independent and illustrate its good properties  with computer simulations. 
Finally, we illustrate the advantages of gSLOPE  in the context of Genome Wide Association Studies. \texttt{R} package \texttt{grpSLOPE} with implementation of our method is available on CRAN.
\end{abstract}


\section{Introduction}

Consider the classical multiple regression model of the form
\begin{equation}
\label{mainmodel}
y=X\beta+z,\end{equation} where $y$ is the $n$ dimensional vector of values of the response variable, $X$ is the $n$ by $p$ experiment (design) matrix and $z\sim\mathcal{N}(0,\sigma^2 \mathbf{I}_n)$. We assume that $y$ and $X$ are known, while $\beta$ is unknown. In many applications the purpose of the statistical analysis is to recover the support of $\beta$, which identifies the set of important regressors. Here, the true support corresponds to truly relevant variables (i.e. variables which have impact on observations).  Common procedures to solve this model selection problem rely on minimization of some objective function consisting of the weighted sum of two components: first term responsible for the goodness of fit and second term penalizing the model complexity. Among such procedures one can mention classical model selection criteria like the Akaike Information Criterion (AIC) \cite{Aka} and the Bayesian Information Criterion (BIC) \cite{Schw}, where the  penalty depends on the number of variables included in the model,  or  LASSO \cite{LASSO}, where the penalty depends on the $\ell_1$ norm of regression coefficients. The main advantage of LASSO over classical model selection criteria is that it is a convex optimization problem and, as such, it can be easily solved even for very large design matrices.  

LASSO solution is obtained by solving the optimization problem 
\begin{equation}
\label{LASSO}
\argmin b\ \ \bigg\{\frac 12\big\|y-Xb\big\|^2+\lambda_L\|b\|_1\bigg\},
\end{equation}
where $\lambda_L$ is a tuning parameter defining the trade-off between the model fit and the sparsity of solution. 
In practical applications the selection of good $\lambda_L$ might be very challenging. For example it has been reported that in high dimensional settings  the popular cross-validation  typically leads  to detection of  a large number of false regressors (see e.g. \cite{SLOPE}). The general rule is that when one reduces $\lambda_L$, then LASSO can identify more elements from the true support (true discoveries) but at the same time it generates more false discoveries.  In general the numbers of true and false discoveries for a given  $\lambda_L$ depend on unknown properties on the data generating mechanism, like the number of true regressors and the magnitude of their effects. A very similar problem occurs when selecting  thresholds for individual tests in the context of multiple testing. Here it was found that the popular Benjamini-Hochberg rule (BH, \cite{BH}), aimed at control of the False Discovery Rate (FDR), adapts to the unknown data generating mechanism and has some desirable optimality properties under a variety of statistical settings (see e.g. \cite{ABDJ, ABOS, ABOS2, ABOS3}). The main property of this rule is that it relaxes the thresholds along the sequence of test statistics, sorted in the decreased order of magnitude. Recently the same idea was used in a new generalization of LASSO, named SLOPE (Sorted L-One Penalized Estimation, \cite{SLOPE2, SLOPE}). Instead of the $\ell_1$ norm (as in LASSO case), the method uses FDR control properties of $J_{\lambda}$ norm, defined as follows; for sequence $\{\lambda\}_{i=1}^p$ satisfying $\lambda_1\geq\ldots\geq\lambda_p\geq0$ and $b\in\mathbb{R}^p$, $J_{\lambda}(b):=\sum_{i=1}^p\lambda_i|b|_{(i)},$ where $|b|_{(1)}\geq\ldots\geq |b|_{(p)}$ is the vector of sorted absolute values of coordinates of $b$. SLOPE is the solution to a convex optimization problem
\begin{equation}
\label{SLOPE}
\argmin b\ \ \bigg\{\frac 12\big\|y-Xb\big\|^2+J_{\lambda}(b)\bigg\},
\end{equation}
which clearly reduces to LASSO  for $\lambda_1=\ldots=\lambda_p=:\lambda_L$.
Similarly as in classical model selection, the support of solution defines the subset of variables estimated as relevant. In \cite{SLOPE} it is shown that when the sequence $\lambda$ corresponds to the decreasing sequence of threshold for BH then SLOPE controls  FDR under orthogonal designs, i.e. when $X^TX=\mathbf{I}_n$. Moreover, in \cite{WE} it is proved that SLOPE with this sequence of tuning parameters adapts to unknown sparsity and is asymptotically minimax under orthogonal and random Gaussian designs.  

In the sequence of examples presented in \cite{SLOPE2}, \cite{SLOPE} and \cite{geneSLOPE} it was shown that SLOPE has very desirable properties in terms of FDR control in case when regressor variables are weakly correlated. While there exist other interesting approaches which allow to control FDR under correlated designs (e.g., \cite{ko}),  the efforts to prevent detection of false regressors which are strongly correlated with true ones inevitably lead to a loss of power. 
An alternative approach to deal with strongly correlated predictors is to simply give up the idea of distinguishing between them and include all of them into the selected model as a group.
This leads to the problem of group selection in linear regression, extensively investigated and applied in many fields of science.
In many of these applications the groups are selected not only due to the strong correlations but also by taking into account the problem specific scientific knowledge. It is also common to cluster dummy variables corresponding to different levels of qualitative predictors.

Probably the most known convex optimization method for selection of groups of explanatory variables is the group LASSO (gLASSO) \cite{Bakin}. For a fixed tuning parameter, $\lambda_{gL}>0$, the gLASSO estimate is most frequently (e.g. \cite{gLASSO1}, \cite{gLASSO5}) defined as a solution to optimization problem
\begin{equation}
\label{gLASSO}
\argmin b\ \ \bigg\{\frac 12\Big\|y-\sum_{i=1}^mX_{I_i}b_{I_i}\Big\|_2^2+\sigma\lambda_{gL}\sum_{i=1}^m\sqrt{|I_i|}\|b_{I_i}\|_2\bigg\},
\end{equation}
where the sets $I_1,\ldots,I_m$ form a partition of the set $\{1,\ldots,p\}$, $|I_i|$ denotes the number of elements in set $I_i$, $X_{I_i}$ is the submatrix of $X$ composed of columns indexed by $I_i$ and $b_{I_i}$ is the restriction of $b$ to indices from $I_i$. The method introduced in this article is, however, closer to the alternative version of gLASSO, in which penalties are imposed on $\|X_{I_i}b_{I_i}\|_2$ rather than $\|b_{I_i}\|_2$. This method was formulated in \cite{gLASSO6}, where authors defined estimate of $\beta$ as
\begin{equation}
\label{gLASSO2}
\beta^\ES{gL}:=\argmin b\ \ \bigg\{\frac 12\Big\|y-\sum_{i=1}^mX_{I_i}b_{I_i}\Big\|_2^2+\sigma\lambda_{gL}\sum_{i=1}^m\sqrt{|I_i|}\|X_{I_i}b_{I_i}\|_2\bigg\},
\end{equation}
with the condition $\|X_{I_i}\beta^\ES{gL}_{I_i}\|_2>0$ serving as a group relevance indicator.

Similarly as in the context of regular model selection, the properties of gLASSO strongly depend on the shrinkage parameter $\lambda_{gL}$, whose optimal value is the function of unknown parameters of true data generating mechanism. Thus, a natural question arises if the idea of SLOPE can be used for construction of the similar adaptive procedure  for the group selection. To answer this query in this paper we define and investigate the properties of the group SLOPE (gSLOPE). We formulate the respective optimization problem  and provide the algorithm for its solution. We also define the notion of the group FDR (gFDR), and provide the theoretical choice of the sequence of regularization parameters, which guarantees that gSLOPE controls gFDR in the situation when variables in different groups are orthogonal to each other. Moreover, we prove that the resulting procedure adapts to unknown sparsity and is asymptotically minimax with respect to the estimation of the proportions of variance of the response variable explained by regressors from different groups. Additionally, we provide a way of constructing the sequence of regularization parameters under the assumption that the regressors from distinct groups are independent and use computer simulations to show that it allows to control gFDR. Good properties of group SLOPE are illustrated using the practical example of Genome Wide Association Study. \texttt{R} package \texttt{grpSLOPE} with implementation of our method is available on CRAN.

\section{Group SLOPE}
\subsection{Formulation of the optimization problem}
\label{subsec:gs2711944}
Let the design matrix $X$ belong to the space $M(n,p)$ of matrices with $n$ rows and $p$ columns. Furthermore, suppose that $I=\{I_1,\ldots,I_m\}$ is some partition of the set $\{1,\ldots,p\}$, i.e. $I_i$'s are nonempty sets, $I_i\cap I_j=\emptyset$ for $i\neq j$ and $\bigcup I_i = \{1,\ldots,p\}$. We will consider the linear regression model with $m$ groups of the form
\begin{equation}
\label{gmodel}
y=\sum_{i=1}^mX_{I_i}\beta_{I_i}+z,
\end{equation}
where $X_{I_i}$ is the submatrix of $X$ composed of columns indexed by $I_i$ and $\beta_{I_i}$ is the restriction of $\beta$ to indices from the set $I_i$. We will use notations $l_1,\ldots,l_m$ to refer to the ranks of submatrices $X_{I_1},\ldots,X_{I_m}$. To simplify notations in further part, we will assume that $l_i>0$ (i.e. there is at least one nonzero entry of $X_{I_i}$ for all $i$). Besides this, $X$ may be absolutely arbitrary matrix, in particular any linear dependencies inside submatrices $X_{I_i}$ are allowed.

In this article we will treat the value $\|X_{I_i}\beta_{I_i}\|_2$ as a measure of an impact of $i$th group on the response and we will say that the group $i$ is truly relevant if and only if $\|X_{I_i}\beta_{I_i}\|_2>0$. Thus our task of the identification of the relevant groups is equivalent with finding the support of the vector $\XI{\beta}:= \big(\|X_{I_1} \beta_{I_1}\|_2, \ldots, \|X_{I_m} \beta_{I_m}\|_2\big)^\mathsf{T}$.

To estimate the nonzero coefficients of $\XI{\beta}$, we will use a new penalized method, namely group SLOPE (gSLOPE). For a given nonincreasing sequence of  nonnegative tuning parameters, $\lambda_1,\ldots,\lambda_m$, a given sequence of positive weights, $w_1,\ldots,w_m$, and a design matrix, $X$, the gSLOPE, $\beta^\ES{gS}$, is defined as solutions to
\begin{equation}
\label{gSLOPE}
\beta^\ES{gS}: = \argmin b\ \ \Big\{\frac 12\big\|y-Xb\big\|_2^2+\sigma J_{\lambda}\big( W\XI{b}\big)\Big\},
\end{equation}
where $W$ is a diagonal matrix with
$W_{i,i}:=w_i,$ for $i=1,\ldots,m$. The estimate of $\XI{\beta}$ support is simply defined by the indices corresponding to nonzeros of $\XI{\beta^\ES{gS}}$.

It is easy to see that when one considers $p$ groups containing only one variable (i.e. singleton groups situation), then taking all weights equal to one reduces (\ref{gSLOPE}) to SLOPE (\ref{SLOPE}). On the other hand, taking $w_i=\sqrt{|I_i|}$ and putting $\lambda_1=\ldots=\lambda_m=:\lambda_{gL}$, immediately gives gLASSO problem (\ref{gLASSO2}) with the smoothing parameter $\lambda_{gL}$. The gSLOPE could be therefore treated both: as the extension to SLOPE, and the extension to group LASSO.

Now, let us define $\widetilde{p}=l_1+\ldots+l_m$ and consider the following partition, $\II=\{\II_1,\ldots,\II_m\}$, of the set $\{1,\ldots,\widetilde{p}\}$
\begin{equation}
\II_1:=\big\{1,\ldots, l_1\big\},\quad \II_2:=\big\{l_1+1,\ldots, l_1+l_2\big\},\quad \ldots,\quad \II_m:=\Big\{\sum_{j=1}^{m-1}l_i+1,\ldots, \sum_{j=1}^{m}l_i\Big\}.
\end{equation}
Observe that each $X_{I_i}$ can be represented as $X_{I_i}=U_iR_i$, where $U_i$ is a matrix with $l_i$ orthogonal columns of a unit $l_2$ norm, whose span coincides with the space spanned by the columns of $X_{I_i}$ , and $R_i$ is the corresponding matrix of a full row rank. Define $n$ by $l$ matrix $\widetilde{X}$ by putting $\widetilde{X}_{\II_i}:=U_i$ for $i=1,\ldots,m$. Now observe that denoting  $c_{\II_i}:=R_ib_{I_i}$ for $i\in\{1,\ldots,m\}$ we immediately obtain
\begin{equation}
\arraycolsep=1.4pt\def\arraystretch{1.8}
\begin{array}{c}
Xb=\sum\nolimits_{i=1}^mX_{I_i}b_{I_i} = \sum\nolimits_{i=1}^mU_iR_ib_{I_i} = \sum\nolimits_{i=1}^m\widetilde{X}_{\II_i}c_{\II_i}=\widetilde{X}c,\\
\Big(\XI{b}\Big)_i=\|X_{I_i}b_{I_i}\|_2 = \|R_ib_{I_i}\|_2= \|c_{\II_i}\|_2	
\end{array}
\end{equation}
and the problem (\ref{gSLOPE}) can be equivalently presented in the form
\begin{equation}
\label{24111029}
\left\{
\begin{array}{l}
c^\ES{gS}: = \argmin {c}\ \ \bigg\{\frac12\big\|y-\widetilde{X}c\big\|_2^2+\sigma J_{\lambda}\Big( W\iI{c}\Big)\bigg\} \\
c^\ES{gS}_{\II_i}:=R_i\beta^\ES{gS}_{I_i},\ i=1,\ldots,m
\end{array}
\right.,
\end{equation}
for $\iI{c}:= \big(\|c_{\II_1}\|_2, \ldots, \|c_{\II_m}\|_2\big)^\mathsf{T}.$ 
Therefore to identify the relevant groups and estimate their group effects it is enough to solve the optimization problem \eqref{24111029}. We will say that (\ref{24111029}) is the standardized version of the problem \eqref{gSLOPE}.

\begin{remark} 
Similar formulation of the group SLOPE was proposed in \cite{Gossmann2015}. However \cite{Gossmann2015} considers only the case when the weights $w_i$ are equal to the square root of the group size and penalties are imposed directly on $\|\beta_{I_i}\|_2$ rather than on group effects $\|X_{I_i} \beta_{I_i}\|_2$. This makes the method of \cite{Gossmann2015} dependent on scaling or rotations of variables in a given group.  In comparison to \cite{Gossmann2015}, where a Monte Carlo approach for estimating the regularizing sequence was proposed, our article provides choice of the smoothing parameters which provably allow for FDR control in case where the regressors in different groups are orthogonal to each other and its modification, which according to our simulation study allows for FDR control where regressors in different groups are independent. 
\end{remark}  


\subsection{Numerical algorithm}

As shown in Appendix \ref{31071347} the function $J_{\lambda, W, \II}(b):=J_{\lambda}\Big( W\iI{b}\Big)$ is a norm and the optimization problem  (\ref{24111029}) can be solved by using proximal gradient methods.  In our \texttt{R} package \texttt{grpSLOPE} available on CRAN (The Comprehensive R Archive Network) the accelerated proximal gradient method known as FISTA \cite{FISTA} is applied, which uses the specific procedure for choosing steps sizes, to achieve fast convergence rate. The proximal operator for gSLOPE is obtained by appropriate transformation and reduction of the problem, so the fast proximal operator for SLOPE \cite{SLOPE2} can be used. To derive proper stopping criteria, we have considered dual problem to gSLOPE and employed the strong duality property. The detailed description of the proximal operator for gSLOPE as well as of the dual norm and conjugate of grouped sorted $l_1$ norm is provided in the Appendix \ref{31071347}.


\subsection{Group FDR}

Group SLOPE is designed to select groups of variables, which might be very strongly correlated within a group or even linearly dependent. In this context we do not intend to identify single important predictors but rather want to point at the groups which contain at least one true regressor. To theoretically investigate the properties of gSLOPE in this context we now introduce the respective notion of group FDR (gFDR).
\begin{definition}
Consider model (\ref{gmodel}) and let $\beta^\ES{gS}$ be an estimate given by (\ref{gSLOPE}). We define two random variables: the number of all groups selected by gSLOPE (Rg) and the number of groups falsely discovered by gSLOPE (Vg), as
$$Rg:=\big | \big\{i:\ \|X_{I_i}\beta^\ES{gS}_{I_i}\|_2\neq 0\big\}\big |, \qquad Vg:=\big |\big\{i:\ \|X_{I_i}\beta_{I_i}\|_2=0,\ \|X_{I_i}\beta^\ES{gS}_{I_i}\|_2 \neq 0\big\}\big |.$$
\end{definition}
\begin{definition}
We define the false discovery rate for groups (gFDR) as
\begin{equation}
gFDR: = \mathbb{E}\left [ \frac{Vg}{\max\{Rg, 1\}} \right].
\end{equation}
\end{definition}


\subsection{Control of gFDR when variables from different groups are orthogonal}\label{subsec:06232149}

Our goal is the identification of the regularizing sequence for gSLOPE such that gFDR can be controlled at any given level $q\in(0,1)$. In this section we will provide such a sequence, which provably controls gFDR in case when variables in different groups are orthogonal to each other. In subsequent sections we will replace this condition with the weaker assumption of the stochastic independence of regressors in different groups. Before the statement of the main theorem on gFDR control, we will recall the definition of $\chi$ distribution and define a scaled $\chi$ distribution.

\begin{definition}
We will say that a random variable $X_1$ has a $\Chi$ distribution with $l$ degrees of freedom, and write $X_1 \sim \Chi_l$, when $X_1$ could be expressed as $X_1 = \sqrt{X_2},$ for $X_2$ having a $\Chi^2$ distribution with $l$ degrees of freedom.
We will say that a random variable $X_1$ has a scaled $\Chi$ distribution with $l$ degrees of freedom and scale $\mathcal{S}$, when $X_1$ could be expressed as $X_1 = \mathcal{S}\cdot X_2,$ for $X_2$ having a $\Chi$ distribution with $l$ degrees of freedom. We will use the notation $X_1 \sim \mathcal{S}\Chi_l$. 
\end{definition}
\begin{theorem}[gFDR control under orthogonal case]
\label{gFDRcontrol}
Consider model (\ref{gmodel}) with the design matrix $X$ satisfying $X_{I_i}^\mathsf{T}X_{I_j}=0$, for any $i\neq j$. Denote the number of zero coefficients in $\XI{\beta}$ by $m_0$ and let $w_1,\ldots,w_m$ be positive numbers. Moreover, define the sequence of regularizing parameters $\lambda^{max}=(\lambda^{max}_1,\ldots,\lambda^{max}_m)^\mathsf{T}$, with
\begin{equation}\label{lambda_max} 
\lambda^{max}_i:=\max\limits_{j=1,\ldots,m}\left\{\frac1{w_j}F^{-1}_{\Chi_{l_j}}\left (1-\frac{q\cdot i}{m}\right )\right\},
\end{equation}
where $F_{\Chi_{l_j}}$ is a cumulative distribution function of $\Chi$ distribution with $l_j$ degrees of freedom.
Then any solution, $\beta^\ES{gS}$, to problem gSLOPE \eqref{gSLOPE} generates the same vector $\XI{\beta^\ES{gS}}$ and it holds
\begin{equation*}
gFDR =\mathbb{E}\left [ \frac{Vg}{\max\{Rg, 1\}} \right]\leq q \cdot \frac{m_0}{m}.
\end{equation*}
\end{theorem}

\begin{proof}
We will start with the standardized version of the gSLOPE problem, given by \eqref{24111029}. Based on results discussed in Appendix \ref{Sec:alt_rep}, we can consider an equivalent formulation of \eqref{24111029}
\begin{equation}
\label{17022353}
\left\{
\begin{array}{l}
c^*=\argmin{c}\left\{\frac12\sum_{i=1}^m\big(\|\widetilde{y}_{\II_i}\|_2-w^{-1}_ic_i\big)^2+J_{\sigma\lambda}(c)\right\}\\
\|X_{I_i}\beta^\ES{gS}_{I_i}\|_2=c^*_i \big(w_i\|\widetilde{y}_{\II_i}\|_2\big)^{-1}\widetilde{y}_{\II_i} ,\quad i=1,\ldots,m\;\;,
\end{array}
\right.
\end{equation}
where $\y=\X^T y$ has a multivariate normal distribution $\mathcal{N}\big(\widetilde{\beta},\ \sigma^2 \mathbf{I}_{\widetilde{p}}\big)$ with $\widetilde{\beta}_{\II_i}=R_i\beta_{I_i}$. The uniqueness of $\XI{\beta^\ES{gS}}$ follows simply from the uniqueness of $c^*$ in \eqref{17022353}. Define random variables ${R:=\big | \big\{i:\ c_i^*\neq 0\big\}\big |}$ and ${V:=\big |\big\{i:\ \|\widetilde{\beta}_{\II_i}\|_2=0,\quad c^*_i \neq 0\big\}\big |}$. Clearly, then $Rg=R$ and $Vg=V$. Consequently, it is enough to show that
$$\mathbb{E}\left[\frac{V}{\max\{R, 1\}}\right]\leq q \cdot \frac{m_0}{m}.$$
Without loss of generality we can assume that groups $I_1,\dots,I_{m_0}$ are truly irrelevant, which gives $\|\widetilde{\beta}_{\II_1}\|_2=\ldots=\|\widetilde{\beta}_{\II_{m_0}}\|_2=0$ and $\|\widetilde{\beta}_{\II_j}\|_2>0$ for $j>m_0.$ Suppose now that $r,i$ are some fixed indices from $\{1,\ldots,m\}$. From definition of $\lambda^{max}_r$
\begin{equation}
\lambda^{max}_r\geq \frac1{w_i}F^{-1}_{\chi_{l_i}}\left(1-\frac{qr}m\right)\ \Longrightarrow\ 1-F_{\chi_{l_i}}\left(\lambda^{max}_rw_i\right)\leq\frac{qr}m.
\end{equation}
Now, let us fix $i\leq m_0$. Since $\sigma^{-1}\|\widetilde{y}_{\II_i}\|_2\sim\chi_{l_i}$ we have
\begin{equation}
\label{09022303}
\mathbb{P}\left(w_i^{-1}\|\widetilde{y}_{\II_i}\|_2\geq \sigma\lambda^{max}_r\right)=\mathbb{P}\left(\sigma^{-1}\|\widetilde{y}_{\II_i}\|_2\geq \lambda^{max}_rw_i\right)=1-F_{\chi_{l_i}}\left(\lambda^{max}_rw_i\right)\leq \frac{qr}m.
\end{equation}
Now, denote by $\widetilde R^i$ the number of nonzero coefficients in SLOPE estimate \eqref{17022353} after eliminating $i$th group of explanatory variables. Thanks to lemmas \ref{lemma2} and \ref{lemma3}, we immediately get
\begin{equation}
\label{09022302}
\big\{\iI{\widetilde{y}}:\ c^*_i\neq 0\textrm{ and }R=r\big\}\subset\big\{\iI{\widetilde{y}}:\ w_i^{-1}\|\widetilde{y}_{\II_i}\|_2>\sigma\lambda^{max}_r\textrm{ and }\widetilde{R}^i=r-1\big\},
\end{equation}
which together with  (\ref{09022303}) raises
\begin{equation}
\begin{split}
\mathbb{P}(c_i^*\neq0\textrm{ and }R=r)\ & \leq\ \mathbb{P}\left(w_i^{-1}\|\widetilde{y}_{\II_i}\|_2>\sigma\lambda^{max}_r\textrm{ and }\widetilde{R}^i=r-1\right)\\
&=\ \mathbb{P}\left(w_i^{-1}\|\widetilde{y}_{\II_i}\|_2>\sigma\lambda^{max}_r\right) \mathbb{P}\left(\widetilde{R}^i=r-1\right)\\
&\leq\ \frac{qr}m\mathbb{P}\left(\widetilde{R}^i=r-1\right).
\end{split}
\end{equation}
Therefore
\begin{equation}\label{eq:thm_last}
\begin{split}
&\mathbb{E}\left[\frac{V}{\max\{R, 1\}} \right]=\sum_{r=1}^m\mathbb{E}\left[\frac{V}{r}\mathds{1}_{\{R=r\}}\right] = \sum_{r=1}^m\frac1r\mathbb{E}\left[\sum_{i=1}^{m_0}\mathds{1}_{\{c^*_i\neq0\}}\mathds{1}_{\{R=r\}}\right]=\\
&\sum_{r=1}^m\frac1r\sum_{i=1}^{m_0}\mathbb{P}\left(c^*_i\neq0\textrm{ and }R=r\right)\ \leq\ \sum_{i=1}^{m_0}\frac qm\sum_{r=1}^m\mathbb{P}\big(\widetilde{R}^i=r-1\big) = \frac{qm_0}m,
\end{split}
\end{equation}
which finishes the proof.
\end{proof}

Figure \ref{19091846} illustrates the performance of gSLOPE under the design matrix $X=\mathbf{I}_p$ (hence the rank of $i$ group, $l_i$, coincides with its size), with $p=5000$. In  Figure \ref{19091846} (a) all groups are of the same size $l=5$, while in Figures \ref{19091846} (b) - (d) the explanatory variables are clustered into $m=1000$ groups  of sizes from the set $\{3, 4, 5, 6, 7\}$; $200$ groups of each size. Each coefficient of $\beta_{I_i}$, in truly relevant group $i$, was generated independently from $U[0.1, 1.1]$ distribution and then $\beta_{I_i}$ was scaled such that $\big(\XI{\beta}\big)_i=a\sqrt{l_i}$. Parameter $a$ was selected to satisfy the condition $\frac1m\sum_{i=1}^ma\sqrt{l_i} = \frac1m\sum_{i=1}^mB(m,l_i)$, where $B(m_i,l)$ is defined in (\ref{05031641}). Such signals are comparable to the maximal noise and can be detected with moderate power, which allows for a meaningful comparison between different methods. 

\begin{figure}[h!]
\includegraphics[width=\textwidth]{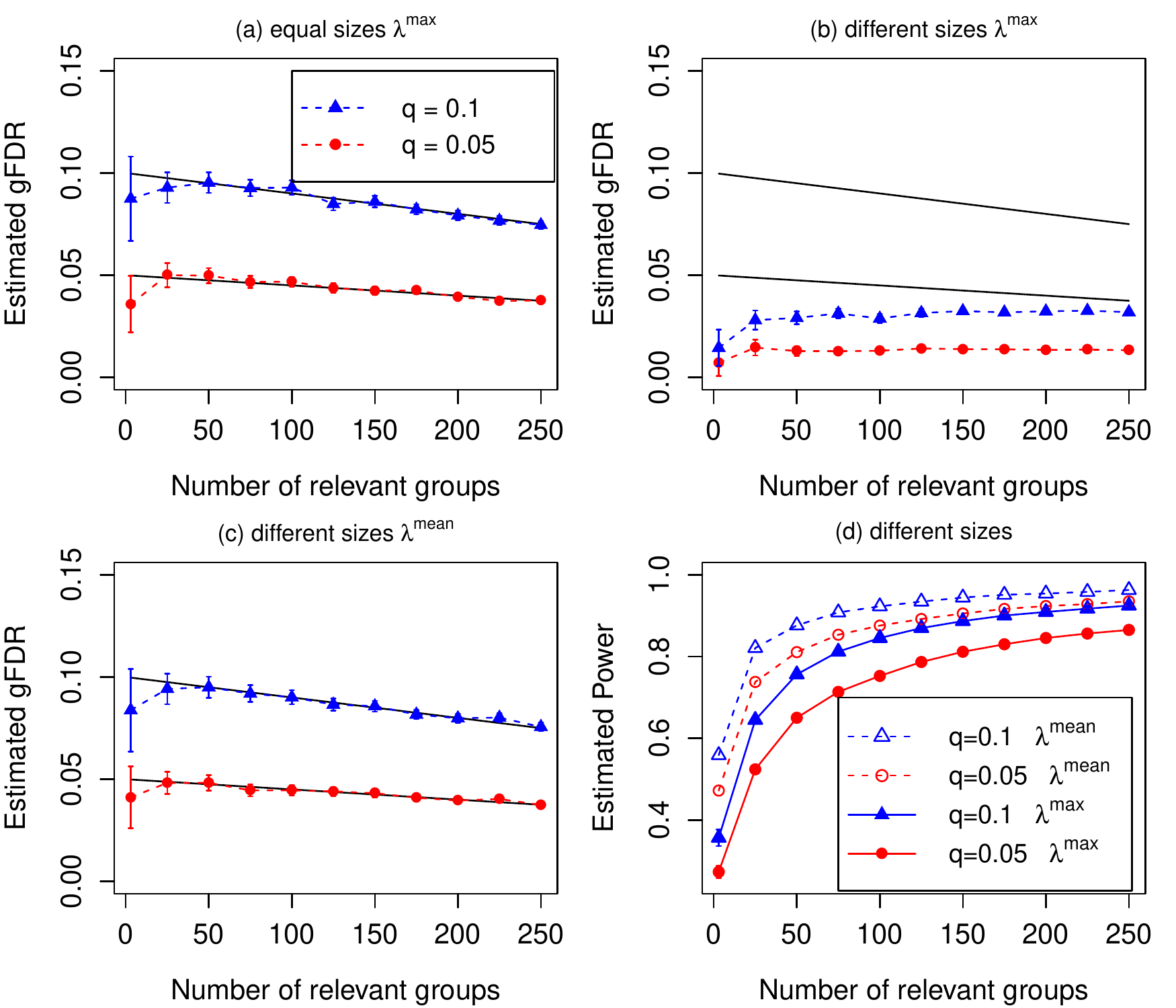}
 \caption{Orthogonal situation with $n=p=5000$ and $m=1000$. In (a) all groups are of the same size $l=5$, while in (b)-(d) there are 200 groups of each of sizes  $l_i\in \{3,4,5,6,7\}$.  In (a) and (b) gSLOPE works with the regularizing sequence $\lambda^{max}$, while in (c) and (d) $\lambda^{mean}$ is used. For each target gFDR level and true support size, $300$ iterations were performed. Bars correspond to $\pm 2$SE. Black straight lines represent the "nominal" gFDR level $q\cdot\big((m-k)/m\big)$, for $k$ being true support size. Weights are defined as $w_i:=\sqrt{l_i}$.
 }
\label{19091846}
\end{figure}

Figure \ref{19091846} (a) illustrates that the sequence $\lambda^{max}$ allows to keep gFDR very close to the "nominal" level when groups are of the same size. However, Figure \ref{19091846} (b) shows that for groups of different size $\lambda^{max}$ is rather conservative, i.e. the achieved gFDR is significantly lower than assumed. This suggests that the shrinkage (dictated by $\lambda$) could be slightly decreased, such that the method gets more power and still achieves the gFDR below the assumed level. Returning to the proof of Theorem \ref{gFDRcontrol}, we can see that for each $i\in\{1,\ldots,m\}$ we have 
\begin{equation}\label{con_max}
1-F_{\Chi_{l_i}}\left(\lambda^{max}_rw_i\right)\leq\frac{qr}m\;\;,
\end{equation}
with equality holding only for $i$ being the index of the maximum in (\ref{lambda_max}). In the result the inequality in (\ref{eq:thm_last}) is usually strict and the true gFDR might be substantially smaller than the nominal level. 
The natural  relaxation of (\ref{con_max}) is to require only that
\begin{equation}
\label{08121449}
\sum_{i=1}^m\Big(1-F_{w^{-1}_i\Chi_{l_i}}(\lambda_r)\Big)\leq qr.
\end{equation}
Replacing the inequality in (\ref{08121449}) by equality yields the strategy of choosing the relaxed $\lambda$ sequence  
\begin{equation}
\label{08191453}
\lambda^{mean}_r:= \overline{F}^{-1}\left(1-\frac{qr}m\right)\quad\textrm{for}\quad \overline{F}(x): = \frac1m\sum_{i=1}^mF_{w^{-1}_i\Chi_{l_i}}(x),\quad r\in\{1,\ldots,m\},
\end{equation}
where $F_{w^{-1}_i\Chi_{l_i}}$ is the cumulative distribution function of scaled chi distribution with $l_i$ degrees of freedom and scale $\mathcal{S}=w_i^{-1}$. In Figure \ref{19091846}c we present estimated gFDR, for tuning parameters given by (\ref{08191453}).  The results suggest that with relaxed version of tuning parameters, we can still achieve the "average" gFDR control, where the "average" is with respect to the uniform distribution over all possible signal placements. As shown in Figure \ref{19091846}d, application of $\lambda^{mean}$ allows to achieve a substantially larger power than the one provided by $\lambda^{max}$. Such a strategy could be especially important in situation, when differences between the smallest and the largest quantiles (among distributions $w_i^{-1}\Chi_{l_i}$) are relatively large and all groups have the same prior probability of being relevant. 

\subsection{The accuracy of estimation}

Up until this point, we have only considered the testing properties of gSLOPE. Though originally proposed to control the FDR, surprisingly, SLOPE enjoys appealing estimation properties as well \cite{WE}. It thus would be desirable to extend this link between testing and estimation for gSLOPE. In measuring the deviation of an estimator from the ground truth $\beta$, as earlier, we focus on the group level instead of an individual. Accordingly, here we aim to estimate parts of variance of $Y$ explained by every group, which are contained in the vector $\llbracket \beta\rrbracket_{X,I}: = \big(\|X_{I_1} \beta_{I_1}\|_2, \ldots, \|X_{I_m} \beta_{I_m}\|_2\big)^\mathsf{T} $ or $\iI{\widetilde{\beta}}: = \big(\|\widetilde{\beta}_{\II_1}\|_2, \ldots, \|\widetilde{\beta}_{\II_m}\|_2\big)^\mathsf{T} $, equivalently. For illustration purpose, we employ the setting described as follows. Imagine that we have a sequence of problems with the number of groups $m$ growing to infinity: the design $X$ is orthonormal at groups level; ranks of submatrices $X_{I_i}$, $l_i$, are bounded, that is, $\max l_i \le l$ for some constant integer $l$; denoting by $k \ge 1$ the sparsity level (that is, the number of relevant groups), we assume the asymptotics $k/m \rightarrow 0$. Now we state our minimax theorem, where we write $a \sim b$ if $a/b \rightarrow 1$ in the asymptotic limit, and $\|\XI{\beta}\|_0$ denotes the number of nonzero entries of $\XI{\beta}$. The proof makes use of the same techniques for proving Theorem 1.1 in \cite{WE} and is deferred to the Appendix.
\begin{theorem}\label{minimax}
Fix any constant $q \in (0, 1)$, let $w_i = 1$ and $\lambda_i = F_{\chi_l}^{-1}(1 - qi/m)$ for $i=1,\ldots,m$. Under the preceding conditions, gSLOPE is asymptotically minimax over the nearly black object $\big\{\beta: \big\|\XI{\beta}\big\|_0 \le k\big\}$, i.e.,
\[
\sup_{\|\XI{\beta}\|_0 \le k} \mathbb{E} \left( \Big\| \XI{\beta^\ES{gS}} - \XI{\beta} \Big\|_2^2\right) \sim  \inf_{\widehat\beta}\sup_{\|\XI{\beta}\|_0 \le k} \mathbb{E} \left( \Big\| \XI{\widehat\beta} - \XI{\beta} \Big\|_2^2\right),
\] 
where the infimum is taken over all measurable estimators $\widehat\beta(y, X)$.
\end{theorem}
Notably, in this theorem the choice of $\lambda_i$ does not assume the knowledge of sparsity level. Or putting it differently, in stark contrast to gLASSO, gSLOPE is adaptive to a range of sparsity in achieving the exact minimaxity. Combining Theorems~\ref{gFDRcontrol} and \ref{minimax}, we see the remarkable link between FDR control and minimax estimation also applies to gSLOPE \cite{ABDJ, WE}. While it is out of the scope of this paper, it is of great interest to extend this minimax result to general design matrices.

\subsection{The impact of chosen weights}
In this subsection we will discuss the influence of chosen weights, $\{w_i\}_{i=1}^m$, on results. Let $I=\{I_1,\ldots,I_m\}$ be a given partition into groups and $l_1,\ldots,l_m$ be ranks of submatrices $X_{I_i}$. Assume the orthogonality at group level, i.e., that it holds $X_{I_i}^\mathsf{T}X_{I_j}=0$, for $i\neq j$, and suppose that $\sigma=1$. The support of $\XI{\beta}$ coincides with the support of vector $c^*$ defined in \eqref{17022353}, namely 
\begin{equation}
c^* = \argmin{c}\ \frac12\Big\|\iI{\y}-W^{-1}c\Big\|_2^2+J_{\lambda}(c),
\end{equation} 
where $W^{-1}$ is a diagonal matrix with positive numbers $w_1^{-1},\ldots,w_m^{-1}$ on the diagonal. Suppose now, that $c^*$ has exactly $r$ nonzero coefficients. From Corollary \ref{06021230}, these indices are given by $\{\pi(1),\ldots,\pi(r)\}$, where $\pi$ is permutation which orders $W^{-1}\iI{\y}$. Hence, the order of realizations $\big\{w_i^{-1}\|\widetilde{y}_{\II_i}\|_2\big\}_{i=1}^m$ decides about the subset of groups labeled by gSLOPE as relevant. Suppose that groups $I_i$ and $I_j$ are truly relevant, i.e., $\|\Beta_{\II_i}\|_2>0$ and $\|\Beta_{\II_j}\|_2>0$. The distributions of $\|\y_{\II_i}\|_2$ and $\|\y_{\II_j}\|_2$ are noncentral $\Chi$ distributions, with $l_i$ and $l_j$ degrees of freedom, and the noncentrality parameters equal to $\|\Beta_{\II_i}\|_2$ and $\|\Beta_{\II_j}\|_2$, respectively. Now, the expected value of the noncentral $\Chi$ distribution could be well approximated by the square root of the expected value of the noncentral $\Chi^2$ distribution, which gives
$$\mathbb{E}\big(w_i^{-1}\|\y_{\II_i}\|_2\big)\approx w_i^{-1}\sqrt{\mathbb{E}\big(\|\y_{\II_i}\|^2_2\big)}=w_i^{-1}\sqrt{l_i+\|\Beta_{\II_i}\|_2^2}.$$
Therefore, roughly speaking, truly relevant groups $I_i$ and $I_j$ are treated as comparable, when it occurs ${l_i/w_i^2+\|\Beta_{\II_i}\|_2^2/w_i^2\approx l_j/w_j^2+\|\Beta_{\II_j}\|_2^2/w_j^2}$. This gives us the intuition about the behavior of gSLOPE with the choice $w_i=\sqrt{l_i}$ for each $i$. Firstly, gSLOPE treats all irrelevant groups as comparable, i.e. the size of the group has a relatively small influence on it being selected as a false discovery. Secondly,  gSLOPE treats two truly relevant groups as comparable, if groups effect sizes satisfy the condition $\big(\XI{\beta}\big)_i/\big(\XI{\beta}\big)_j\approx\sqrt{l_i}/\sqrt{l_j}$. The derived condition could be recast as $\|X_{I_i}\beta_{I_i}\|_2^2/l_i \approx \|X_{I_j}\beta_{I_j}\|_2^2/l_j$. This gives a nice interpretation: with the choice $w_i:=\sqrt{l_i}$, gSLOPE treats two groups as comparable, when these groups have similar squared effect group sizes per coefficient. One possible idealistic situation, when such a property occurs, is when all $\beta_i$'s in truly relevant groups are comparable.

In Figure \ref{21090909} we see that when the condition $\big(\XI{\beta}\big)_i/\big(\XI{\beta}\big)_j=\sqrt{l_i}/\sqrt{l_j}$ is met, the fractions of groups with different sizes in the selected truly relevant groups (STRG) are approximately equal. To investigate the impact of selected weights on the set of discovered groups, we performed simulations with different settings, namely we used $w_i=1$ and $w_i = l_i$ (without changing other parameters). With the first choice, larger groups are penalized less than before, while the second choice yields the opposite situation. This is reflected in the proportion of each groups in STRG (Figure \ref{21090909}). 
\begin{figure}[h!]
\centering
\begin{subfigure}{.33\textwidth}
  \centering
	\includegraphics[width=1\linewidth]{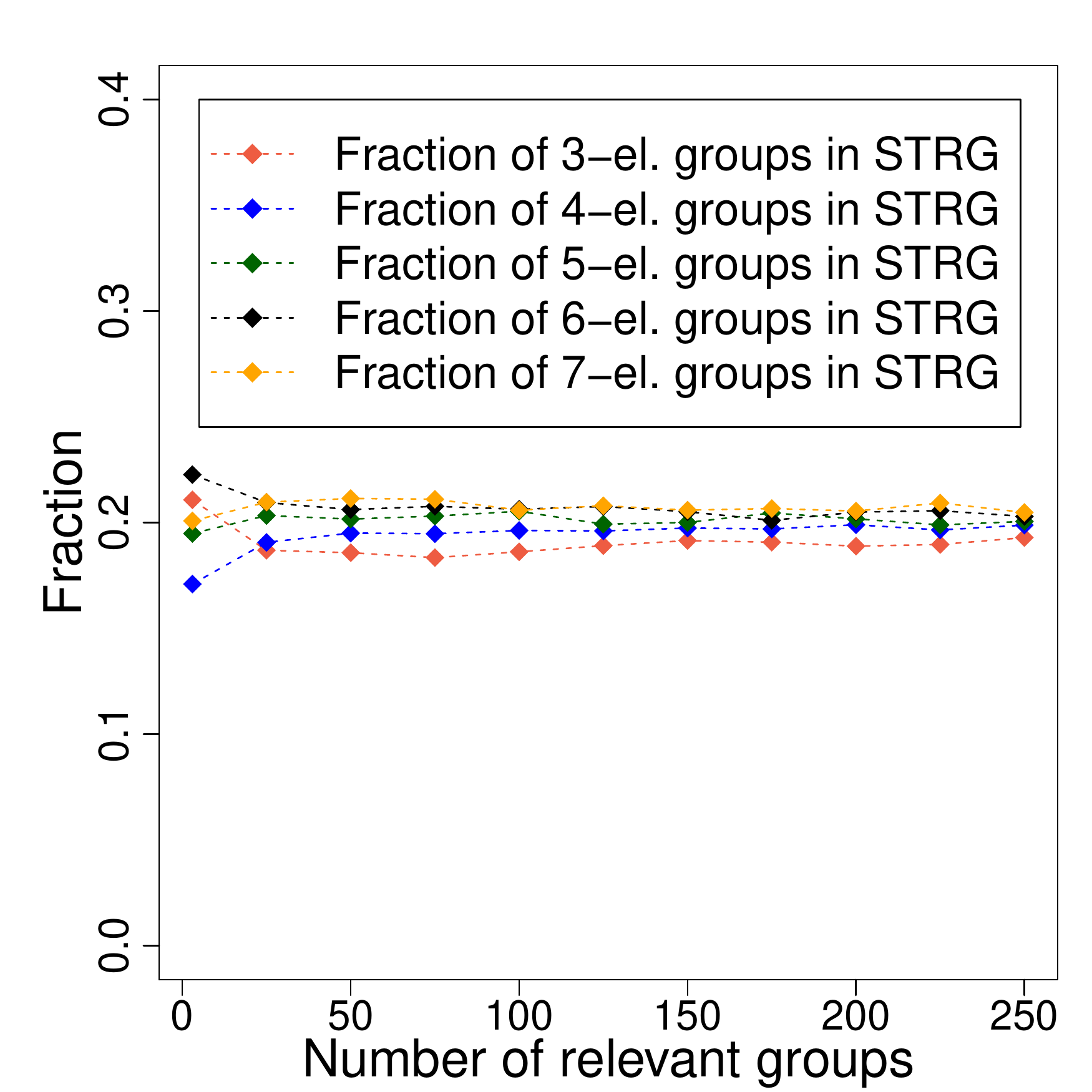}
  \caption{Structure of STRG, $w_i:=\sqrt{l_i}$}
\end{subfigure}%
\begin{subfigure}{.33\textwidth}
  \centering
  \includegraphics[width=1\linewidth]{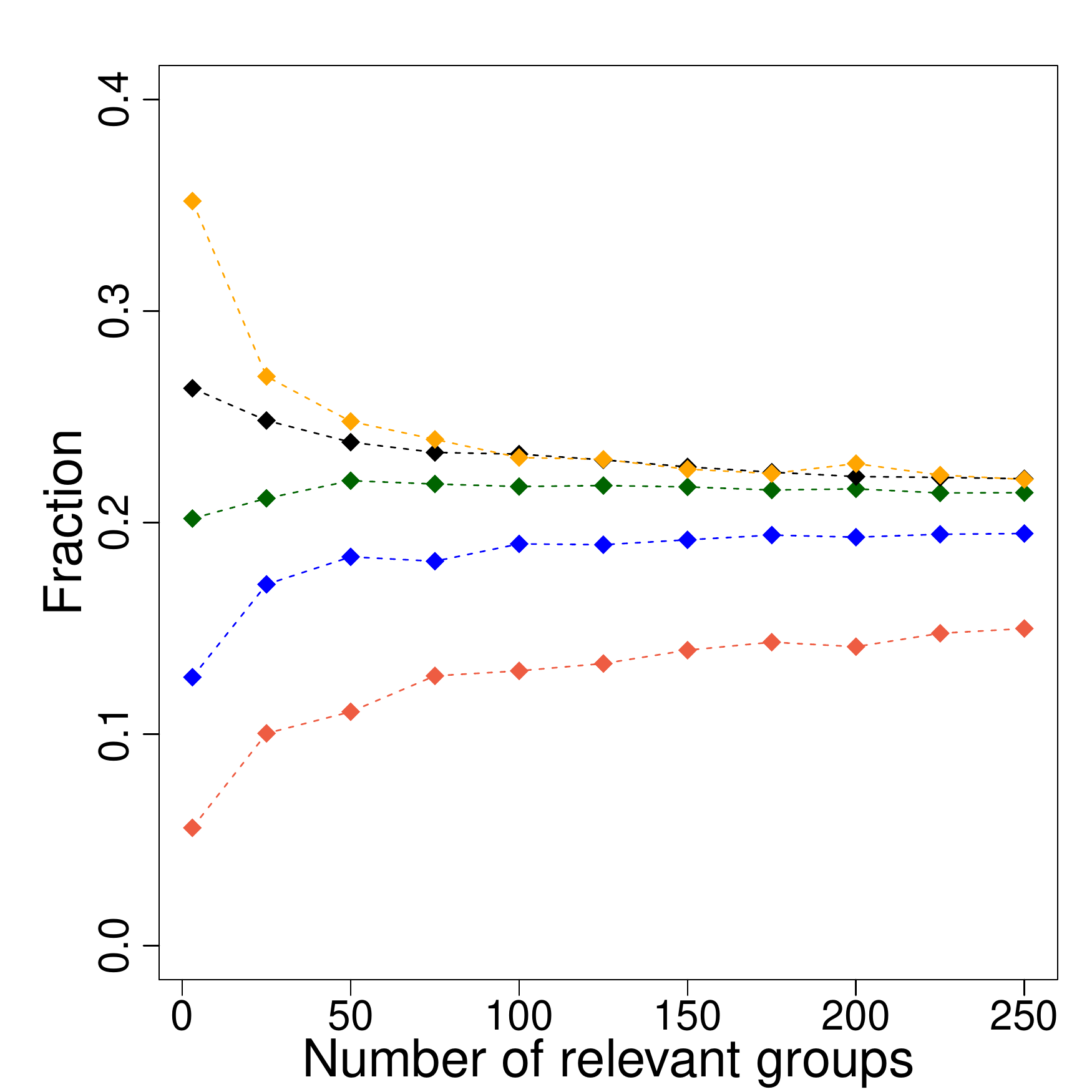}
  \caption{Structure of STRG, $w_i:=1$}
\end{subfigure}%
\begin{subfigure}{.33\textwidth}
  \centering
  \includegraphics[width=1\linewidth]{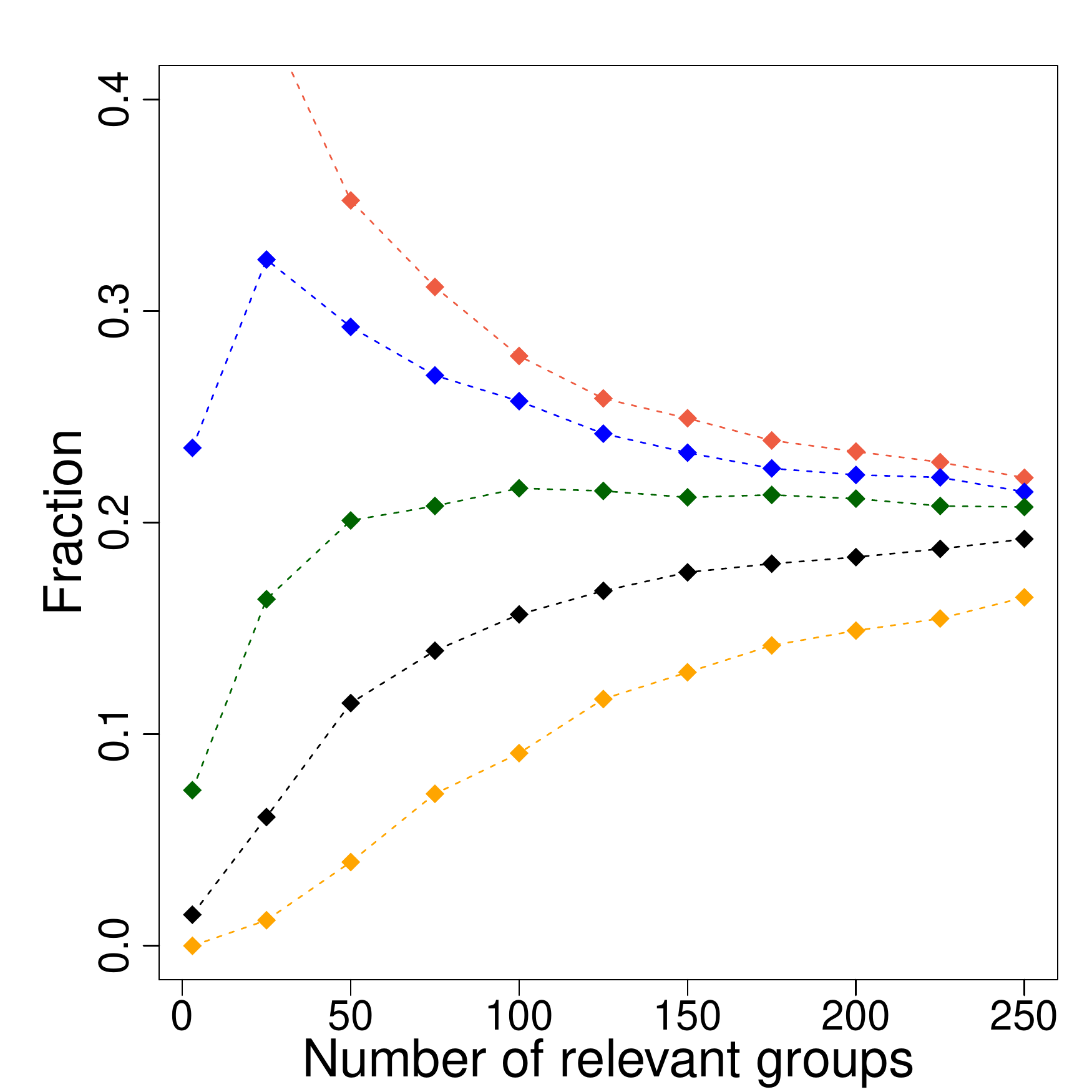}
  \caption{Structure of STRG, $w_i:=l_i$}
\end{subfigure}%
\caption{Fraction of each group sizes in selected truly relevant groups (STRG). Beyond the weights, this simulation was conducted with the same setting as in experiments summarized in Figure 1 for $\lambda^{mean}$. In particular, for truly relevant groups $i$ and $j$, it occurs $\big(\XI{\beta}\big)_i/\big(\XI{\beta}\big)_j=\sqrt{l_i}/\sqrt{l_j}$. Target gFDR level was fixed as $0.05$.}
\label{21090909}
\end{figure}
The values of gFDR are very similar under all choices of weights. 

\subsection{Independent groups and unknown $\sigma$}

The assumption that variables in different groups are orthogonal to each other can be satisfied only in rare situations of specifically designed experiments. However, in a variety of applications one can assume that variables in different groups are independent. Such a situation occurs for example in the context of identifying influential genes using distant genetic markers, whose genotypes can be considered as stochastically independent. In this case a group can be formed by clustering  dummy variables corresponding to different genotypes of a given marker.  
Though the difference between stochastic independence and algebraic orthogonality seems rather small, it turns out that small sample correlations between independent regressors together with the shrinkage of regression coefficients lead to magnifying the effective noise and require the adjustment of the tuning sequence $\lambda$ (see \cite{FDR_LASSO} for discussion of this phenomenon in the context of LASSO).  Concerning regular SLOPE, this problem was addressed by heuristic modification of $\lambda$, proposed in \cite{SLOPE} and \cite{SLOPE2}. This modified sequence was calculated based upon the assumption that explanatory variables are randomly sampled from the Gaussian distribution. However, simulation results from \cite{SLOPE2} illustrate that it controls FDR also in case when the columns of the design matrix correspond to additive effects of independent SNPs and the number of causal genes is moderately small.  

Following  ideas for regular SLOPE presented in \cite{SLOPE2}, we propose the Procedure \ref{11091537} for calculating the sequence of tuning parameters in case when variables in different groups are independent. The heuristics justifying this choice are substantially more technically involved than the heuristics for regular SLOPE and their details are presented in the Appendix \ref{subsec:06232149}. Procedure \ref{11091537} is based on the sequence of $\lambda^{mean}$ but the version for the conservative choice $\lambda^{max}$ follows analogously. The proposed sequence of tuning parameters flattens out for a certain value $k^{\star}$ dependent on $q, n$ and $l_1,\ldots, l_m$. It  is supposed to control gFDR when the number of identified groups is not much larger than $k^{\star}$.

\begin{algorithm}
{\fontsize{9pt}{5pt}\selectfont
  \caption{Sequence of tuning parameters for independent groups}
	\label{11091537}
  \begin{algorithmic}    
		\State \textbf{input:} $q\in (0,1)$,\ \ $w_1,\ldots,w_m>0$,\ \ $p,\ n,\ m,\ l_1,\ldots, l_m \in\mathbb{N}$
		\State $\lambda_i:=\overline{F}^{-1}\left(1-\frac qm\right),\quad$ for $\quad\overline{F}(x): = \frac1m\sum_{i=1}^mF_{w^{-1}_i\chi_{l_i}}(x)$;
		\State  \textbf{for} $i\in\{2,\ldots,m\}$:
		\State \indent $\lambda^S: = (\lambda_1,\ldots, \lambda_{i-1})^\mathsf{T}$;
		\State \indent $\mathcal{S}_j: = \sqrt{\frac{n-l_j(i-1)}n+\frac{w_j^2\|\lambda^S\|^2_2}{n-l_j(i-1)-1}},\qquad$ for $j\in\{1,\ldots,m\}$;
		\State \indent $\lambda^*_i:=\overline{F}^{-1}_{\mathcal{S}}\left(1-\frac {qi}m\right),\quad$ for $\quad\overline{F}_{\mathcal{S}}(x): =  \frac1m\sum_{j=1}^m F_{\mathcal{S}_jw_j^{-1}\chi_{l_j}}(x)$;
		\State \indent if $\lambda^*_i\leq \lambda_{i-1}$, then put $\lambda_i:=\lambda^*_i$. Otherwise, stop the procedure and put $\lambda_j:=\lambda_{i-1}$ for $j\geq i$;
		\State  \textbf{end for}
  \end{algorithmic}
	}
\end{algorithm}

Up until this moment, we have used $\sigma$ in gSLOPE optimization problem, assuming that this parameter is known . However, in many applications $\sigma$ is unknown and its estimation is an important issue. When $n>p$, the standard procedure is to use the unbiased estimator of $\sigma^2$, $\hat{\sigma}_{\textnormal{\tiny OLS}}^2$, given by
\begin{equation}
\label{11091832}
\hat{\sigma}_{\textnormal{\tiny OLS}}^2:=\big(y-X\beta^{\textnormal{\tiny OLS}}\big)^\mathsf{T}\big(y-X\beta^{\textnormal{\tiny OLS}}\big)/(n-p),\textrm{ for }\beta^{\textnormal{\tiny OLS}}:=(X^\mathsf{T}X)^{-1}X^\mathsf{T}y.
\end{equation}
For the target situation, with $p$ much larger than $n$, such an estimator can not be used. To estimate $\sigma$ we will therefore apply the procedure which was dedicated for this purpose in \cite{SLOPE2} in the context of SLOPE. Below we present algorithm adjusted to gSLOPE (Procedure \ref{11091824}).
\begin{algorithm}
{\fontsize{9pt}{5pt}\selectfont
  \caption{gSLOPE with estimation of $\sigma$}
	\label{11091824}
  \begin{algorithmic}
		\State \textbf{input:} $y$,\ $X$ and $\lambda$ (defined for some fixed $q$)
		\State  \textbf{initialize:} $S_+=\emptyset$;
		\State  \textbf{repeat}
		\State  \indent $S=S_+$;
		\State  \indent \textrm{compute RSS obtained by regressing }$y${ onto variables in }$S$;
		\State  \indent set $\hat{\sigma}^2=RSS/(n-|S|-1)$;
		\State  \indent compute the solution $\beta^{\textnormal{\tiny gS}}$ to gSLOPE with parameters $\hat{\sigma}$ and sequence $\lambda$;
		\State  \indent set $S_+=\operatorname{supp}(\beta^{\textnormal{\tiny gS}})$;
		\State  \textbf{until} $S_+=S$
  \end{algorithmic}
	}
\end{algorithm}
The idea standing behind the procedure is simple. The gSLOPE property of producing sparse estimators is used, and in each iteration columns in design matrix are first restricted to support of $\beta^{\textnormal{\tiny gS}}$, so that the number of rows exceeds the number of columns and (\ref{11091832}) can be used. Algorithm terminates when gSLOPE finds the same subset of relevant variables as in the preceding iteration.

To investigate the performance of gSLOPE under the Gaussian design and various group sizes, we performed simulations with $1000$ groups. Their sizes were drawn from the binomial distribution, $Bin(1000; 0.008)$, so as the expected value of the group size was equal to $8$ (Figure \ref{21091752}c). As a result, we obtained $7917$ variables, divided into $1000$ groups (the same division was used in all iterations and scenarios). For each sparsity level and the gFDR level $0.1$, and each iteration we generated entries of the design matrix using $\mathcal{N}\big(0,\frac1n\big)$ distribution, then $X$ was standardized and the values of response variable were generated according to model (\ref{gmodel}) with $\sigma=1$ and signals generated as in simulations for Figure \ref{19091846}. To identify relevant groups based on the simulated data we have used the iterative version of gSLOPE, with $\sigma$ estimation (Procedure \ref{11091824}) and lambdas given by Procedure \ref{11091537}. We performed $200$ repetitions for each scenario, $n$ was fixed as $5000$. Results are represented in Figure \ref{21091752} and show that our procedure allows to control gFDR at the assumed level.

Additionally, Figure \ref{21091752} compares gSLOPE to gLASSO with two  choices of the smoothing parameter $\lambda$. Firstly, we used $\lambda=\lambda^{mean}_1$, which allows to control FDR under the total null hypothesis. Secondly, for each of the iterations we chose $\lambda$ based on leave-one-out cross-validation.   It turns out that the first of these choices becomes rather conservative when the number of truly relevant groups increases. Then gLASSO has a smaller FDR but also a much smaller power than SLOPE (by a factor of three for $k=60$). Cross-validation works in the opposite way - it yields a large power but also results in a huge proportion of false discoveries, which in our simulations systematically exceeds 60\%.   

\begin{figure}[h!]
\centering
\includegraphics[width=1\linewidth]{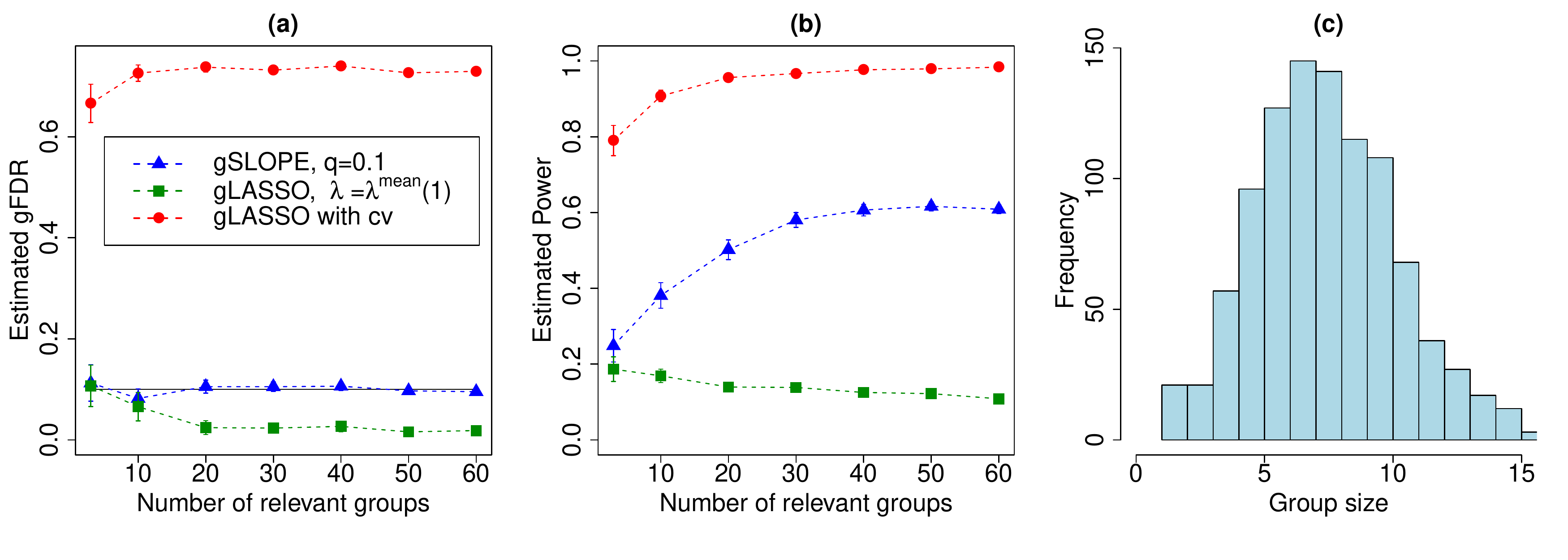}
 \caption{Independent regressors and various group sizes: $m=1000$, $p=7917$ and $n=5000$. Bars correspond to $\pm 2$SE. Entries of design matrix were drawn from $\mathcal{N}(0,1/n)$ distribution and truly relevant signal, $i$, was generated such as $\|X_{I_i}\beta_{I_i}\|_2=\frac1m\sum_{i=1}^mB(m,l_i)$, where $B(m,l)$ is defined in \eqref{05031641}.}
\label{21091752}
\end{figure}


\subsection{Simulations in the context of Genome-Wide Association Studies}\label{SNP_simul}

To test the performance of gSLOPE in the context of Genome-Wide Association Studies (GWAS) we have used the North Finland Birth Cohort (NFBC) dataset, available in dbGaP with accession
number phs000276.v2.p1 (\url {http://www.ncbi.nlm.nih.gov/projects/gap/cgi-bin/study.cgi?study_id=phs000276.v2.p1}) and described in detail in \cite{Sabatti}. The raw data contains $364\;590$ markers for $5\;402$ subjects. To obtain roughly independent SNPs this data set was initially screened such that in the final data set the maximal correlation between any pair of SNPs  does not exceed $\sqrt{0.1}=0.316$. The reduced data set contains $p=26\;233$ SNPs.

The explanatory variables for our genetic model were defined in Table \ref{Tab:coding}, where ${\it a}$ denotes the less frequent (variant) allele.
\begin{table}
\caption{Coding for explanatory variables}
\centering
\begin{tabular}{|c|c|c|}
\hline
genotype&additive dummy variable $\widetilde X$&dominance dummy variable $\widetilde Z$\\
\hline
{\it aa}&2&0\\
{\it aA}&1&1\\
{\it AA}&0&0\\
\hline
\end{tabular}
\label{Tab:coding}
\end{table}
In case when  population frequencies of both alleles are the same, variables $\widetilde X$ and $\widetilde Z$ are uncorrelated. In other cases correlations between these variables is different from zero and can be very strong for rare genetic variants. Since each SNP is described by two dummy variables, the full design matrix $[\widetilde X\ \widetilde Z]$  contains $52\; 466$ potential regressors.
This matrix  was then centered and standardized, so the columns of the final design matrix $[X\ Z]$ have zero mean and unit norm. 
		
The trait values are simulated according to two scenarios. In Scenario 1 we simulate from an additive model, where each of the causal SNPs influences the trait only through the additive dummy variable in matrix $X$,
\begin{equation}\label{eq:additive}
y=X\beta_X+\epsilon\;\;.
\end{equation}
Here $\epsilon \sim \mathcal{N}(0, \mathbf{I})$, the number of `causal' SNPs $k$ varies
between $1$ and $80$ and each causal SNP has an additive effect (non-zero components of $\beta_X$) equal to $5$ or $-5$, with $P(\beta_{Xi}=5)=P(\beta_{Xi}=-5)=0.5$. In each of $100$ iterations of our experiment causal SNPs were randomly selected from the full set of $26\; 233$ SNPs. 

The additive model (\ref{eq:additive}) assumes that for each of the SNPs the expected value of the trait for the heterozygote ${\it aA}$ is the average of expected trait values for both homozygotes ${\it aa}$ and ${\it AA}$. This idealistic assumption is usually not satisfied and many of the SNPs exhibit some dominance effects. To illustrate the performance of gSLOPE in the presence of dominance effects we simulated data according to Scenario $2$; 
\begin{equation}\label{eq:dominance}
y=[X\ Z]\left [\begin{BMAT}(c)[0.5pt,0pt,0.7cm]{c}{cc}\beta_X\\ \beta_Z \end{BMAT} \right ]+\epsilon\;\;
\end{equation}
which differs from Scenario 1 by adding dominance effects (non-zero components of $\beta_Z$), which for each of $k$ selected SNPs are sampled from the uniform distribution on$[-5,-3] \cup [3,5]$. The simulated data sets were analyzed using three different approaches:
\begin{itemize}
\setlength\itemsep{0em}
	\item gSLOPE with $p=26\;233$ groups, where each of the groups contains two explanatory variables, describing the additive and the dominance effect of the same SNP,  
	\item SLOPE$_X$, where the regular SLOPE is used to search through the reduced design matrix $X$ (as in \cite{SLOPE2} or  \cite{geneSLOPE}),
	\item SLOPE$_{XZ}$, where the regular SLOPE is used to search through the full design matrix $[X\ Z]$.
\end{itemize}

In all versions of SLOPE we used the iterative procedure for estimation of $\sigma$ and the sequence $\lambda$ heuristically adjusted to the case of the Gaussian design matrix, as implemented in the CRAN packages \texttt{SLOPE} and \texttt{grpSLOPE}.

Figure \ref{SNP_sim} provides the summary of this simulation study. Here FDR and power are calculated at the SNP level. Specifically, in case of SLOPE$_{XZ}$ the SNP is  counted as a one discovery if the corresponding additive or the dominance dummy variable is selected.

\begin{figure}[h!]
\includegraphics[width=\textwidth]{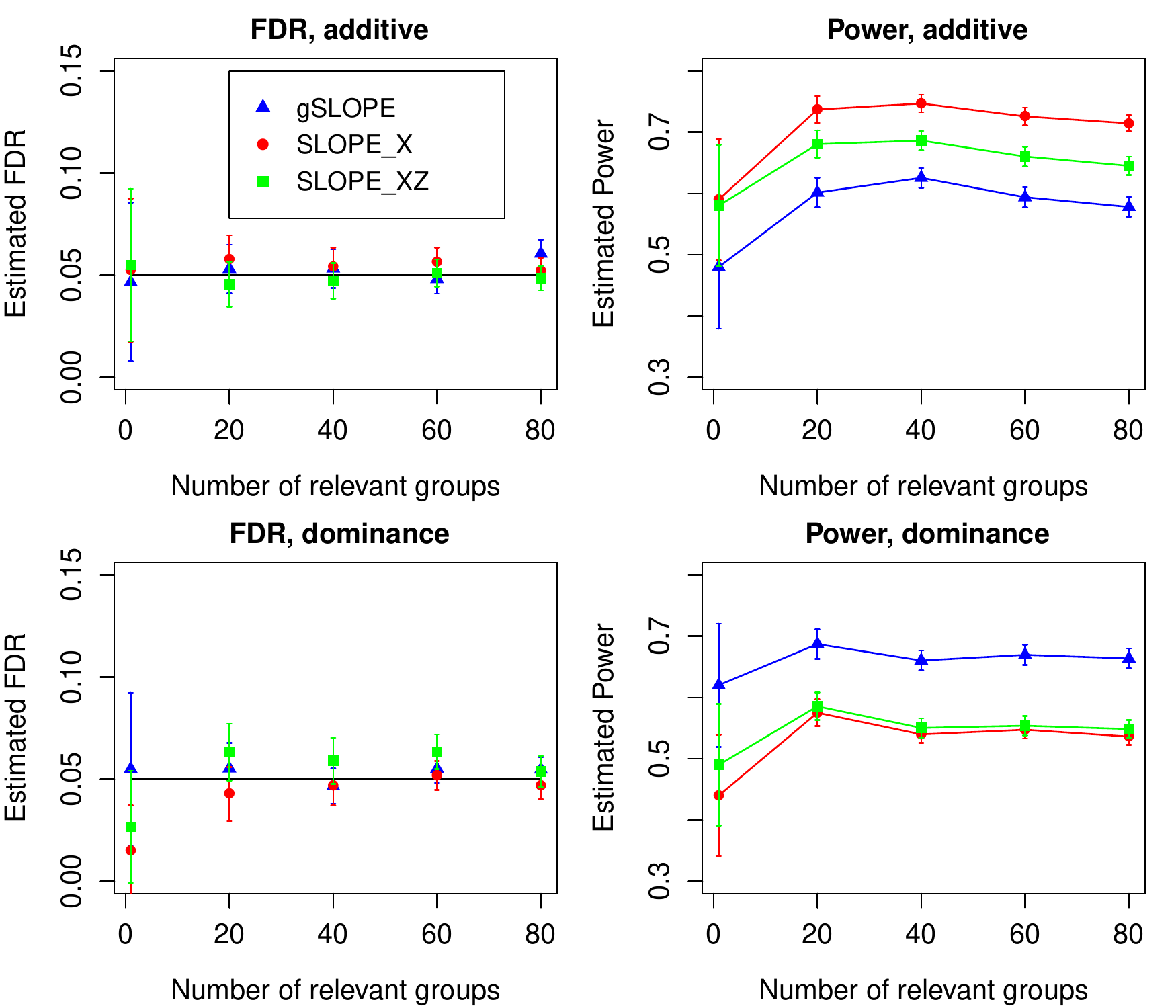}
 \caption{Simulations using real SNP genotypes: $n=5\;402$, $p=26\;233$. Power and gFDR  are estimated based on $100$ iterations of each simulation scenario. Upper panel illustrates the situation where all causal SNPs have only additive effects, while in lower panel each causal SNP has also some dominance effect. }
\label{SNP_sim}
\end{figure}

As shown in Figure \ref{SNP_sim}, for both of the simulated scenarios all versions of SLOPE control gFDR for all considered values of $k$. 
When the data are simulated according to the additive model the highest power is offered by SLOPE$_X$, with the power of gSLOPE being smaller by approximately $13\%$ over the whole range of $k$. However, in the presence of large dominance effects the situation is reversed and gSLOPE offers the highest power, which systematically exceeds the power of SLOPE$_X$ by the symmetric amount of $13\%$. In our simulations SLOPE$_{XZ}$ has intermediate performance and does not substantially improve the power of SLOPE$_X$ in the presence of dominance effects.  
Thus our simulations suggest that gSLOPE provides an information complementary to SLOPE$_X$ and might be a useful tool in the context of Genome-Wise Association Studies.


\subsection{gSLOPE under GWAS application: real phenotype data}\label{GWAS_realdata}

Finally, we have applied group SLOPE to identify SNPs associated with four lipid phenotypes available in NFBC dataset: high-density lipoproteins (HDL), low-density lipoproteins (LDL), triglycerides (TG), and total cholesterol (CHOL). The data set contains  genotypes of $364\,590$ SNPs for $5\,402$ individuals and was previously analyzed in \cite{geneSLOPE} using regular SLOPE to search for additive SNP effects. Before this analysis the data were reduced by applying the p-value threshold for single marker t-tests and  selecting representatives of strongly correlated SNPs, also based on p-values.  Since this pre-processing selects most promising SNPs by performing multiple testing on the full set of $p=364590$ SNPs, the sequence of the tuning parameters for SLOPE needs to be adjusted to this value of $p$ rather than to the number of selected representatives. The algorithm for this analysis is implemented in \texttt{R} package \texttt{geneSLOPE} and its details are explained in \cite{geneSLOPE}.  According to extensive simulation study and real data analysis reported in \cite{geneSLOPE}, \texttt{geneSLOPE} allows to control FDR for the analysis with full size GWAS data. 

In our data analysis we used three versions of SLOPE: geneSLOPE for additive effects (as in \cite{geneSLOPE}),  geneSLOPE$_{XZ}$, with  the design matrix extended by inclusion of dominance dummy variables, and  gene group SLOPE (geneGSLOPE). In geneSLOPE$_{XZ}$ and geneGSLOPE representative SNPs were selected  based on the one way ANOVA tests. For all these procedures the pre-processing was based on p-value threshold $p<0.05$ and the correlation cutoff $\rho<0.3$, which allowed to reduce the data set to roughly 8500 of interesting representative SNPs. For the convenience of the reader, the Procedure \ref{09061754} for the full geneGSLOPE analysis is provided below.

\begin{algorithm}
\caption{geneGSLOPE procedure}
\label{09061754}
  \begin{algorithmic}
	\State \textbf{Input}: $r \in(0,1)$, $\pi\in(0,1]$
	\vspace{3pt}
	\State   \textbf{Screen SNPs}:
		\State (1) For each SNP calculate independently the $p$-value for the ANOVA test with the null hypothesis, $H_0:\mu_{aa}=\mu_{aA}=\mu_{AA}$.
		\State (2)  Define the set ${\cal B}$ of indices corresponding to SNPs whose $p$-values are smaller than $\pi$.
		\vspace{3pt}
		\State \textbf{Cluster SNPs}:
		\State (3) Select the SNP $j$ in ${\cal B}$ with the smallest $p$-value and find all SNPs whose Pearson correlation with this selected SNP is larger than or equal to $r$.
		\State (4) Define this group as a cluster and SNP $j$ as the representative of the cluster. Include SNP $j$ in ${\cal S}$, and  remove the entire cluster from ${\cal B}$.
		\State (5) Repeat steps (3)-(4) until ${\cal B}$ is empty. Denote by $m$ number of all clumps (this is also the number of elements in ${\cal S}$). 
		\vspace{3pt}
		\State \textbf{Selection}:
		\State (6) Apply the iterative gSLOPE method (i.e. gSLOPE with $\sigma$ estimation and correction for independent regressors) on $X_{\cal S}$, being matrix $X$ restricted to columns corresponding to the set ${\cal S}$ of selected SNPs. Here, the tuning parameters, vector $\lambda$, is defined as in Procedure \ref{11091537}, with $p$ being the number of all initial SNPs, and then this vector is restricted only to first $m$ coefficients.
		\State (7) Representatives which were selected indicate the selection of entire clumps.
		\end{algorithmic}
\end{algorithm} 

Results in the context of number of discoveries given by geneSLOPE, geneSLOPE$_{XZ}$ and geneGSLOPE are summarized in Table \ref{18590924}, where we can observe that both geneSLOPE and geneSLOPE$_{XZ}$, gave identical results for LDL, CHOL and TG. Compared to these methods geneGSLOPE did not reveal any new response-related SNPs for LDL and CHOL. Actually, for these two traits geneGSLOPE missed some SNPs detected by the other two methods. 
\begin{table}[ht]
\centering
	\begin{tabular}{|c|cccc|}
  \hline 
	&HDL & LDL & TG & CHOL\\
  \hline
	geneSLOPE&7 &6 &2 &5\\
  \hline
	geneSLOPE$_{XZ}$&8 &6 &2 &5\\
  \hline
	geneGSLOPE   &8&4&8&4\\
	\hline
	New discoveries: geneSLOPE$_{XZ}$ &1&0&0&0\\
	\hline
	New discoveries: geneGSLOPE &2&0&6&0\\
  \hline
	\end{tabular}\caption{Number of discoveries in real data analysis}
	\label{18590924}
\end{table}

A different situation takes place for TG, where geneGSLOPE identifies 6 additional SNPs as compared to the other two methods. All these detections have a similar structure, showing a significant recessive effect of the minor allele.  In all these cases the minor allele frequency was smaller than 0.1. The detection of such "rare" recessive effects by the simple linear regression model is rather difficult, since the regression line adjusts mainly to the two prevalent genotype groups and is almost flat \cite{Lettre}. 

In case of HDL all three versions of SLOPE gave different results. geneSLOPE$_{XZ}$ identifies one new SNP as compared to geneSLOPE, while geneGSLOPE identifies one more SNP and misses one of the discoveries obtained by other two methods. In Figure \ref{227280916} we compare two exemplary discoveries: one detected at the same time by geneSLOPE and geneGSLOPE (known discovery) and one detected only by geneGSLOPE (new discovery). This example clearly shows the additive effect of the previously detected SNP and the recessive character of the second SNP. In case of new discovery there are only 5 individuals in the last genotype group, which makes the change in the mean not detectable by simple linear regression.  
\begin{figure}[ht]
\centering
\begin{subfigure}{.45\textwidth}
  \centering
	\includegraphics[width=1\linewidth]{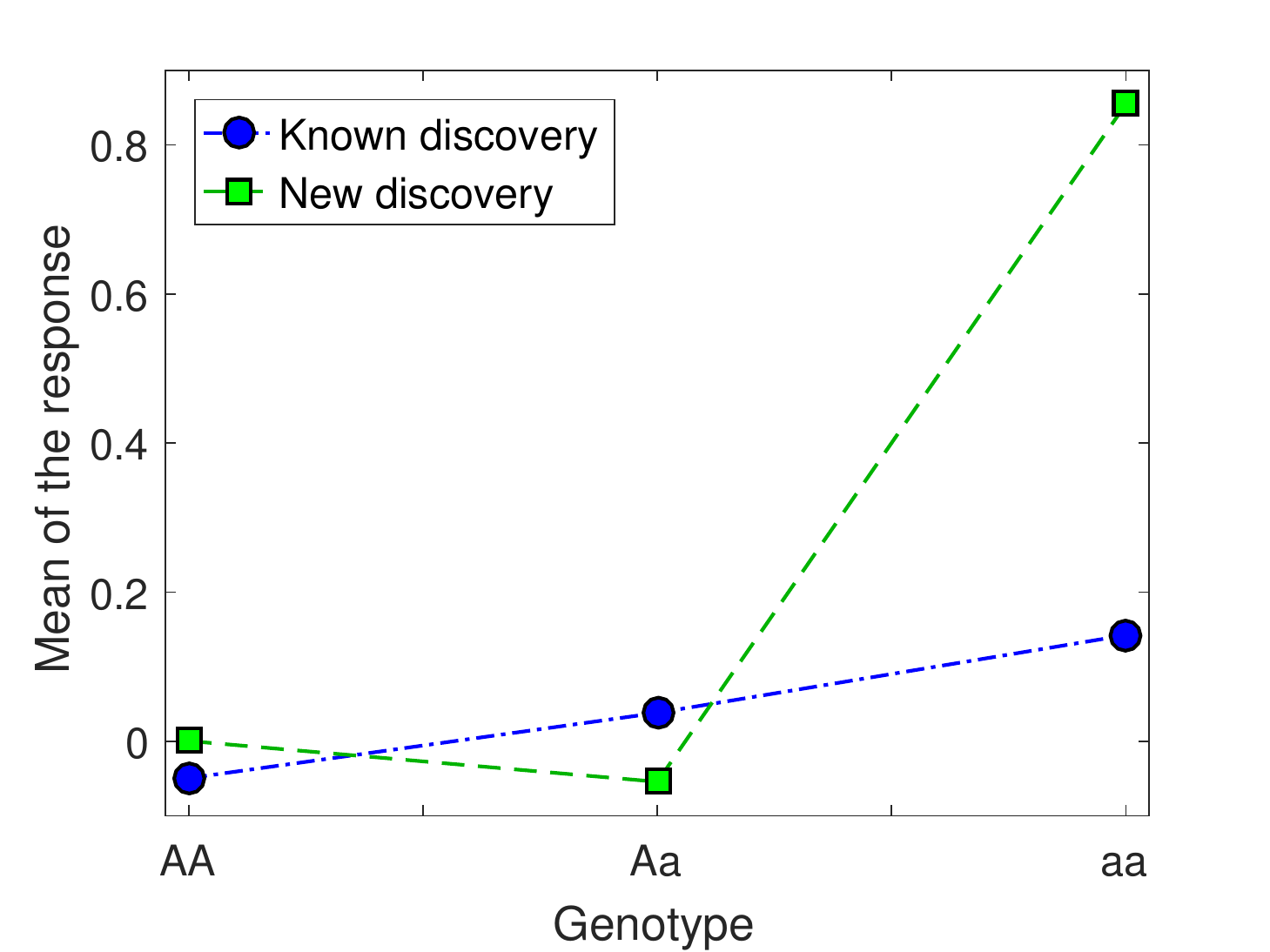}
  \caption{Mean of $y$ with respect to genotype}
\end{subfigure}%
\begin{subfigure}{.45\textwidth}
  \centering
  \includegraphics[width=1\linewidth]{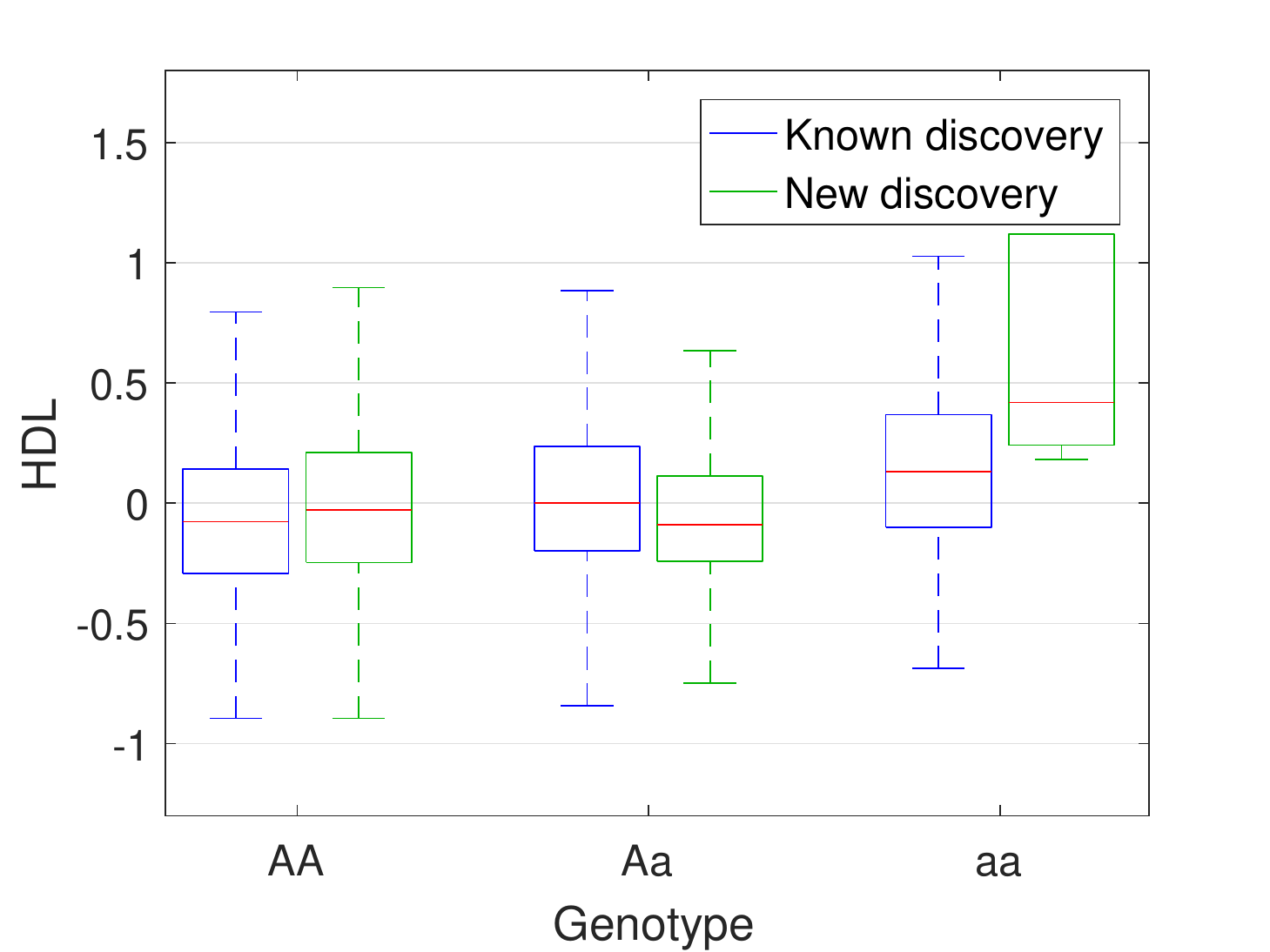}
  \caption{Boxplots (outliers not shown)}
\end{subfigure}%
\caption{Comparison of discovery detected by both geneSLOPE and geneGSLOPE (known discovery), and discovery detected only by geneGSLOPE (new discovery).}
\label{227280916}
\end{figure}

The results of real data analysis agree with results of simulations. They show that geneGSLOPE has a lower power than geneSLOPE for detection of additive effects but can be very helpful in detecting rare recessive variants. Thus these two methods are complementary to each other and can be used together to enhance the power of detection of influential genes.

\section{Discussion}

Group SLOPE is a new convex optimization procedure for selection of important groups of explanatory variables, which can be considered as a  generalization of group LASSO and of SLOPE. In this article we provide an algorithm for solving group SLOPE and discuss the choice of the sequence of regularizing parameters. Our major focus is the control of group FDR, which can be obtained when variables in different groups are orthogonal to each other or they are stochastically independent and the signal is sufficiently sparse.  After some preprocessing of the data such situations occur frequently in the context of genetic studies, which in this paper serve as a major example of applications. While we concentrated mainly on using gSLOPE to group dummy variables corresponding to different effects of the same SNP, gSLOPE can be used to group SNPs based on biological function, physical location etc. We also expect this method to be advantageous in the context of identification of groups of rare genetic variants, where considering their joint effect on phenotype should substantially increase the power of detection. 

The major purpose of controlling FDR rather than absolutely eliminating false discoveries is the wish to increase the power of detection of signals which are comparable to the noise level. As shown by a variety of theoretical and empirical results, this allows SLOPE to obtain an optimal balance between the number of false and true discoveries and leads to very good estimation and predictive properties (see e.g. \cite{SLOPE}, \cite{SLOPE2} or \cite{WE}). Our Theorem \ref{minimax} illustrates that these good estimation properties are inherited by group SLOPE.

We provide the regularizing sequence $\lambda^{max}$, which provably controls gFDR in case when variables in different groups are orthogonal. Additionally, we propose its relaxation $\lambda^{mean}$, which according to our extensive simulations controls "average" gFDR, where the average is with respect to all possible signal placements. This sequence can be easily modified taking into account the prior distribution on the signal placement. Such "Bayesian" version of gSLOPE and the proof of control of the respective average gFDR remains an interesting topic for a further research. 

Another important topic for a further research is the formal proof  of gFDR control when variables in different groups are independent and setting precise limits on the sparsity levels under which it can be done. Asymptotic formulas, which allow for very accurate prediction of FDR for LASSO under Gaussian design are provided in \cite{FDR_LASSO}. We expect that similar results can be obtained for SLOPE and gSLOPE and generalized to the case of random matrices, where variables are independent and come from sub-Gaussian distributions. However, the technical complexity of results reported in \cite{FDR_LASSO} illustrates that this task is rather challenging.  An alternative approach for the perfect  gFDR control under random designs  is to couple gSLOPE with the new knock-off procedure proposed in \cite{knockrandom}. We expect that such a combination should allow to increase the power of detection of relevant features, as compared to other methods currently used with knock-offs. 

While we concentrated on control of FDR in case when groups of variables are roughly orthogonal to each other, it is worth mentioning that original SLOPE has very interesting properties also in case when regressors are strongly correlated. As shown e.g. in \cite{OWL}, the Sorted L-One norm has a tendency to average estimated regression coefficients over groups of strongly correlated predictors, which enhances the predictive properties. This also allows not to lose important predictors due to their correlation with other features. We expect similar properties to hold for gSLOPE but the investigation of the properties of gSLOPE when variables in different groups are strongly correlated remains an interesting topic for a further research.


\section*{Acknowledgement}
We would like to thank Emmanuel J. Cand\`es and Jan Mielniczuk for helpful remarks and suggestions and Christine Peterson for screening the North Finland Birth Cohort (NFBC) dataset. D.~B. would like to thank Professor Jerzy Ombach for significant help with the process of obtaining access to the data. D.~B. and M.~B. are  supported by European Union's 7th Framework Programme for research, technological development and demonstration under Grant Agreement no 602552 and by the Polish Ministry of Science and Higher Education according to agreement 2932/7.PR/2013/2. Additionally D.B. acknowledges the support from NIMH grant R01MH108467, W.~S. was partially supported by a General Wang Yaowu Stanford Graduate Fellowship.

\vspace{20 pt}

\bibliographystyle{plain}
\bibliography{gSLOPE}

\begin{thebibliography}{10}

\bibitem{ABDJ}
{Abramovich F. and Benjamini Y. and Donoho D.~L. and Johnstone I.~M.}
\newblock {Adapting to unknown sparsity by controlling the false discovery
  rate}.
\newblock {\em {Ann. Statist.}}, {34}({2}):{584--653}, {2006}.

\bibitem{abramowitz1964}
{Abramowitz M. and Stegun I.~A.}
\newblock {\em {Handbook of mathematical functions: with formulas, graphs, and
  mathematical tables}}.
\newblock Number~{55}. {Courier Corporation}, {1964}.

\bibitem{Aka}
{Akaike H.}
\newblock {A New Look at the Statistical Model Identification}.
\newblock {\em {IEEE Transactions on Automatic Control }}, 19(6):716--723,
  1974.

\bibitem{Bakin}
{Bakin S.}
\newblock {Adaptive regression and model selection in data mining problems}.
\newblock {1999}.

\bibitem{ko}
{Barber R.~F. and Cand\`es E.~J.}
\newblock {Controlling the False Discovery Rate via Knockoffs}.
\newblock {\em {The Annals of Statistics}}, {43}({5}):{2055--2085}, {2015}.

\bibitem{FISTA}
{Beck A. and Teboulle M.}
\newblock {A fast iterative shrinkage-thresholding algorithm for linear inverse
  problems}.
\newblock {\em {SIAM Journal on Imaging Sciences}}, 1:183--202, 2009.

\bibitem{BH}
{Benjamini Y. and Hochberg Y.}
\newblock {Controlling the False Discovery Rate: A Practical and Powerful
  Approach to Multiple Testing}.
\newblock {\em {Journal of the Royal Statistical Society, Series B}},
  {57}(1):{289--300}, {1995}.

\bibitem{ABOS}
{Bogdan M. and Chakrabarti A. and Frommlet F. and Ghosh J. K.}
\newblock {Asymptotic {B}ayes Optimality under sparsity of some multiple
  testing procedures }.
\newblock {\em {Annals of Statistics}}, {39}:{1551--1579}, {2011}.

\bibitem{SLOPE2}
{Bogdan M. and van den Berg E. and Sabatti C. and Su W. and Cand\`es E. J.}
\newblock {SLOPE -- adaptive variable selection via convex optimization.}
\newblock {\em {Annals of Applied Statistics}}, 9(3):1103--1140, {2015}.

\bibitem{SLOPE}
{Bogdan M. and van den Berg E. and Su W. and Cand\`es E. J.}
\newblock {Statistical Estimation and Testing via the Ordered $\ell_1$ Norm}.
\newblock {\em {arXiv:1310.1969}}, {2013}.

\bibitem{geneSLOPE}
{Brzyski, D. and Peterson, C.B. and Sobczyk, P. and Cand\`es, E.J. and Bogdan,
  M. and Sabatti,C.}
\newblock {Controlling the rate of GWAS false discoveries}.
\newblock {\em {bioRxiv 058230, to appear in Genetics}}, {2016}.

\bibitem{knockrandom}
{Cand\`es, E.J and Fan, Y. and Janson L. and Lv J.}
\newblock {Panning for gold: model-free knockoffs for high-dimensional
  controlled variable selection}.
\newblock {\em {arXiv:1610.02351}}, {2016}.

\bibitem{donoho1994minimax}
{Donoho D.~L. and Johnstone I.~M.}
\newblock {Minimax risk over $\ell_p$-balls for $\ell_q$-error}.
\newblock {\em {Probability Theory and Related Fields}}, {99}({2}):{277--303},
  {1994}.

\bibitem{OWL}
{Figueiredo M. A. T. and Nowak R. D.}
\newblock {Ordered Weighted $l_1$ Regularized Regression with Strongly
  Correlated Covariates: Theoretical Aspects}.
\newblock {\em {Proceedings of the 19th International Conference on Artificial
  Intelligence and Statistics, JMLR:W{\normalfont \&}CP}}, {51}:{930--938},
  {2016}.

\bibitem{ABOS3}
{Frommlet F. and Bogdan M.}
\newblock {Some optimality properties of {F}{D}{R} controlling rules under
  sparsity }.
\newblock {\em {Electronic Journal of Statistics}}, {7}:{1328--1368}, {2013}.

\bibitem{Gossmann2015}
{Gossmann A. and Cao S. and Wang Y.-P.}
\newblock {Identification of significant genetic variants via {SLOPE}", and its
  extension to Group {SLOPE}}.
\newblock In {\em {Proceedings of the International Conference on
  Bioinformatics, Computational Biology and Biomedical Informatics}}, {2015}.

\bibitem{RearIneq}
{Hardy G.H. and Littlewood, J.E. and P\'{o}lya, G}.
\newblock {\em { Inequalities}}.
\newblock 1952.

\bibitem{chisquareQ}
{Inglot T.}
\newblock {Inequalities for quantiles of the chi-square distribution}.
\newblock {\em {Probability and Mathematical Statistics}},
  {30}({2}):{339--351}, {2010}.

\bibitem{Lettre}
{Lettre G. and Lange C. and Hirschhorn J. N.}
\newblock {Genetic model testing and statistical power in population-based
  association studies of quantitative traits}.
\newblock {\em {Genetic Epidemiology}}, {31}({4}):{ 358--362}, {2007}.

\bibitem{ABOS2}
{Neuvial P. and Roquain E.}
\newblock {On false discovery rate thresholding for classification under
  sparsity}.
\newblock {\em {Annals of Statistics}}, {40}:{2572--2600}, {2012}.

\bibitem{Sabatti}
{Sabatti C. and Service S. K. and Hartikainen A. and Pouta A. and Ripatti S.
  and Brodsky J. and Jones C. G. and Zaitlen N. A. and Varilo T. and Kaakinen
  M. and Sovio U. and Ruokonen A. and Laitinen J. and Jakkula E. and Coin L.
  and Hoggart C. and Collins A. and Turunen H. and Gabriel S. and Elliot P. and
  McCarthy M. I. and Daly M. J. and Järvelin M. and Freimer N. B. and Peltonen
  L.}
\newblock {Genome-wide association analysis of metabolic traits in a birth
  cohort from a founder population}.
\newblock {\em {Nature Genetics}}, {41}({1}):{35--46}, {2009}.

\bibitem{Schw}
{Schwarz G.}
\newblock {Estimating the Dimension of a Model}.
\newblock {\em {The Annals of Statistics}}, 6(2):461--464, 1978.

\bibitem{gLASSO5}
{Simon N. and Friedman J. and Hastie T. and Tibshirani R.}
\newblock {A sparse-group lasso}.
\newblock {\em {Journal of Computational and Graphical Statistics}}, 2013.

\bibitem{gLASSO6}
{Simon N. and Tibshirani R.}
\newblock {Standardization and the group lasso penalty}.
\newblock Technical report, 2011.

\bibitem{FDR_LASSO}
{Su W. and Bogdan M. and Cand\`es E.J}.
\newblock {False discoveries occur early on the lasso path}.
\newblock {\em {arXiv:1511.01957, to appear in Ann. Statist.}}, {2015}.

\bibitem{WE}
{Su W. and Cand\`es E.}
\newblock {{SLOPE} is adaptive to unknown sparsity and asymptotically minimax}.
\newblock {\em {Annals of Statistics}}, {40}:{1038--1068}, {2016}.

\bibitem{LASSO}
{Tibshirani R.}
\newblock {Regression Shrinkage and Selection Via the Lasso}.
\newblock {\em {Journal of the Royal Statistical Society, Series B}},
  {58}:{267--288}, {1996}.

\bibitem{FISTA2}
{Tseng P.}
\newblock {On accelerated proximal gradient methods for convex-concave
  optimization}.
\newblock 2008.

\bibitem{gLASSO1}
{Yuan M. and Lin Y.}
\newblock {Model selection and estimation in regression with grouped
  variables}.
\newblock {\em {Journal of the Royal Statistical Society, Series B}},
  68(1):49--67, 2006.

\end{thebibliography}
\newpage
\appendix
\section{$J_\lambda$ norm properties}
For nonnegative, nonincreasing sequence $\lambda_1 \geq \ldots \geq \lambda_p \geq 0$ consider function $ \mathbb{R}^p \ni b \longmapsto J_\lambda(b) \in \mathbb{R}$ given by $J_\lambda(b)=\sum_{i=1}^{p}\lambda_i \cdot |b|_{(i)}$, where $|b|_{(1)}\geq \ldots\geq |b|_{(p)}$ is the vector of sorted absolute values.
\begin{proposition} \label{Props1}
If $a$, $b \in \mathbb{R}^p$ are such that $|a| \preceq |b|$, then $|a|_{(\cdot)} \preceq |b|_{(\cdot)}$.
\end{proposition}

\begin{proof}
Without loss of generality we can assume that $a$ and $b$ are nonnegative and that it occurs $a_1\geq \ldots\geq a_p$. We will show that $a_k \leq b_{(k)}$ for $k\in \{1,\ldots,p\}$. Fix such $k$ and consider the set $S_k:=\{b_i:\  b_i \geq a_k\}$. It is enough to show that $|S_k|\geq k$.
For each $j\in \{1,\ldots,k\}$ we have
\begin{equation*}
b_j \geq a_j \geq a_k\ \Longrightarrow\ b_j\in S_k,
\end{equation*} 
what proves the last statement.
\end{proof}

\begin{corollary}
\label{17021046}
Let $a \in \mathbb{R}^p$, $b \in \mathbb{R}^p$ and $|a| \preceq |b|$ then Proposition (\ref{Props1}) instantly gives that ${J_\lambda (a) \leq J_\lambda (b)}$, since $J_\lambda (a)=\lambda^\mathsf{T}|a|_{(\cdot)} \leq \lambda^\mathsf{T}|b|_{(\cdot)}= J_\lambda (b)$.
\end{corollary}

\begin{proposition} \label{Props 3}
For fixed sequence $\lambda_1\geq\ldots\geq\lambda_p\geq0$, let $b\in \mathbb{R}^p$ be such that $b \succeq 0$ and $b_j > b_l$ for some ${j,l\in \{1,\ldots,p\}}$. For $0 < \varepsilon \leq (b_j-b_l)/2$, define $b_\varepsilon\in\mathbb{R}^p$ by conditions $(b_\varepsilon)_l:=b_l+\varepsilon$, $(b_\varepsilon)_j:=b_j-\varepsilon$ and $(b_\varepsilon)_i:=b_i$ for $i\notin\{j,l\}$. Then $J_\lambda(b_\varepsilon) \leq J_\lambda(b)$.
\end{proposition}
\begin{proof}
Let $\pi:\{1,\ldots,p\}\longrightarrow\{1,\ldots,p\}$ be permutation such as $\sum_{i=1}^p\lambda_i (b_\varepsilon)_{(i)}=\sum_{i=1}^p\lambda_{\pi(i)} (b_\varepsilon)_i$ for each $i$ in $\{1,\ldots,p\}$ and $\lambda_{\pi(j)}\geq\lambda_{\pi(l)}$. From the rearrangement inequality (Theorem 368 in \cite{RearIneq}),
\begin{equation}
\begin{split}
J_{\lambda}(b)-J_{\lambda}(b_\varepsilon)&=\ \sum_{i=1}^p\lambda_ib_{(i)}-\sum_{i=1}^p\lambda_i(b_\varepsilon)_{(i)}=\sum_{i=1}^p\lambda_ib_{(i)}-\sum_{i=1}^p\lambda_{\pi(i)} (b_\varepsilon)_i\\
&\geq\ \sum_{i=1}^p\lambda_{\pi(i)}b_i-\sum_{i=1}^p\lambda_{\pi(i)} (b_\varepsilon)_i=\varepsilon\big(\lambda_{\pi(j)}-\lambda_{\pi(l)}\big)\geq0.
\end{split}
\end{equation}
\end{proof}


\section{Numerical algorithm}\label{31071347}
In this section we will discuss the convexity of the objective function and the algorithm for computing the solution to gSLOPE problem (\ref{gSLOPE}). Our optimization method is based on the fast algorithm for evaluation of the proximity operator (prox) for sorted $\ell_1$ norm, which was derived in \cite{SLOPE}.
\subsection{Convexity of the objective function}
\noindent To show that the objectives in problems \eqref{gSLOPE} and (\ref{24111029}) are convex functions, we will prove the following propositions
\begin{proposition}
Function $J_{\lambda, W, \II}(b):=J_{\lambda}\Big( W\iI{b}\Big)$ is a norm for any nonnegative, nonincreasing sequence $\{\lambda_i\}_{i=1}^m$ containing at least one nonzero element, partition $\II$ of the set $\{1,\ldots,\widetilde{p}\}$ and diagonal matrix $W$ with positive elements on diagonal.
\end{proposition}
\begin{proof}
It is easy to see that $J_{\lambda, W, \II}(c) = 0$ if and only if $c=0$ and that for any scalar $\alpha\in \mathbb{R}$ it occurs $J_{\lambda, W, \II}(\alpha c) = |\alpha|J_{\lambda, W, \II}(c)$. We will show that $J_{\lambda, W, \II}$ satisfies the triangle inequality. Let $b, c$ be any vectors from $\mathbb{R}^{\widetilde{p}}$. From the positivity of $w_i$'s we have $W\iI{a+b}\preceq W\iI{a} + W\iI{b}$. Therefore, Corollary \ref{17021046} yields
\begin{equation}
\begin{split}
J_{\lambda, W, \II}\big(a+ b\big)& = \  J_{\lambda}\Big( W\iI{a+b}\Big)\leq\  J_{\lambda}\Big( W\iI{a} + W\iI{b}\Big) \\ &\leq\ J_{\lambda}\Big( W\iI{a}\Big) + J_{\lambda}\Big(W\iI{b}\Big) = J_{\lambda, W, \II}(a) + J_{\lambda, W, \II}(b)\;\;,
\end{split}
\end{equation}
since $J_{\lambda}$ is a norm.
\end{proof}

\begin{proposition}
Function $J_{\lambda}\Big( W\XI{b}\Big)$ is a seminorm for any nonnegative, nonincreasing sequence $\{\lambda_i\}_{i=1}^m$, partition $I$ of the set $\{1,\ldots,p\}$, design matrix $X\in M(n,p)$ and diagonal matrix $W$ with positive elements on diagonal.
\end{proposition}
\begin{proof}
Clearly, $J_{\lambda}\Big( W\XI{\alpha b}\Big)= |\alpha|J_{\lambda}\Big( W\XI{b}\Big)$, for any scalar $\alpha\in \mathbb{R}$. Moreover, for any $a,b\in \mathbb{R}^p$, it holds $W\XI{a+b}\preceq W\XI{a} + W\XI{b}$, and the triangle inequality could be proved similarly as in the previous proposition.
\end{proof}


\subsection{Proximal gradient method}
Consider unconstrained optimization problem of form
\begin{equation}
\label{07030001}
\minimize{b} \ f(b) = g(b) + h(b),
\end{equation}
where $g$ and $h$ are convex functions and $g$ is differentiable (for example LASSO and SLOPE are of such form). There exist efficient methods, namely {\itshape proximal gradient algorithms}, which could be applied to find numerical solution for such objective functions. To design efficient algorithms, however, $h$ must be prox-capable, meaning that there is known fast algorithm for computing the proximal operator for $h$, 
\begin{equation}
prox_{th}(y):=\argmin{b}\left\{\frac1{2t}\|y-b\|_2^2+h(b)\right\},
\end{equation}
for each $y\in\mathbb{R}^p$ and $t>0$.
The iterative algorithm works as follows. Suppose that in $k$ step $b^{(k)}$ is the current guess. Then, guess $b^{(k+1)}$ is given by
\begin{equation}
\label{17021824}
b^{(k+1)}:=\argmin{b}\left\{g\Big(b^{(k)}\Big)+\left\langle\nabla g\Big(b^{(k)}\Big),b-b^{(k)}\right\rangle+\frac1{2t}\|b-b^{(k)}\|_2^2+h(b)\right\}.
\end{equation} 
The two first terms in objective function in (\ref{17021824}) are Taylor approximation of $g$, third addend is a proximity term which is responsible for searching an update reasonably close and $t$ can be treated as a step size.

Problem (\ref{17021824}) could be reformulated to  
\begin{equation}
b^{(k+1)}:=\argmin{b}\left\{\frac12\left\|b^{(k)}-t\nabla g\big(b^{(k)}\big)-b\right\|_2^2+th(b)\right\},
\end{equation}
hence $b^{(k+1)}=prox_{th}\Big(b^{(k)}-t\nabla g\big(b^{(k)}\big)\Big)$, which justifies the need for existence of a fast algorithm computing values of the proximal operator. In each step the value of $t$ could be changed raising the sequence $\{t_i\}_{i=1}^{\infty}$. In situation when $g(b)=\frac12\|y-Xb\|_2^2$, we get the following algorithm.
\begin{algorithm}
  \caption{Proximal gradient algorithm}
	\label{24021312}
  \begin{algorithmic}
		\State \textbf{input:} $b^{[0]}\in \mathbb{R}^p$,\ \ k=0
    \While{ ( Stopping criteria are not satisfied) }
		\State 1. $b^{[k+1]}=prox_{t_kh_{\lambda}}\Big(b^{[k]}-t_kX^\mathsf{T}\big(Xb^{[k]}-y\big)\Big)$;
		\State 2. $k\gets k+1$.
    \EndWhile
  \end{algorithmic}
\end{algorithm}

\noindent It is known that $t_i$'s could be selected in different ways to ensure that $f(b^{(k)})$ converges to the optimal value \cite{FISTA}, \cite{FISTA2}.
\subsection{Proximal operator for gSLOPE}
Let $I=\{I_1,\ldots,I_m\}$, $l_i$ be rank of submatrix $X_{I_i}$ for $i=1,\ldots,m$ and $\lambda=(\lambda_1,\ldots,\lambda_m)^\mathsf{T}$ be a vector satisfying $\lambda_1 \geq\ldots\geq\lambda_m\geq0$. We will now employ the proximal gradient method to find the numerical solution to $\eqref{gSLOPE}$. As stated in subsection \ref{subsec:gs2711944}, we can focus on the equivalent optimization problem \eqref{24111029}, namely we aim to solve problem
\begin{equation}
\label{27111021}
b^*:=\argmin {b}\ \ \bigg\{\frac12\big\|y-\widetilde{X}b\big\|_2^2+\sigma J_{\lambda}\Big( W\iI{b}\Big)\bigg\},
\end{equation}
with $\II = \{\II_1,\ldots,\II_m\}$ being a partition of the set $\{1,\ldots,\widetilde{p}\}$ where $\widetilde{p}=l_1+\ldots+l_m$.

Without loss of generality we assume that $\sigma=1$. Since considered objective is of form (\ref{07030001}), we can apply proximal gradient algorithm, provided that norm $J_{\lambda, \II, W}$ is prox-capable. To compute the proximal operator for $J_{\lambda, \II, W}$ we we must be able to minimize $\frac1{2t}\|y-b\|_2^2+J_{\lambda, \II, W}(b),$ for any $y\in\mathbb{R}^{\widetilde{p}}$ and $t>0$. Multiplying objective by positive number, $t$, does not change the solution. Such operation leads to new objective function, $\frac12\|y-b\|_2^2+J_{t\lambda, \II, W}(b)$. This shows that it is enough to derive a fast algorithm for finding the numerical solution to the problem
\begin{equation}
\label{07030029}
prox_J(y):=\argmin{b}\left\{\frac12\|y-b\|_2^2+J_{\lambda, \II, W}(b)\right\},
\end{equation}
which could be applicable to arbitrary sequence $\lambda_1\geq\ldots\geq\lambda_m\geq 0.$

We will start with situation when $W$ is identity matrix. Simply, then $prox_J(y)$ is proximal operator for function $J_{\lambda,\II}(b):=J_{\lambda}(\I{b})$. In such a case computing (\ref{07030029}) could be immediately reduced to finding prox for $J_{\lambda}$ norm, since thanks to (\ref{17022353}) we have
\begin{equation}
\left\{
\begin{array}{l}
c^*=\argmin{c}\left\{\frac12\big\|\iI{y}-c\big\|_2^2+J_{\lambda}(c)\right\}\\
\big(prox_J(y)\big)_{\II_i}=c^*_i \big(\|y_{\II_i}\|_2\big)^{-1}y_{\II_i} ,\quad i=1,\ldots,m
\end{array}.
\right.
\end{equation}
Consequently, $prox_J(y)$ could be obtained by applying two steps procedure: find $c^*$ by using fast prox algorithm for $J_{\lambda}$ for vector $\iI{y}$, and compute $prox_J(y)$ by applying simple calculus to $c^*$.

Consider now general situation with fixed positive numbers $w_1,\ldots, w_m$ and define diagonal matrix $M$ by conditions ${M_{\II_i,\II_i}:=w_i^{-1}\mathbf{I}_{l_i}}$, for $i=1,\ldots,m$. Then
\begin{equation}
J_{\lambda,\II,W}(b)=J_{\lambda}\big(W\iI{b}\big) = J_{\lambda}\big(\iI{M^{-1}b}\big)=J_{\lambda,\II}\big(M^{-1}b\big).
\end{equation}
Since $M$ is nonsingular, we can substitute $\eta:=M^{-1}b$ and consider equivalent formulation of \eqref{27111021}
\begin{equation}
\label{07031414}
\left\{
\begin{array}{l}
\eta^*:=\argmin{\eta}\left\{\frac12\|y-\X M\eta\|_2^2+J_{\sigma\lambda,\II}\big(\eta\big)\right\},\\
b^* = M\eta^*
\end{array}.
\right.
\end{equation}
Therefore, after modifying the design matrix, gSLOPE can be always recast as problem with unit weights. Since $J_{\lambda,\II}$ is prox-capable, applying proximal gradient method to (\ref{07031414}) is straightforward. To implement the method introduced in this article, we have used a modified version of Procedure \ref{24021312}, the accelerated proximal gradient method known as FISTA \cite{FISTA}. In particular FISTA gives a precise procedure for choosing steps sizes, to achieve a fast convergence rate. To derive proper stopping criteria, we have considered dual problem to gSLOPE, described in the following section, and employed the strong duality property.  


\subsection{Dual norm and conjugate of grouped sorted $\ell_1$ norm}
\label{subsec:app26112110}
Let $f:\mathbb{R}^p\rightarrow \mathbb{R}$ be a norm. We will use notation $f^D$ to refer to the dual norm to $f$, i.e function defined as $f^D(x):=\underset{b}\max\big\{x^\mathsf{T}b: f(b)\leq1\big\}$. It could be shown (see \cite{SLOPE2}), that the set $C_{\lambda}$, defined as
$C_{\lambda}:=\Big\{x\in\mathbb{R}^p:\ \sum_{i=1}^k|x|_{(i)}\leq\sum_{i=1}^k\lambda_i,\ k=1,\ldots,p\Big\}$, is unit ball of the dual norm to $J_{\lambda}$ for any nonnegative, nonincreasing sequence $\{\lambda_i\}_{i=1}^p$ with at least one nonzero element. We will now consider the dual norm to $J_{\lambda,I,W}(b)=J_{\lambda}\big(W\I{b}\big)$. It holds
\begin{equation}
\begin{split}
J_{\lambda,I,W}^D(x)=\ &\underset{b}\max\big\{x^\mathsf{T}b:\ J_{\lambda,I,W}(b)\leq 1\big\}=\underset{b}\max\big\{x^\mathsf{T}b:\ J_{\lambda}(W\I{b})\leq1\big\}=\\
&\underset{b,c}\max\big\{x^\mathsf{T}b:\ J_{\lambda}(c)\leq1,\ c=W\I{b}\big\}=\underset{c}\max\big\{x^\mathsf{T}b^c:\ J_{\lambda}(c)\leq1,\ c\succeq0\big\},
\end{split}
\end{equation}
where $b^c$ is defined as $b^c:=\argmax{b}\left\{x^\mathsf{T}b:\ c=W\I{b}\right\}.$ This problem is separable and for each $i$ we have $b^c_{I_i}=\argmax{}\big\{x_{I_i}^\mathsf{T}b_{I_i}:\ c^2_i=w_i^2\|b_{I_i}\|^2_2\big\}$. Applying the Lagrange multiplier method quickly yields $x_{I_i}^\mathsf{T}b_{I_i}^c=c_iw_i^{-1}\|x_{I_i}\|_2$. Consequently,
\begin{equation}
\begin{split}
J_{\lambda,I,W}^D(x) =&\ \underset{c}\max\big\{(W^{-1}\I{x})^\mathsf{T}c:\ J_{\lambda}(c)\leq1,\ c\succeq0\big\} =\\ &\ \underset{c}\max\big\{(W^{-1}\I{x})^\mathsf{T}c:\ J_{\lambda}(c)\leq1\big\} =
\ J_{\lambda}^D\big(W^{-1}\I{x}\big).
\end{split}
\end{equation}
Therefore, $\big\{x:\ J_{\lambda,I, W}^D(x)\leq 1\big\}=\big\{x:\ J_{\lambda}^D(W^{-1}\I{x})\leq 1\big\}=\big\{x:\ W^{-1}\I{x}\in C_{\lambda}\big\}.$ Since the conjugate of norm is equal to zero for arguments from unit ball of dual norm, and equal to infinity otherwise, we immediately get
\begin{corollary}
\label{07071055}
The conjugate function for $J_{\lambda,I,W}$ is the function $J^*_{\lambda,I,W}$ defined as
\begin{equation}
J^*_{\lambda,I,W}(x)=\left\{\begin{array}{cl}
0,&W^{-1}\I{x}\in C_{\lambda}\\
\infty,&\textrm{otherwise}
\end{array}
\right..
\end{equation}
\end{corollary}


\subsection{Stopping criteria for numerical algorithm}
Without loss of generality assume that $\sigma=1$. We will start with optimization problem in (\ref{07031414}), namely
\begin{equation}
\label{24022120}
\minimize{\eta}\ \ f(\eta)=\frac12\|y-\X M\eta\|_2^2+J_{\lambda}\big(\iI{\eta}\big)
\end{equation}
for $\iI{\eta}=\Big(\|\eta_{\II_1}\|_2,\ldots,\|\eta_{\II_m}\|_2\Big)^\mathsf{T}$ and ${M_{\II_i,\II_i}=\frac1{w_i}\mathbf{I}_{l_i}}$, $i=1,\ldots,m$. This problem could be written in equivalent form
\begin{equation}
\label{07071412}
\begin{array}{cl}
\minimize{\eta, r, c}&\frac12\|r\|^2_2+c\\[.1cm]
\textrm{s.t.}&\left\{\begin{array}{l}J_{\lambda,\II}(\eta)-c\leq0\\y-r-\X M\eta=0\end{array}\right.
\end{array}
\end{equation}
$\big($notice that for $(\eta^*, r^*, c^*)$ being solution, it must occurs $c^*=J_{\lambda,\II}(\eta^*)\big)$. Since ($\ref{07071412}$) is convex and $(\eta_0, r_0, c_0)$, for $\eta_0=0$, $r_0=y$ and $c_0=1$, is strictly feasible, the strong duality holds. Lagrange dual function for this problem is given by
\begin{equation}
\begin{split}
g(\mu,\nu)=&\inf_{\eta,r,c}\ \bigg\{\frac12\|r\|_2^2+c+\mu^\mathsf{T}\big(y-r-\X M\eta\big)+\nu\big(J_{\lambda,\II}(\eta)-c\big)\bigg\}=\\
&\mu^\mathsf{T}y+\inf_r\ \bigg\{\frac12\|r\|^2_2-\mu^\mathsf{T}r\bigg\}+\inf_c\ \big\{c-\nu c\big\}+\inf_{\eta}\ \big\{-\mu^\mathsf{T}\X M\eta+\nu J_{\lambda,\II}(\eta)\big\}.
\end{split}
\end{equation}
Now, since the minimum of $\frac12\|r\|^2_2-\mu^\mathsf{T}r$ is taken for $r=\mu$, we have
\begin{equation}
g(\mu,\nu)=\mu^\mathsf{T}y-\frac12\|\mu\|^2_2+\inf_c\ \big\{c-\nu c\big\}-J_{\nu\lambda,\II}^*\big((\X M)^\mathsf{T}\mu\big).
\end{equation} 
Then $\nu^*=1$ and from Corollary \ref{07071055}, the dual problem to (\ref{07071412}) is equivalent to
\begin{equation}
\label{07071448}
\begin{array}{cl}
\maximize{\mu}&\mu^\mathsf{T}y-\frac12\|\mu\|_2^2\\[.1cm]
\textrm{s.t.}&\iI{M\X^\mathsf{T}\mu}\in C_{\lambda}
\end{array}.
\end{equation}
Let $(\eta^*,r^*,c^*)$ be primal and $(\mu^*, 1)$ be dual solution to (\ref{07071412}). Obviously, $\mu^*=r^*=y-\X M\eta^*$ and $c^*=J_{\lambda,\II}(\eta^*)$. Furthermore, from strong duality we have
\begin{equation}
\frac12\|y-\X M\eta^*\|^2_2+J_{\lambda,\II}(\eta^*) = (y-\X M\eta^*)^\mathsf{T}y-\frac12\|y-\X M\eta^*\|_2^2,
\end{equation}
which gives $(\X M\eta^*)^\mathsf{T}(y-\X M\eta^*)=J_{\lambda,\II}\big(\eta^*\big)$. Now, for current approximate $\eta^{[k]}$ of solution to (\ref{24022120}), achieved after applying proximal gradient method, we define the current duality gap for $k$ step as
\begin{equation}
\rho(\eta^{[k]})=(\X M\eta^{[k]})^\mathsf{T}(y-\X M\eta^{[k]})-J_{\lambda,\II}\big(\eta^{[k]}\big)
\end{equation} 
and we will determine the infeasibility of $\mu^{[k]}:=y-\X M\eta^{[k]}$ by using the measure
\begin{equation}
\textrm{infeas}\big(\mu^{[k]}\big): = \max\left\{J^D_{\lambda,\II}\big(M \X^\mathsf{T}\mu^{[k]}\big)-1,0\right\}
\end{equation}
To define the stopping criteria we have applied the widely used procedure: treat $\rho(\eta^{[k]})$ as indicator telling how far $\eta^{[k]}$ is from true solution and terminate the algorithm when this difference and infeasibility measure are sufficiently small. Summarizing, we have derived algorithm according to scheme
\begin{algorithm}
  \caption{group SLOPE}
  \begin{algorithmic}[!]
    \State \textbf{input:} infeas.tol:  {\itshape positive number determining the tolerance for infeasibility};
		\State \ \ \ \ \ \ \ \ \ \ dual.tol:\ \ \ {\itshape positive number determining the tolerance for duality gap};
		\State \ \ \ \ \ \ \ \ \ \ \ \ $k:=0$,\ \  $\eta^{[0]}$,\ \ $\mu^{[0]}:=\mu(\eta^{[0]})$,\ \ $\textrm{infeas}^{[0]}:=\textrm{infeas}\big(\mu^{[0]}\big)$,\ \ $\rho^{[0]}:=\rho(\eta^{[0]})$;
    \While{ ( $\textrm{infeas}^{[k]}>\textrm{infeas.tol}$\ \  or\ \  $\rho^{[k]}>\textrm{dual.tol}$) }
		\State 1. $k\gets k+1$;
		\State 2. get $\eta^{[k]}$ from Procedure \ref{24021312};
		\State 3. $\mu^{[k]}:=\mu(\eta^{[k]})$;
		\State 4. $\textrm{infeas}^{[k]}:=\textrm{infeas}\big(\mu^{[k]}\big)$, $\rho^{[k]}:=\rho(\eta^{[k]})$; 
    \EndWhile
		\State $\beta_{gS}: = M\eta^{[k]}$.
  \end{algorithmic}
\end{algorithm}


\section{Alternative representation in the orthogonal case}\label{Sec:alt_rep}

\label{subsec:06232149}
Suppose that the experiment matrix is orthogonal at group level, i.e. it holds $X_{I_i}^\mathsf{T}X_{I_j} = \mathbf{0}$, for every $i,j\in \{1,\ldots,m\}$, $i\neq j$. In such a case, $\widetilde{X}$ in problem \eqref{24111029} is orthogonal matrix, i.e. $\widetilde{X}^\mathsf{T}\widetilde{X}=\mathbf{I}_{\widetilde{p}}$. If $n=\widetilde{p}$, i.e. $\widetilde{X}$ is a square and orthogonal matrix, we also have $\widetilde{X}\widetilde{X}^\mathsf{T} = \mathbf{I}_{\widetilde{p}}$ and it obeys $\|\widetilde{X}^\mathsf{T}b\|^2_2=b^\mathsf{T}\widetilde{X}\widetilde{X}^\mathsf{T}b=\|b\|^2_2$ for $b\in\mathbb{R}^{\widetilde{p}}$. For the general case with $n\geq \widetilde{p}$, we can extend $\widetilde{X}$ to a square matrix by adding new orthonormal columns and defining $\widetilde{X}_C:=\big [ \widetilde{X}\ C \big ]$, where $C$ is composed of vectors (columns) being some complement to orthogonal basis of $\mathbb{R}^{\widetilde{p}}$. For $y\in\mathbb{R}^n$ and $b \in \mathbb{R}^{\widetilde{p}}$ we get:
\begin{equation}\label{ortog} \Big\|y-\widetilde{X}b\Big\|^2_2=\Big\|\widetilde{X}_C^\mathsf{T}\left (y-\widetilde{X}b \right)\Big\|^2_2=\left\| \left[\begin{BMAT}(b,8pt,10pt){c}{cc}\widetilde{X}^\mathsf{T}\\ C^\mathsf{T}\end{BMAT}\right]y - \left[\begin{BMAT}(b,8pt,15pt){c}{cc} b\\ \mathbf{0}\end{BMAT}\right] \right\|^2_2=\Big \|\widetilde{X}^\mathsf{T}y-b\Big \|^2_2+const,\end{equation} 
which implies that under orthogonal situation the optimization problem in $(\ref{24111029})$ could be recast as 
\begin{equation}
\label{gSLOPE_ort}
\argmin b\ \ \left\{\frac 12\big\|\widetilde{y}-b\big\|_2^2+\sigma J_{\lambda}\big( W\iI{b}\big)\right\},
\end{equation}
for $\widetilde{y}: = \widetilde{X}^\mathsf{T}y$. After introducing new variable to problem (\ref{gSLOPE_ort}), namely $c\in\mathbb{R}^m$, we get the equivalent formulation
\begin{equation}
\label{07010812}
\argmin{b,c}\ \ \left\{\frac 12\big\|\widetilde{y}-b\big\|_2^2+\sigma J_{\lambda}(c):\ c=W\iI{b}\right\}.
\end{equation}
\begin{proposition}
\label{07061804}
Let $f(b,c):\mathbb{R}^p\times\mathbb{R}^m\longrightarrow\mathbb{R}$ be any function and consider optimization problem ${\argmin{b,c}\big\{f(b,c):\ (b,c)\in\mathcal{D}\big\}}$ with unique solution $(b^*,c^*)$ and feasible set $\mathcal{D}\subset\mathbb{R}^p\times\mathbb{R}^m$. Define $\mathcal{D}^c:=\big\{c\in\mathbb{R}^m|\ \exists b\in\mathbb{R}^p: (b,c)\in \mathcal{D}\big\}$. Suppose that for any $c\in\mathcal{D}^c$, there exists unique solution, $b^c$, to problem ${\argmin{b}\big\{f(b,c):\ (b,c)\in\mathcal{D}\big\}}$. Moreover, assume that the solution to ${\argmin{c}\big\{f(b^c,c):\ c\in\mathcal{D}^c\big\}}$ is unique. Then, it occurs
\begin{equation}
\label{07061532}
\left\{
\begin{array}{l}
c^*= \argmin{c}\big\{f(b^c,c):\ c\in\mathcal{D}^c\big\}\\
b^* = b^{c^*}
\end{array}
\right..
\end{equation}
\end{proposition} 
\begin{proof}
Suppose that there exists $(b^0,c^0)\in\mathcal{D}$, such that $f(b^0,c^0)<f(b^*,c^*)$, where $b^*$ and $c^*$ are defined as in (\ref{07061532}). We have
\begin{equation}
f(b^{c^0},c^0)\leq f(b^0,c^0)<f(b^*,c^*)=f(b^{c^*},c^*),
\end{equation}
which leads to the contradiction with definition of $c^*$.
\end{proof}
We will apply the above proposition to (\ref{07010812}). Let $(b^*,c^*)$ be solution to (\ref{07010812}). Then $b^*$ is also solution to convex problem (\ref{gSLOPE_ort}) with strictly convex objective function and therefore is unique. Since $c^*=W\iI{b^*}$, $c^*$ is unique as well. In considered situation ${\mathcal{D}^c=\big\{c:\ c\succeq0\big\}}$. We will start with solving the problem $b^c=\argmin{b}\ \left\{\frac 12\big\|\widetilde{y}-b\big\|_2^2+\sigma J_{\lambda}(c):\ c=W\iI{b}\right\}.$ The additive constant in the objective could be omitted. Moreover, for each $i\in\{1,\ldots,m\}$ we have
\begin{equation}
\label{07011714}
b^c_{\II_i}=\argmin{b_{\II_i}}\left\{\big\|\widetilde{y}_{\II_i}-b_{\II_i}\big\|^2_2:\ w^2_i\|b_{\II_i}\|^2_2-c^2_i=0\right\}.
\end{equation}
The Lagrange Multipliers method quickly yields $b^c_{\II_i}=(w_i\|\widetilde{y}_{\II_i}\|_2)^{-1}c_i\widetilde{y}_{\II_i}$ and, consequently, it holds $\left\|\widetilde{y}_{\II_i}-b^c_{\II_i}\right\|_2^2=\left(\|\widetilde{y}_{\II_i}\|_2-w_i^{-1}c_i\right)^2.$
From Proposition \ref{07061804}, we get the following procedure for solution, $b^*$, to problem (\ref{gSLOPE_ort})
\begin{equation}
\left\{
\begin{array}{l}
c^*=\argmin{c}\left\{\frac12\sum_{i=1}^m\big(\|\widetilde{y}_{\II_i}\|_2-w^{-1}_ic_i\big)^2+J_{\sigma\lambda}(c)\right\}\\
b^*_{\II_i}=c^*_i \big(w_i\|\widetilde{y}_{\II_i}\|_2\big)^{-1}\widetilde{y}_{\II_i} ,\quad i=1,\ldots,m
\end{array}
\right.
\end{equation}
(notice that we applied Proposition \ref{cnonneg} to omit the constraints $c\succeq0$ and that the objective function in definition of $c^*$ is strictly feasible, which guarantees the unique solution. The above procedure yields conclusion, that indices of groups estimated by gSLOPE as relevant coincide with the support of solution to SLOPE problem with diagonal matrix having inverses of weights $w_1,\ldots,w_m$ on diagonal. Moreover, after defining $\widetilde{\beta}\in\mathbb{R}^{\widetilde{p}}$ by conditions $\widetilde{\beta}_{\II_i}:=R_i\beta_{I_i}$, $i=1,\ldots,m$, we simply have $\iI{\widetilde{\beta}} = \XI{\beta}$ and
\begin{equation}
\widetilde{y}=\widetilde{X}^Ty=\widetilde{X}^T\bigg(\sum_{i=1}^mU_iR_i\beta_{I_i}+z\bigg) = \widetilde{X}^T\big(\widetilde{X}\widetilde{\beta}+z\big) = \widetilde{\beta}+\widetilde{X}^Tz,\quad \textrm{hence}\ \widetilde{y}\sim \mathcal{N}\big(\widetilde{\beta},\ \sigma^2 \mathbf{I}_{\widetilde{p}}\big).
\end{equation}

Summarizing, if the assumption about the orthogonality at groups level is in use, one can consider the statistically equivalent model $\widetilde{y}\sim \mathcal{N}\big(\widetilde{\beta},\ \sigma^2 \mathbf{I}_{\widetilde{p}}\big)$, define truly relevant groups via the support of $\iI{\widetilde{\beta}}$ and  treat the vector $\iI{b^*}=\big(\frac{c^*_1}{w_1},\ldots, \frac{c^*_m}{w_m}\big)$ as an gSLOPE estimate of group effect sizes, where $b^*$ and $c^*$ are defined in \eqref{17022353}, i.e. it holds $\iI{b^*}=\XI{\beta^\ES{gS}}$ for any solution $\beta^\ES{gS}$ to problem $\eqref{gSLOPE}$.


\section{SLOPE with diagonal experiment matrix}\label{ap:diagonal}

Let $y\in\mathbb{R}^p$ be fixed vector and $d_1,\ldots,d_p$ be positive numbers. We will use notation $diag(d_1,\ldots,d_p)$ to define the diagonal matrix $D$ such as $D_{i,i}=d_i$ for $i=1,\ldots, p$. Denote $d:=(d_1,\ldots, d_p)^\mathsf{T}$ and let $b^*$ be the solution to SLOPE optimization problem with diagonal experiment matrix, i.e. the solution to
\begin{equation}
\label{diagSLOPE}
\minimize{b}\,f(b): = \frac 12\big\|y-Db\big\|_2^2+J_{\lambda}\big( b\big).
\end{equation} 
Since $f$ is strictly convex function, the solution to (\ref{diagSLOPE}) is unique. It is easy to observe, that changing sign of $y_i$ corresponds to changing sign at $i$th coefficient of solution as well as permuting coefficients of $y$ together with $d_i's$ permutes coefficients of $b^*$. We will summarize this observations below without proofs.
\begin{proposition}
\label{PropPrem}
Let $\pi:\{1,\ldots,p\}\longrightarrow\{1,\ldots,p\}$ be given permutation with $P_{\pi}$ as corresponding matrix. Then:
\newline i) $P_{\pi}DP_{\pi}^\mathsf{T}=diag(d_{\pi(1)},\ldots, d_{\pi(p)})$;
\newline ii) $b_{\pi}:=P_{\pi}b^*$ is solution to 
$\minimize{b}\,f_{\pi}(b): = \frac12\Big\|P_{\pi}y-P_{\pi}DP_{\pi}^\mathsf{T}b\Big\|_2^2+ J_\lambda(b);$
\newline iii) $b_S:=Sb^*$ is solution to
$\minimize{b}\,f_S(b): = \frac12\Big\|Sy-Db\Big\|_2^2+ J_\lambda(b),$
\newline where $S$ is diagonal matrix with entries on diagonal coming from set $\{-1,1\}$.
\end{proposition}
\begin{proposition}
\label{cnonneg}
If $y\succeq 0$, then $b^*\succeq 0$.
\end{proposition}
\begin{proof}
Suppose that for some $r$ it occurs $b_r<0$ for any $b\in \mathbb{R}^p$. If $y_r=0$, then taking $\widehat{b}$ defined as $\widehat{b}_i:=\left\{\begin{array}{ll}0,&i=r\\ b_i,&\textrm{otherwise}\end{array}\right.$, we get $|\widehat{b}|\preceq |b|$ and Corollary \ref{17021046} gives $J_{\lambda}(\widehat{b})\leq J_{\lambda}(b)$. Consequently,
\begin{equation*}
f(b)-f(\widehat{b})\geq\frac 12\big\|y-Db\big\|_2^2 - \frac 12\big\|y-D\widehat{b}\big\|_2^2= \frac12(y_r-d_rb_r)^2-\frac12(y_r+d_r\widehat{b}_r)^2=\frac12d_r^2b_r^2>0.
\end{equation*}
Hence $b$ could not be the solution. Now consider case when $y_r>0$ and define $\widehat{b}$ by putting $\widehat{b}_i:=\left\{\begin{array}{ll}-b_r,&i=r\\ b_i,&\textrm{otherwise}\end{array}\right..$
Then we have $J_{\lambda}(b)=J_{\lambda}(\widehat{b})$ and
\begin{equation*}
f(b)-f(\widehat{b})=\frac12(y_r-d_rb_r)^2-\frac12(y_r+d_rb_r)^2=-2y_rd_rb_r>0.
\end{equation*}
and, as before, $b$ could not be optimal.
\end{proof}
\begin{proposition}
\label{06021048}
Let $b^*$ be the solution to problem (\ref{diagSLOPE}), $\{y_i\}_{i=1}^p$ be nonnegative sequence, $\{d_i\}_{i=1}^p$ be the sequence of positive numbers and assume that
\begin{equation}
d_1y_1\geq\ldots\geq d_py_p.
\end{equation}
If $b^*$ has exactly $r$ nonzero entries for $r>0$, then the set $\{1,\ldots,r\}$ corresponds to the support of $b^*$. 
\end{proposition}
\begin{proof}
It is enough to show that $$\big(j\in \{2,\ldots,m\},\ \ b^*_j\neq 0\big)\ \Longrightarrow\ b^*_{j-1} \neq 0.$$ Suppose that this is not true. From Proposition \ref{cnonneg} we know that $b^*$ is nonnegative, hence we can find $i$ from $\{2,\ldots,m\}$ such as $b_j^*>0$ and $b_{j-1}^*=0$. For $\varepsilon \in \big(0,b_j^*/2]$ define vector $b_{\varepsilon}$ by putting $(b_\varepsilon)_{j-1}:=\varepsilon$, $(b_\varepsilon)_j:=b_j^*-\varepsilon$ and $(b_\varepsilon)_i:=b^*_i$ for $i\notin\{j,l\}$. From Proposition \ref{Props 3} we have that $J_{\lambda}(b_{\varepsilon})\leq J_{\lambda}(b^*)$, which gives
\begin{equation}
\begin{split}
f(b^*)-f(b_{\varepsilon})&\geq \frac12\big(y_{j-1}-d_{j-1}b^*_{j-1}\big)^2 + \frac12\big(y_j-d_jb^*_j\big)^2 \\
&- \frac12\big(y_{j-1}-d_{j-1}b_{\varepsilon}(j-1)\big)^2 - \frac12\big(y_j-d_jb_{\varepsilon}(j)\big)^2=\\
&\varepsilon\left(A - \frac{d_{j-1}^2+d_j^2}2\cdot\varepsilon\right),\\
&\textrm{for }A:=(y_{j-1}d_{j-1}-y_jd_j)+d_j^2b_j^*>0.
\end{split}
\end{equation}
Therefore, we can find $\varepsilon>0$ such as $f(b^*)>f(b_{\varepsilon})$, which contradicts the optimality of $b^*$.
\end{proof}
Consider now problem (\ref{diagSLOPE}) with arbitrary sequence $\{y_i\}_{i=1}^p$. Suppose that $b^*$ has exactly $r>0$ nonzero coefficients and that $\pi:\{1,\ldots,p\}\longrightarrow\{1,\ldots,p\}$ is permutation which gives the order of magnitudes for $Dy$, i.e. $d_{\pi(1)}|y|_{\pi(1)}\geq\ldots\geq d_{\pi(p)}|y|_{\pi(p)}$. Basing on our previous observations, we get important
\begin{corollary}
\label{06021230}
If $b^*$ is the solution to (\ref{diagSLOPE}) having exactly $r>0$ nonzero coefficients and $\pi$ is permutation which places components of $D|y|$ in a nonincreasing order, i.e. $d_{\pi(i)}|y|_{\pi(i)}=|Dy|_{(i)}$ for $i=1,\ldots, p$, then the support of $b^*$ is composed of the set $\{\pi(1),\ldots,\pi(r)\}$.
\end{corollary}
The next three lemmas were proven in \cite{SLOPE} in situation when $d_1=\ldots=d_p=1$. We will follow the reasoning from this paper to prove the generalized claims. The main difference is that in general case the solution to considered problem (\ref{diagSLOPE}) does not have to be nonincreasingly ordered, under assumption that $d_1y_1\geq\ldots\geq d_py_p\geq0$ (which is the case for $d_1=\ldots=d_p=1$). This makes that generalizations of proofs presented in \cite{SLOPE} are not straightforward. 
\begin{lemma}
\label{lemma1}
Consider nonnegative sequence $\{y_i\}_{i=1}^p$ and sequence of positive numbers $\{d_i\}_{i=1}^p$ such as $d_1y_1\geq\ldots\geq d_py_p.$ If $b^*$ is solution to problem (\ref{diagSLOPE}) having exactly $r$ nonzero entries, then for every $j\leq r$ it holds that
\begin{equation}
\label{06291805}
\sum_{i=j}^r(d_iy_i-\lambda_i)>0
\end{equation}
and for every $j\geq r+1$
\begin{equation}
\label{06021637}
\sum_{i=r+1}^j(d_iy_i-\lambda_i)\leq0.
\end{equation}
\end{lemma}
\begin{proof}
From Proposition \ref{06021048} we know that $b_i^*>0$ for $i\in\{1,\ldots,r\}$. Let us define $$\widetilde{b}_i:=\left\{\begin{array}{rl}b_i^*-h,&i\in\{j,\ldots,r\}\\b^*_i,&\textrm{otherwise}.\end{array}\right.,$$ where we restrict only to sufficiently small values of $h$, so as to the condition $\widetilde{b}_i>0$ is met for all $i$ from $\{j,\ldots,r\}$. For such $h$ we have $b^*_{(r+1)}=\ldots=b^*_{(p)}=\widetilde{b}_{(r+1)}=\ldots=\widetilde{b}_{(p)}=0$. Therefore there exists permutation $\pi:\{1,\ldots,r\}\longrightarrow\{1,\ldots,r\}$ such as $\sum_{i=1}^r\lambda_i\widetilde{b}_{(i)}=\sum_{i=1}^r\lambda_{\pi(i)}\widetilde{b}_i$. For such permutation we have
\begin{equation}
\label{06021616}
\begin{split}
J_{\lambda}(b^*)-J_{\lambda}(\widetilde{b})=&\ \sum_{i=1}^r\lambda_ib^*_{(i)}-\sum_{i=1}^r\lambda_i\widetilde{b}_{(i)}=\sum_{i=1}^r\lambda_ib^*_{(i)}-\sum_{i=1}^r\lambda_{\pi(i)}\widetilde{b}_i\\
\geq&\ \sum_{i=1}^r\lambda_{\pi(i)}b^*_i-\sum_{i=1}^r\lambda_{\pi(i)}\widetilde{b}_i=h\sum_{i=j}^r\lambda_{\pi(i)}\geq h\sum_{i=j}^r\lambda_i,
\end{split}
\end{equation}
where the first inequality follows from the rearrangement inequality and second is the consequence of monotonicity of $\{\lambda_i\}_{i=1}^p$. We also have
\begin{equation}
\label{06021617}
\begin{split}
\|y-Db^*\|_2^2-\|y-D\widetilde{b}\|_2^2=&\ \sum_{i=j}^r(y_i-d_ib_i^*)^2-\sum_{i=j}^r(y_i-d_ib_i^*+d_ih)^2\\
=&\ 2h\sum_{i=j}^r(d_i^2b_i^*-d_iy_i)-h^2\sum_{i=j}^rd_i^2.
\end{split}
\end{equation}
Optimality of $b^*$, (\ref{06021616}) and (\ref{06021617}) yield
\begin{equation}
0\geq f(b^*)-f(\widetilde{b})\geq h\sum_{i=j}^r(d_i^2b_i^*-d_iy_i+\lambda_i)-\frac12h^2\sum_{i=j}^rd_i^2,
\end{equation}
for each $h$ from the interval $[0,\varepsilon]$, where $\varepsilon>0$ is some (sufficiently small) value. This gives ${\sum_{i=j}^r(d_i^2b_i^*-d_iy_i+\lambda_i)\leq 0}$ and consequently
\begin{equation}
\sum_{i=j}^r(d_iy_i-\lambda_i)\geq \sum_{i=j}^rd_i^2b_i^*>0.
\end{equation}
To prove claim (\ref{06021637}), consider a new sequence defined as $\widetilde{b}_i:=\left\{\begin{array}{rl}h,&i\in\{r+1,\ldots,j\}\\b^*_i,&\textrm{otherwise}.\end{array}\right..$ We will restrict our attention only to $0<h<\min\{b_i^*:\ i\leq	 r\},$ so as to $b^*_{(\cdot)}$ and $\widetilde{b}_{(\cdot)}$ are given by applying the same permutation to $b^*$ and $\widetilde{b}$, respectively. Moreover, for each $i$ from $\{r+1,\ldots,j\}$ it holds $\widetilde{b}_{(i)}=\widetilde{b}_i=h$. From optimality of $b^*$
\begin{equation*}
0\geq f(b^*)-f(\widetilde{b})=\frac12\sum_{i=r+1}^j\left(y_i^2-(y_i-d_ih)^2\right)-\sum_{i=r+1}^j\lambda_ih=h\sum_{i=r+1}^j(d_iy_i-\lambda_i)-\frac12h^2\sum_{i=r+1}^jd_i^2,
\end{equation*}
for all considered $h$, which leads to (\ref{06021637}). 
\end{proof}
\begin{lemma}
\label{lemma2}
Let $b^*$ be solution to problem (\ref{diagSLOPE}) with nonnegative, nonincreasing sequence $\{\lambda_i\}_{i=1}^p$. Let $R(b^*)$ be number of all nonzeros in $b^*$ and $r\geq1$. Then, for any $i\in\{1,\ldots,p\}$
$$\big\{y:\ b^*_i\neq 0\textrm{ and }R(b^*)=r\big\}=\big\{y:\ d_i|y_i|>\lambda_r\textrm{ and }R(b^*)=r\big\}.$$
\end{lemma} 
\begin{proof} 
Suppose that $b^*$ has $r>0$ nonzero coefficients and let $\pi$ be permutation which places components of $D|y|$ in a nonincreasing order. From Corollary \ref{06021230} it holds that $\{i:\ b^*_i\neq0\} = \{\pi(1),\ldots, \pi(r)\}$. Define $\widetilde{y}:=P_{\pi}Sy$ and $\widetilde{D}:=P_{\pi}DP_{\pi}\T{T}$, for $S$ being the diagonal matrix such as $S_{i,i}=sgn(y_i)$. Then $P_{\pi}Sb^*$ is solution to problem
\begin{equation}
\label{06291803}
\argmin{b} \frac12\left\|\widetilde{y}-\widetilde{D}b\right\|^2_2+ J_{\lambda}(b),
\end{equation}
which satisfies the assumptions of Lemma \ref{lemma1}. Taking $j=r$ in (\ref{06291805}) and $j=r+1$ in $(\ref{06021637})$ we immediately get
\begin{equation}
\label{06291817}
d_{\pi(r)}|y|_{\pi(r)}>\lambda_r\ \ \textrm{and}\ \ d_{\pi(r+1)}|y|_{\pi(r+1)}\leq \lambda_{r+1}.
\end{equation}
We will now show that $\big\{y:\ b^*_i\neq 0\textrm{ and }R(b^*)=r\big\}\subset\big\{y:\ d_i|y_i|>\lambda_r\textrm{ and }R(b^*)=r\big\}.$ Fix $i\in\{1,\ldots,p\}$ and suppose that $b^*_i$ is nonzero coefficient. Then $i\in\{\pi(1),\ldots,\pi(r)\}$ and therefore $d_i|y_i|\geq d_{\pi(r)}|y|_{\pi(r)}>\lambda_r$, thanks to first inequality from (\ref{06291817}).
To show the second inclusion assume that $d_i|y_i|>\lambda_r$. Then, from the second inequality in (\ref{06291817}), $d_i|y_i|>\lambda_{r+1}\geq d_{\pi(r+1)}|y|_{\pi(r+1)}$, which gives $i\in\{\pi(1),\ldots,\pi(r)\}.$
\end{proof}
\begin{lemma}
\label{lemma3}
For given sequence $\{y_i\}_{i=1}^p$, sequence of positive numbers $\{d_i\}_{i=1}^p$, nonincreasing, nonnegative sequence $\{\lambda_i\}_{i=1}^p$ and fixed $j\in \{1,\ldots,p\}$, consider a following procedure
\begin{itemize}
\item define ${\widetilde{y}:=(y_1,\ldots,y_{j-1},y_{j+1},\ldots,y_p)^\mathsf{T}}$, $\widetilde{D}:=diag(d_1,\ldots,d_{j-1},d_{j+1},\ldots,d_p)$, $\widetilde{d}_i:=\widetilde{D}_{i,i}$ for $i=1,\ldots,p-1$ and ${\widetilde{\lambda}:=(\lambda_2,\ldots,\lambda_p)^\mathsf{T}}$;
\item find $\widetilde{b}^*:=\argmin{b\in\mathbb{R}^{p-1}}\frac{1}{2}\big\|\widetilde{y}-\widetilde{D}b\big\|_2^2+J_{\widetilde{\lambda}}(b);$
\item define $\widetilde{R}^j(\widetilde{b}^*):=|\{i:\ \widetilde{b}^*_i\neq 0\}|.$
\end{itemize}
Then for $r\geq1$ it holds $\big\{y:\ d_j|y_j|>\lambda_r\textrm{ and }R(b^*)=r\big\}\subset\big\{y:\ d_j|y_j|>\lambda_r\textrm{ and }\widetilde{R}^j(\widetilde{b}^*)=r-1\big\}.$
\end{lemma}
\begin{proof}
We have to show that solution $\widetilde{b}^*$ to problem
\begin{equation}
\label{07022239}
\minimize{b}\ F(b): = \frac{1}{2}\sum_{i=1}^{p-1}\left(\widetilde{y}_i-\widetilde{d}_ib_i\right)^2+\sum_{i=1}^{p-1}\widetilde{\lambda}_ib_{(i)}
\end{equation}
has exactly $r-1$ nonzero coefficients. From Proposition \ref{PropPrem} we know that the change of signs of $y_i$'s does not affect the support, hence without loss of generality we can assume that $\widetilde{y}\succeq 0$, and $\widetilde{b}^*\succeq 0$ as a result (from Proposition \ref{cnonneg}). We will start with situation when $d_1y_1\geq\ldots\geq d_py_p$ and consequently $\widetilde{d}_1\widetilde{y}_1\geq\ldots\geq \widetilde{d}_{p-1}\widetilde{y}_p$. If $j$ is fixed index such as $d_j|y_j|>\lambda_r$ and $R(b^*)=r$, this gives
\begin{equation}
\label{08020010}
j\in\{1,\ldots,r\}.
\end{equation}
To show that solution to (\ref{07022239}) has at least $r-1$ nonzero entries, suppose by contradiction that $\widetilde{b}^*$ has exactly $k-1$ nonzero entries with $k<r$. Let us define $\widehat{b}\in\mathbb{R}^{p-1}$ as
$$\widehat{b}_i:=\left\{\begin{array}{ll}h,&i\in\{k,\ldots,r-1\}\\ \widetilde{b}^*_i,&\textrm{otherwise}\end{array}\right.,$$
where $0<h<\min\{\widetilde{b}_1^*,\ldots,\widetilde{b}_{k-1}^*\}$. Then 
\begin{equation}
F(\widetilde{b}^*)-F(\widehat{b})=h\sum_{i=k}^{r-1}(\widetilde{d}_i\widetilde{y}_i-\widetilde{\lambda}_i)-h^2\sum_{i=k}^{r-1}\frac12\widetilde{d}_i^2.
\end{equation}
Now
\begin{equation}
\sum_{i=k}^{r-1}(\widetilde{d}_i\widetilde{y}_i-\widetilde{\lambda}_i)=\sum_{i=k+1}^r(\widetilde{d}_{i-1}\widetilde{y}_{i-1}-\lambda_i)\geq \sum_{i=k+1}^r(d_iy_i-\lambda_i)>0,
\end{equation}
where the first equality follows from $\widetilde{\lambda}_i=\lambda_{i+1}$, the first inequality from $\widetilde{d}_{i-1}\widetilde{y}_{i-1}\geq d_iy_i$ and the second from Lemma \ref{lemma1}. If $h$ is small enough, we get $F(\widehat{b})<F(\widetilde{b}^*)$ which leads to contradiction.

Suppose now by contradiction that $\widetilde{b}^*$ has $k$ nonzero entries with $k\geq r$ and define
$$\widehat{b}_i:=\left\{\begin{array}{ll}\widetilde{b}^*_i-h,&i\in\{r,\ldots,k\}\\ \widetilde{b}^*_i,&\textrm{otherwise}\end{array}\right..$$
Analogously to (\ref{06021616}), we get
$J_{\widetilde{\lambda}}(\widetilde{b}^*)-J_{\widetilde{\lambda}}(\widehat{b})\geq h\sum_{i=r}^k\widetilde{\lambda}_i$ and consequently
\begin{equation}
F(\widetilde{b}^*)-F(\widehat{b})\geq h\left[\sum_{i=r}^k(\widetilde{\lambda}_i-\widetilde{d}_i\widetilde{y}_i)+\sum_{i=r}^k\widetilde{d}_i^2\widetilde{b}^*_i\right]-\frac12h^2\sum_{i=r}^k\widetilde{d}^2_i.
\end{equation}
Now
\begin{equation}
\sum_{i=r}^k(\widetilde{\lambda}_i-\widetilde{d}_i\widetilde{y}_i)=\sum_{i=r+1}^{k+1}(\lambda_i-d_iy_i)\geq0,
\end{equation}
where the first equality follows from definition of $\widetilde{\lambda}$ and ($\ref{08020010}$), while the inequality follows from Lemma \ref{lemma1}. If $h$ is small enough, we get $F(\widehat{b})<F(\widetilde{b}^*)$, which contradicts the optimality of $\widetilde{b}^*$.

Consider now general situation, i.e. without assumption concerning the order of $D|y|$. Suppose that $\pi$, with corresponding matrix $P_{\pi}$, is permutation which orders $D|y|$. Define $y_{\pi}: = P_{\pi}y$ and $D_{\pi}: =P_{\pi}DP_{\pi}^\mathsf{T}$. Applying the procedure described in the statement of Lemma simultaneously to $(y, D, \lambda)$ for $j$, and to $(y_{\pi}, D_{\pi}, \lambda)$ for $\pi(j)$ we end with $\big(\widetilde{y},\widetilde{D},\widetilde{\lambda},\widetilde{R}_1^j(\widetilde{b}^*)\big)$ and $\big(\widetilde{y_{\pi}},\widetilde{D_{\pi}},\widetilde{\lambda},\widetilde{R}_2^{\pi(j)}(\widetilde{b}_{\pi}^*)\big)$. It is straightforward to see, that there exists permutation $\widetilde{\pi}:\{1,\ldots,p-1\}\longrightarrow \{1,\ldots,p-1\}$ such that $\widetilde{y_{\pi}}=P_{\widetilde{\pi}}\widetilde{y}$ and $\widetilde{D_{\pi}} = P_{\widetilde{\pi}}\widetilde{D}P_{\widetilde{\pi}}^\mathsf{T}$. From Proposition \ref{PropPrem} we have that $\widetilde{b}_{\pi}^*=P_{\widetilde{\pi}}\widetilde{b}^*$ and $\widetilde{R}_1^j(\widetilde{b}^*)=\widetilde{R}_2^{\pi(j)}(\widetilde{b}_{\pi}^*)$. Moreover, from the first part of proof $\widetilde{R}_2^{\pi(j)}(\widetilde{b}_{\pi}^*)=r-1,$ which gives the claim. 
\end{proof}

\newcommand{\E}{\mathbb{E}}
\renewcommand{\P}{\mathbb{P}}
\newcommand{\goto}{\rightarrow}
\newcommand{\var}{\operatorname{Var}}
\newcommand{\iid}{i.i.d. }
\newcommand{\floor}[1]{\lfloor #1 \rfloor}
\renewcommand{\d}{\mathrm{d}}
\section{Minimax estimation of gSLOPE}
\label{sec:minim-estim-gslope}
\begin{proof}[Proof of Theorem~\ref{minimax}]
Once again we will employ the equivalent formulation of gSLOPE under assumption about orthogonality at groups level, i.e. problem \eqref{17022353}, and we will consider statistically equivalent model $\widetilde{y}\sim \mathcal{N}\big(\widetilde{\beta},\ \sigma^2 \mathbf{I}_{\widetilde{p}}\big)$, with $\widetilde{\beta}_{\II_i}=R_i\beta_{I_i}$, $i=1,\ldots,m$. Then $\XI{\beta}=\iI{\Beta}$ and for solution $b^*$ to \eqref{17022353} it holds $\iI{b^*}=\XI{\beta^\ES{gS}}$ for any solution $\beta^\ES{gS}$ to problem $\eqref{gSLOPE}$. Without loss of generality, assume $\sigma = 1$. Note that $\|\y_{\II_i}\|_2^2$ is distributed as the noncentral $\chi^2_{l_i}(\|\Beta_{\II_i}\|_2^2)$, where $\|\Beta_{\II_i}\|_2^2$ is the noncentrality. 

The lower bound of the minimax risk can be obtained as follows. For each $\II_{i}$, only $\Beta_j$ with the smallest index $j \in \II_i$ is \textit{possibly} nonzero and the rest $l_i-1$ components of $\Beta_{\II_i}$ are fixed to be zero. Then, this is reduced to a simple Gaussian sequence model with length $m$ and sparsity at most $k$. Given the condition $k/m \goto 0$, this classical sequence model has minimax risk $(1+o(1)) 2k\log(m/k)$ (see e.g. \cite{donoho1994minimax}).

Our next step is to evaluate the worst risk of gSLOPE over the nearly black object. We would completes the proof if we show this worst risk is bounded above by $(1+o(1)) 2k\log(m/k)$. For simplicity, assume that $\|\Beta_{\II_i}\|_2 = 0$ for all $i \ge k+1$ and write $\mu_i = \|\Beta_{\II_i}\|_2, \zeta_i = \|\y_{\II_i}\|_2 \sim \chi_{l_i}(\mu_i^2)$. Denote by $\widehat\zeta$ the SLOPE solution. Then, the risk is
\[
\E\|\widehat\zeta - \mu\|_2^2 = \E \sum_{i=1}^k (\widehat\zeta_i - \mu_i)_2^2 + \E\sum_{i=k+1}^m \widehat\zeta_i^2.
\]
Then, it suffices to show
\begin{equation}\label{eq:on_risk}
\E \left[ \sum_{i=1}^k (\widehat\zeta_i - \mu_i)^2 \right] \le (1+o(1)) 2k\log(m/k)
\end{equation}
and
\begin{equation}\label{eq:off_risk}
\E \left[ \sum_{i=k+1}^m \widehat\zeta_i^2 \right] = o(1)2k\log(m/k).
\end{equation}
Below, Lemmas~\ref{lm:off_supp1}, \ref{lm:off_supp2}, and \ref{lm:off_supp3} together give \eqref{eq:off_risk}. The remaining part of this proof serves to validate \eqref{eq:on_risk}. To start with, we employ the representation $\zeta_i^2 = (\xi_{i1} + \mu_i)^2 + \xi_{i2}^2 + \cdots + \xi_{i l_i}^2$ for \iid $\xi_{ij} \sim \mathcal{N}(0,1)$ (we can assume this representation without loss of generality, since the distribution of $(\xi_{i1} + a_1)^2 + (\xi_{i2}+a_2)^2 + \cdots + (\xi_{i l_i}+a_{l_i})^2$ depends only on the non-centrality $a_1^2 + \cdots + a_{l_i}^2$). As in the proof of Lemma 3.2 in \cite{WE}, we get
\begin{equation}\label{eq:tri_risk}
\begin{aligned}
\sum_{i=1}^k (\widehat\zeta_i - \mu_i)^2  &\le  \left( \|\widehat\zeta_{[1:k]} - \zeta_{[1:k]}\|_2 + \|\zeta_{[1:k]} - \mu_{[1:k]}\|_2 \right)^2\\
& \le \left( \|\lambda_{[1:k]}\|_2 + \|\zeta_{[1:k]} - \mu_{[1:k]}\|_2 \right)^2.
\end{aligned}
\end{equation}
As $l$ is fixed and $k/m \goto 0$, \cite{chisquareQ} gives $\lambda_i \sim \sqrt{2\log\frac{m}{qi}}$ for all $i \le k$. From this we know
\begin{equation}\label{eq:lambda_sum}
\|\lambda_{[1:k]}\|_2^2 = \sum_{i=1}^k \lambda_i^2 \sim 2k\log\frac{m}{k}.
\end{equation}
Next, we see
\begin{multline}\nonumber
\left| \sqrt{(\xi_{i1} + \mu_i)^2 + \xi_{i2}^2 + \cdots + \xi_{i l_i}^2} - \mu_i \right| \le \sqrt{\xi_{i2}^2 + \cdots + \xi_{i l_i}^2} + |\xi_{i1}| \\
\le 2\sqrt{\xi_{i1}^2 + \xi_{i2}^2 + \cdots + \xi_{i l_i}^2} \equiv 2\|\xi_i\|_2,
\end{multline}
which yields
\begin{equation}\label{eq:chi_sum}
\|\zeta_{[1:k]} - \mu_{[1:k]}\|_2^2 \le 4\sum_{i=1}^k \|\xi_i\|_2^2
\end{equation}
Note that $\sum_{i=1}^k \|\xi_i\|_2^2$ is distributed as the chi-square with $l_1 + \cdots + l_k \le lk$ degrees of freedom. Taking \eqref{eq:lambda_sum} and \eqref{eq:chi_sum} together, from \eqref{eq:tri_risk} we get
\begin{align*}
\E \left[ \sum_{i=1}^k (\widehat\zeta_i - \mu_i)^2 \right]  &\le \|\lambda_{[1:k]}\|_2^2 +\E \|\zeta_{[1:k]} - \mu_{[1:k]}\|_2^2 + 2\|\lambda_{[1:k]}\|_2 \E \|\zeta_{[1:k]} - \mu_{[1:k]}\|_2\\
&\le (1+o(1))2k\log\frac{m}{k} + 4lk + 2\sqrt{(1+o(1))2k\log\frac{m}{k}} \cdot \sqrt{4lk}\\
& \sim (1+o(1))2k\log\frac{m}{k},
\end{align*}
where the last step makes use of $m/k \goto \infty$. This establishes \eqref{eq:on_risk} and consequently completes the proof.

\end{proof}

The following three lemmas aim to prove \eqref{eq:off_risk}. Denote by $\zeta_{(1)} \ge \cdots \ge \zeta_{(m-k)}$ the order statistics of $\zeta_{k+1}, \ldots, \zeta_m$. Recall that $\zeta_i \sim \chi_{l_i}$ for $i \ge k+1$. As in the proof of Lemma 3.3 in \cite{WE}, we have
\[
\sum_{i=k+1}^m \widehat\zeta_i^2 \le \sum_{i=1}^{m-k} (\zeta_{(i)} - \lambda_{k+i})_+^2,
\]
where $x_+ = \max\{x, 0\}$. For a sufficiently large constant $A > 0$ and sufficiently small constant $\alpha > 0$ both to be specified later, we partition the sum into three parts:
\[
\sum_{i=1}^{m-k} \E (\zeta_{(i)} - \lambda_{k+i})_+^2 = \sum_{i=1}^{\floor{Ak}} \E (\zeta_{(i)} - \lambda_{k+i})_+^2 + \sum_{i=\lceil Ak \rceil }^{\floor{\alpha m}} \E (\zeta_{(i)} - \lambda_{k+i})_+^2 + \sum_{i=\lceil \alpha m \rceil }^{m-k} \E (\zeta_{(i)} - \lambda_{k+i})_+^2
\]
The three lemmas, respectively, show that each part is negligible compared with $2k\log(m/k)$. We indeed prove a stronger version in which the order statistics $\zeta_{(1)} \ge \cdots \ge \zeta_{(m-k)} \ge \zeta_{(m-k+1)} \ge \cdots \ge \zeta_{(m)}$ come from $m$ \iid $\chi_l$. Let $U_1, \ldots, U_m$ be \iid uniform random variables on $(0, 1)$, and $U_{(1)} \le U_{(2)} \le \cdots \le U_{(m)}$ be the \textit{increasing} order statistics. So we have the representation $\zeta_{(i)} = F^{-1}_{\chi_l}(1 - U_{(i)})$

\begin{lemma}\label{lm:off_supp1}
Under the preceding conditions, for any $A > 0$ we have
\[
\frac1{2k\log(m/k)} \sum_{i=1}^{\lfloor Ak \rfloor} \E(\zeta_{(i)} - \lambda_{k+i})_+^2 \rightarrow 0.
\]
\end{lemma}
\begin{proof}[Proof of Lemma~\ref{lm:off_supp1}]
Recognizing that $l$ is fixed, from \cite{chisquareQ} it follows that
\[
F^{-1}_{\chi_l}(1 - q_1) - F^{-1}_{\chi_l}(1 - q_2) \sim \sqrt{2\log\frac1{q_1}} - \sqrt{2\log\frac1{q_2}}
\]
for $q_1, q_2 \goto 0$. We also know that $\zeta_i$ is distributed as $F^{-1}_{\chi_l}(1 - U_{(i)})$. Making use of these facts, we get
\begin{equation}\nonumber
\begin{aligned}
\E (\zeta_{(i)} - \lambda_{k+i})_+^2 &= \E (F^{-1}_{\chi_l}(1 - U_{(i)}) - F^{-1}_{\chi_l}(1 - q(k+i)/m))_+^2 \\
&\sim \E \left( \sqrt{2\log\frac1{U_{(i)}}} - \sqrt{2\log\frac{m}{q(k+i)}} \right)_+^2\\
&\le \E \left( \sqrt{2\log\frac1{U_{(i)}}} - \sqrt{2\log\frac{m}{q(k+i)}} \right)^2\\
&\lesssim \E \left( \frac{\log^2(q(k+i)/m U_{(i)})}{\log(m/q(k+i))} \right).
\end{aligned}
\end{equation}
Now, we proceed to evaluate 
\[
\E \left[ \log^2\frac{q(k+i)}{m U_{(i)}} \right]= \log^2\frac{q(k+i)}{m} + \E \log^2 U_{(i)} - 2\log\frac{q(k+i)}{m} \E \log U_{(i)}.
\]
Observing that $U_{(i)}$ follows $\mathrm{Beta}(i, m+i-i)$, we get (see e.g. \cite{abramowitz1964})
\[
\begin{aligned}
&\E \log U_{(i)} = -\log\frac{m+1}{i} + \delta_1,\\
&\E \log^2 U_{(i)} = \left( \log\frac{m+1}{i} - \delta_1 \right)^2 + \frac1i - \frac1{m+1} + \delta_2
\end{aligned}
\]
for some $\delta_1 = O(1/i)$ and $\delta_2 = O(1/i^2)$. Thus we can evaluate $\E\log^2\frac{q(k+i)}{mU_{(i)}}$ as
\begin{multline}\nonumber
\E\log^2\frac{q(k+i)}{mU_{(i)}} = \log^2\frac{q(k+i)}{m} - 2\log\frac{q(k+i)}{m}\E \log U_{(i)} + \E\log^2U_{(i)}\\
= \log^2\frac{q(k+i)}{m} + 2\log\frac{q(k+i)}{m}\left(\log\frac{m+1}{i} - \delta_1\right) + \left(\log\frac{m+1}{i} - \delta_1\right)^2 + \frac{1}{i} - \frac{1}{m+1} + \delta_2\\
 = \log^2\frac{q(k+i)(m+1)}{im} - 2\delta_1\log\frac{q(k+i)(m+1)}{im} + \frac1i  - \frac1{m+1} + \delta_1^2 + \delta_2.
\end{multline}
Hence, we get
\begin{align*}
&\sum_{i=1}^{\floor{Ak}}\E(\zeta_{(i)} - \lambda_{k+i})_+^2 \\
&\lesssim \frac1{\log\frac{m}{q(A+1)k}}\left(\underbrace{\sum_{i=1}^{\floor{Ak}}\log^2\frac{q(k+i)(m+1)}{im}}_{\mbox{I}} - \underbrace{\sum_{i=1}^{\floor{Ak}}2\delta_1\log\frac{q(k+i)(m+1)}{im}}_{\mbox{II}} + \underbrace{\sum_{i=1}^{\floor{Ak}}\big(\frac1{i} - \frac1{m+1} + \delta_1^2 + \delta_2\big)}_{\mbox{III}} \right)\\
&= \frac1{\log\frac{m}{q(A+1)k}} (\mbox{I} + |\mbox{II}| + |\mbox{III}|).
\end{align*}
Since $\frac{m}{q(A+1)k} \goto \infty$. The proof would be completed once we show $\mbox{I}, |\mbox{II}|$, and $|\mbox{III}|$ are bounded. To this end, first note that
\begin{equation}\nonumber
\begin{aligned}
\mbox{I} &= \sum_{i=1}^{\floor{Ak}} \log^2\frac{q(k+i)(m+1)}{im} \\&\le \sum_{i=1}^{\floor{Ak}}\max\left\{k\int^{i/k}_{(i-1)/k}\log^2\frac{q(m+1)(1+x)}{mx}\d x, k\int^{(i+1)/k}_{i/k}\log^2\frac{q(m+1)(1+x)}{mx}\d x \right\}\\
& \le 2k\int^{A+1}_0 \log^2\frac{q(m+1)(1+x)}{mx}\d x \asymp k = o\left(2k\log\frac{m}{k}\right).
\end{aligned}
\end{equation}
The second term $\mbox{II}$ obeys
\begin{equation}\nonumber
\begin{aligned}
|\mbox{II}| &\le \sum_{i=1}^{\floor{Ak}}2\Big|\delta_1 \log\frac{q(k+i)(m+1)}{im}\Big| \lesssim \sum_{i=1}^{\floor{Ak}}\frac1i \Big|\log\frac{q(k+i)(m+1)}{im}\Big| \\
&\le \sum_{i=1}^{\floor{Ak}}\max\left\{ k\int^{i/k}_{(i-1)/k}\Big|\log\frac{q(m+1)(1+x)}{mx}\Big|\d x, k\int^{(i+1)/k}_{i/k}\Big|\log\frac{q(m+1)(1+x)}{mx}\Big|\d x \right\}\\
&\le 2k\int^{A+1}_{0}\Big|\log\frac{q(m+1)(1+x)}{mx}\Big|\d x \asymp k = o\left(2k\log\frac{m}{k} \right),
\end{aligned}  
\end{equation}
where we use the fact that $\int^{A+1}_{0}\Big|\log\frac{q(m+1)(1+x)}{mx}\Big|\d x$ is bounded by some constant.
The last term is simply bounded as
\begin{equation}\nonumber
|\mbox{III}| \le \sum_{i=1}^{\floor{Ak}}\Big|\frac1i  - \frac1{m+1} + \delta_1^2 + \delta_2\Big| \lesssim \sum_{i=1}^{\floor{Ak}}\frac{1}{i} \lesssim \log (Ak) = o\left(2k\log\frac{m}{k}\right).
\end{equation}
Combining these established bounds on $\mbox{I}, \mbox{II}$, and $\mbox{III}$ finishes proof.
\end{proof}


\begin{lemma}\label{lm:off_supp2}
Under the preceding conditions, let $A$ be any constant satisfying $q(1+A)/A < 1$ and $\alpha$ be sufficiently small such that $l/\lambda_{k + \floor{\alpha m}} < 1/2$. Then, 
\[
\frac1{2k\log(m/k)} \sum_{i = \lceil Ak \rceil}^{\lfloor \alpha m \rfloor} \E(\zeta_{(i)} - \lambda_{k+i})_+^2 \rightarrow 0.
\]  
\end{lemma}
\begin{proof}[Proof of Lemma~\ref{lm:off_supp2}]
Note that $\lambda_{k + \floor{\alpha m}} \sim \sqrt{2\log\frac{m}{q(k + \floor{\alpha m})}} \sim \sqrt{2\log\frac1{q\alpha}}$. So it is clear that such $\alpha$ exists. Pick any fixed $i$ between $\lceil Ak \rceil$ and$\lfloor \alpha m \rfloor$. As in the proof of Lemma A.4 in \cite{WE}, denote by $\alpha_u = \P(\chi_l > \lambda_{k+i} + u)$. Note that
\begin{equation}\nonumber
\begin{aligned}
\alpha_u &= \P(\chi_l > \lambda_{k+i} + u) = \int^{\infty}_{(\lambda_{k+i} + u)^2} \frac1{\mathrm{e}^{l/2}\Gamma(l/2)} x^{l/2-1} \mathrm{e}^{-x/2} \mathrm{d} x\\
&= \int^{\infty}_{\lambda_{k+i}^2} \frac1{\mathrm{e}^{l/2}\Gamma(l/2)} \left(\frac{(\lambda_{k+i}+u)^2}{\lambda_{k+i}^2} y \right)^{l/2-1} \exp\left(- \frac{(\lambda_{k+i}+u)^2}{2\lambda^2_{k+i}}y \right) \mathrm{d} \frac{(\lambda_{k+i}+u)^2}{\lambda_{k+i}^2} y\\
&= \left(1 + \frac{u}{\lambda_{k+i}}\right)^l \int^{\infty}_{\lambda_{k+i}^2} \frac1{\mathrm{e}^{l/2}\Gamma(l/2)} y^{l/2-1} \exp\left(- \frac{(\lambda_{k+i}+u)^2}{2\lambda^2_{k+i}}y \right) \mathrm{d} y\\
&\le \left(1 + \frac{u}{\lambda_{k+i}}\right)^l \mathrm{e}^{-\lambda_{k+i} u} \int^{\infty}_{\lambda_{k+i}^2} \frac1{e^{l/2}\Gamma(l/2)} y^{l/2-1} \mathrm{e}^{-y/2}\mathrm{d} y\\
&= \left(1 + \frac{u}{\lambda_{k+i}}\right)^l \mathrm{e}^{-\lambda_{k+i} u} \alpha_0\\
&\le \exp \left( \frac{l}{\lambda_{k+i}} u-\lambda_{k+i}u \right) \alpha_0.
\end{aligned}
\end{equation}
With the proviso that $l/\lambda_{k+\lfloor \alpha m \rfloor} < 1/2 < \lambda_{k+\lfloor \alpha m \rfloor}/2$, it follows that
\[
\alpha_u \le \mathrm{e}^{-\lambda_{k+i}u/2 }\alpha_0.
\]
The remaining proof follows from exactly the same reasoning as that of Lemma A.4 in \cite{WE}.

\end{proof}


\begin{lemma}\label{lm:off_supp3}
Under the preceding conditions, for any constant $\alpha > 0$ we have
\[
\frac1{2k\log(m/k)} \sum_{i = \lceil \alpha m \rceil}^{m - k} \E(\zeta_{(i)} - \lambda_{k+i})_+^2 \rightarrow 0.
\]
\end{lemma}

\begin{proof}[Proof of Lemma~\ref{lm:off_supp3}]
Recognizing that the value of the summation increases as $\alpha$ decreases, we only prove the lemma for sufficiently small $\alpha$. In the case of $U_{(i)} \ge \alpha/3$, we get
\begin{equation}\nonumber
\begin{aligned}
(\zeta_{(i)} - \lambda_{k+i})_+ &= \left( F^{-1}_{\chi_l}(1 - U_{(i)}) - F^{-1}_{\chi_l}(1 - q(k+i)/m) \right)_+\\
&\asymp \left(1 - U_{(i)} - (1 - q(k+i)/m)\right)_+\\
&= (q(k+i)/m - U_{(i)})_+,
\end{aligned}
\end{equation}
since both $U_{(i)}$ and $q(k+i)/m$ are bounded below away from zero. Otherwise, we use the trivial inequality $(\zeta_{(i)} - \lambda_{k+i})_+ \le \zeta_{(i)}$. In either case, we get
\[
\begin{aligned}
(\zeta_{(i)} - \lambda_{k+i})_+^2  &\lesssim \zeta_{(i)}^2 \bm{1}_{U_{(i)} < \frac{\alpha}{3}} + \left(\frac{q(k+i)}{m} - U_{(i)} \right)_+^2\\
&= \Big(F^{-1}_{\chi_l}(1 - U_{(i)}) \Big)^2 \bm{1}_{U_{(i)} < \frac{\alpha}{3}} + \left(\frac{q(k+i)}{m} - U_{(i)} \right)_+^2\\
&\asymp 2\log \left( \frac1{ U_{(i)}} \right) \bm{1}_{U_{(i)} < \frac{\alpha}{3}} + \left(\frac{q(k+i)}{m} - U_{(i)} \right)_+^2\\
&\lesssim \log \left( \frac1{ U_{(i)}} \right) \bm{1}_{U_{(i)} < \frac{\alpha}{3}} + \bm{1}_{U_{(i)} \le \frac{q(k+i)}{m}}.\\
\end{aligned}
\]
Hence,
\[
\sum_{i = \lceil \alpha m \rceil}^{m - k} \E(\zeta_{(i)} - \lambda_{k+i})_+^2 \lesssim \sum_{i = \lceil \alpha m \rceil}^{m - k} \E\left(\log \left( \frac1{ U_{(i)}} \right); U_{(i)} < \frac{\alpha}{3} \right) + \sum_{i = \lceil \alpha m \rceil}^{m - k} \P\left(U_{(i)} \le \frac{q(k+i)}{m}\right)
\]
In the remaining proof we aim to show
\begin{equation}\label{eq:third_alpha_sum}
\sum_{i = \lceil \alpha m \rceil}^{m - k} \E\left(\log \left( \frac1{ U_{(i)}} \right); U_{(i)} < \frac{\alpha}{3} \right) \goto 0
\end{equation}
and
\begin{equation}\label{eq:count_alpha}  
\sum_{i = \lceil \alpha m \rceil}^{m - k} \P\left(U_{(i)} \le \frac{q(k+i)}{m}\right) \goto 0.
\end{equation}
This is more than we need since $2k\log(m/k) \goto \infty$.

Each summand of \eqref{eq:third_alpha_sum} is bounded above by
\begin{equation}\nonumber
\begin{aligned}
\E\left(\log \left( \frac1{ U_{(\lceil \alpha m \rceil)}} \right); U_{(\lceil \alpha m \rceil)} < \frac{\alpha}{3} \right) &= \displaystyle\int^{\frac{\alpha}{3}}_0 \frac{x^{\lceil \alpha m \rceil-1}(1-x)^{m-\lceil \alpha m \rceil} \log \frac1x}{\operatorname{B}(\lceil \alpha m \rceil, m+1-\lceil \alpha m \rceil)} \d x\\
&\le \int^{\frac{\alpha}{3}}_0 \frac{x^{\lceil \alpha m \rceil-1}\log \frac1x}{\operatorname{B}(\lceil \alpha m \rceil, m+1-\lceil \alpha m \rceil)} \d x  \\
&= \frac1{\lceil \alpha m \rceil^2 \operatorname{B}(\lceil \alpha m \rceil, m+1-\lceil \alpha m \rceil)}\int^{(\frac{\alpha}{3})^{\lceil \alpha m \rceil}}_0 \log\frac1{y} ~ \d y\\
&\sim \frac{(\alpha/3)^{\lceil \alpha m \rceil} \log\frac3{\alpha}}{\lceil \alpha m \rceil \operatorname{B}(\lceil \alpha m \rceil, m+1-\lceil \alpha m \rceil)}.
\end{aligned}
\end{equation}
The last line obeys
\[
\begin{aligned}
\log\left[ \frac{(\alpha/3)^{\lceil \alpha m \rceil} }{\operatorname{B}(\lceil \alpha m \rceil, m+1-\lceil \alpha m \rceil)} \right] &\sim -\alpha m \log\frac3{\alpha} + \alpha m \log\frac1{\alpha}  + (1-\alpha)m \log \frac1{1-\alpha}\\
& = -\alpha m \log 3  + (1-\alpha)m \log \frac1{1-\alpha}.\\
\end{aligned}
\]
For small $\alpha$, we get $-\alpha \log 3  + (1-\alpha) \log \frac1{1-\alpha} = -\alpha\log 3 + (1+o(1))(1-\alpha)\alpha = -(\log 3  -1 + o(1))\alpha$. (Note that $\log 3 -1 = 0.0986\ldots > 0$.) This immediately yields
\[
\E\left(\log \left( \frac1{ U_{(\lceil \alpha m \rceil)}} \right); U_{(\lceil \alpha m \rceil)} < \frac{\alpha}{3} \right) \sim \mathrm{e}^{-(\log 3  -1 + o(1))\alpha m},
\]
which implies \eqref{eq:third_alpha_sum} since $m\mathrm{e}^{-(\log 3  -1 + o(1))\alpha m} \goto 0$.

Next, we turn to show \eqref{eq:count_alpha}. Note that $\P\left(U_{(i)} \le \frac{q(k+i)}{m}\right)$ actually is the tail probability of the binomial distribution with $m$ trials and success probability $\frac{q(k+i)}{m}$. Hence, by the Chernoff bound, this
probability is bounded as
\[
\P\left(U_{(i)} \le \frac{q(k+i)}{m}\right) \le \exp\left( -m \operatorname{KL}(i/m || q(k+i)/m) \right),
\]
where $\operatorname{KL}(a||b) := a \log\frac{a}{b} + (1-a)\log\frac{1-a}{1-b}$ is the Kullback-Leibler divergence. Thanks to $i \ge \lceil \alpha m \rceil \gg k$, simple analysis reveals that
\[
\operatorname{KL}(i/m || q(k+i)/m) \ge (1+o(1))i \left( \log\frac1q - 1 + q \right) / m.
\]
Combining the last two displays gives
\[
\P\left(U_{(i)} \le \frac{q(k+i)}{m}\right) \le \mathrm{e}^{-(1+o(1)) \left( \log\frac1q - 1 + q \right)i}.
\]
Plugging the above inequality into \eqref{eq:count_alpha} yields
\[
\sum_{i = \lceil \alpha m \rceil}^{m - k} \P\left(U_{(i)} \le \frac{q(k+i)}{m}\right) \le  \sum_{i = \lceil \alpha m \rceil}^{m - k} \mathrm{e}^{-(1+o(1)) \left( \log\frac1q - 1 + q \right)i}   \goto 0,
\]
where the last step follows from $\log\frac1q - 1 + q  > 0$ and $\lceil \alpha m \rceil \goto \infty$.
\end{proof}

\section{Strength of signals}
Consider the case when all submatrices $X_{I_i}$ have the same rank, $l>0$, $w>0$ is used as the universal weight and $X$ is orthogonal at groups level. From the interpretation of gSLOPE estimate coming from \eqref{17022353}, we see that the identification of the relevant groups could be summarized as follows: $\lambda$ decides on the number, $R$, of groups labeled as relevant, which correspond to indices of the $R$ largest values among $w^{-1}\|\y_{\II_1}\|_2,\ldots,w^{-1}\|\y_{\II_m}\|_2$. The random variables $w^{-1} \|\y_{\II_i}\|_2$ have a (possibly) non-central $\Chi$ distributions with $l$ degrees of freedom and noncentrality parameters given by the entries of $\iI{\Beta}$. Now, the nonzero $\|\Beta_{\II_i}\|_2$ could be perceived as a strong signal, if with the high probability the random variable having the noncentral $\Chi$ distribution with the noncentrality parameter $\|\Beta_{\II_i}\|_2$ is large comparing to the background composed of the independent random variables with the $\Chi_l$ distributions (then signal is likely to be identified by gSLOPE; otherwise, the signal could be easily covered by random disturbances and its identification has more in common with good luck than with the usage of particular method). The important quantity, which could be treated as a breaking point, is the expected value of the maximum of the background noise. Group effects being close to this value, could be perceived as medium under the orthogonal case and weak under the occurrence of correlations between groups. The above reasoning applied to the considered case, yields the issue of approximation of the expected value of the maximum of $m$ independent $\Chi_l$-distributed variables. Suppose that $\Psi_i\sim \Chi_l$ for $i=\{1,\ldots,m\}$. From Jensen's inequality we have 
$$\mathbb{E}\left(\max_{i=1,\ldots,m}\{\Psi_i\}\right) = \mathbb{E}\left(\sqrt{\max_{i=1,\ldots,m}\{\Psi_i^2\}}\right)\leq\sqrt{\mathbb{E}\left(\max_{i=1,\ldots,m}\{\Psi^2_i\}\right)},$$ 
hence we will replace the last problem by the problem of finding the reasonable upper bound on the expected value of the maximum of $m$ independent, $\Chi_l^2$-distributed variables.
\begin{theorem}
Let $\Psi_1,\ldots,\Psi_m$ be independent variables, $\Psi_i\sim\Chi_l^2$ for all $i$. Then
\begin{equation}
\label{05031231}
\mathbb{E}\left(\max_{i=1,\ldots,m}\{\Psi_i\}\right)\leq \frac{4\ln(m)}{1-m^{-\frac2l}}.
\end{equation}
\end{theorem}
\begin{proof}
Denote $M_m:=\max_{i=1,\ldots,m}\{\Psi_i\}$. From the Jensen's inequality applied to $e^{tM_m}$ we have
\begin{equation}
e^{t\mathbb{E}[M_m]}\leq \mathbb{E}\left[e^{tM_m}\right]=\mathbb{E}\left[\max_{i=1,\ldots,m}e^{t\Psi_i}\right]\leq\sum_{i=1}^m\mathbb{E}\left[e^{t\Psi_i}\right].
\end{equation}
We will consider only $t\in[0,\frac12)$. Since the moment generating function for $\Chi^2_l$ distribution is given by $MGF:=(1-2t)^{-\frac l2}$, for each $i$ it holds $\mathbb{E}\left[e^{t\Psi_i}\right]=(1-2t)^{-\frac l2}$ and we get $e^{t\mathbb{E}[M_m]}\leq m (1-2t)^{-\frac l2}$. Applying the natural logarithm to both sides yields
\begin{equation}
\label{05031230}
\mathbb{E}[M_m]\leq\frac{\ln(m)+\ln\left((1-2t)^{-\frac l2}\right)}{t},\ \ t\in\left[0, 1/2\right).
\end{equation}
Define $t_{m,l}: = \frac{1-m^{-\frac 2l}}2$. Then for all positive, natural numbers $l$ and $m$ we have $t_{m,l}\in[0,\frac12)$. Plugging $t_{m,l}$ to the right side of (\ref{05031230}) gives inequality (\ref{05031231}) and finishes the proof.
\end{proof}
\noindent The above theorem gives us the motivation to use the quantity $\sqrt{4\ln(m)/(1-m^{-2/l})}$ as the upper bound on the expected value of maximum over $m$ independent $\Chi_l$-distributed variables. In all simulations, which we have performed to investigate the performance of gSLOPE, we have generated the effects for truly relevant groups basing on these upper bounds. In particular, in experiments where $l_i$'s as well as weights were identical, we aimed at $\mathbb{E}\left(\|\y_{\II_i}\|_2\right)=\sqrt{4\ln(m)/(1-m^{-2/l})}$, for the truly relevant group $i$. Since $\mathbb{E}\left(\|\y_{\II_i}\|_2\right)\approx \sqrt{\|\Beta_{\II_i}\|_2^2+l}$, this yields the setting
\begin{equation}
\label{05031641}
\|\Beta_{\II_i}\|_2 = B(m,l), \qquad\textrm{for}\qquad B(m,l):=\sqrt{4\ln(m)/(1-m^{-2/l})-l}
\end{equation}
for groups chosen to be truly relevant.


\section{The sequence of tuning parameters when variables in different groups are independent}\label{subsec:06232149}
To model the situation when variables in different groups are stochastically independent we will assume that $n$ by $p$ design matrix is a realization of the random matrix with independent entries drawn from the normal distribution, $\mathcal{N}\big(0,\frac 1n\big)$, so as the expected value of $X_i\T{T}X_j$ is equal to $1$ for $i = j$, and equal to $0$ otherwise. The main objective is to derive the lambda sequence, which could be applied to achieve gFDR control under assumption that the $\XI{\beta}$ is sparse. At first we will confine ourselves only to the case $l_1=\ldots=l_m:=l$, $w_1=\ldots=w_m:=w$ and when the number of elements in each group is relatively small as compared to the number of observations ($l<<n$). For simplicity in this subsection we will fix $\sigma = 1$. In case when $\sigma \neq 1$, the proposed sequence lambda should be multiplied by $\sigma$, as in expression (\ref{gSLOPE}). In the heuristics presented in this subsection, we will use the notation $A \approx B$, in order to express that with large probability the differences between corresponding entries of matrices $A$ and $B$ are very small. 

In situation when entries of $X$ come from $\mathcal{N}\big(0,\frac 1n\big)$ distribution and sizes of groups are relatively  small, a very good approximation of $\beta^\ES{gS}$ could be obtained by $\hat{\beta}$, defined as
\begin{equation}
\label{gSLOPE_1}
\hat{\beta}: = \argmin b\ \ \bigg\{\frac 12\Big\|y-Xb\Big\|_2^2+\sigma J_{\lambda}\Big( W\I{b}\Big)\bigg\}.
\end{equation}
Assume for simplicity that $\|\beta_{I_1}\|_2>\ldots> \|\beta_{I_s}\|_2 >0$, $\|\beta_{I_j}\|_2=0$ for $j>s$, $\hat{\beta}$ satisfies the same conditions for some $\lambda$ and the true model is sparse. Divide $I$ into two families of sets $I^s:=\{I_1,\ldots,I_s\}$ and $I^c:=\{I_{s+1},\ldots,I_m\}$. To derive optimality condition for $\hat{\beta}$ we will prove the following
\begin{theorem}
\label{27111155}
Let $b\in\mathbb{R}^p$ be such that $\|b_{I_1}\|_2>\ldots>\|b_{I_s}\|_2>0$, $\|b_{I_j}\|_2=0$ for $j>s$ and denote $\lambda^c:=(\lambda_{s+1},\ldots,\lambda_m)\T{T}$. If $g\in \partial J_{\lambda}\big(w\I{b}\big)$, then it holds:
\begin{equation}
\left\{
\begin{array}{l}
g_{I_i}=w\lambda_i\frac{b_{I_i}}{\|b_{I_i}\|_2},\ i=1,\ldots, s\\
\llbracket g\rrbracket_{I^c}\in C_{w\lambda^c}
\end{array}
\right.,
\end{equation}
where the set $C_{\lambda}$ (here with $w\lambda^c$ instead of $\lambda$) is defined in appendix (\ref{subsec:app26112110}).
\end{theorem}
\begin{proof}
For $b\in \mathbb{R}^p$ define $J_{\lambda,I}(b):=J_{\lambda}\big(\I{b}\big)$ and put
$H:=\big\{h\in\mathbb{R}^p:\ \|(b+h)_{I_1}\|_2>\ldots>\|(b+h)_{I_s}\|_2,\ \|(b+h)_{I_s}\|_2>\|h_{I_j}\|_2,\ j>s\big\}.$
If $g\in\partial J_{\lambda,I}(b)$, then for all $h\in H$ from definition of subgradient it holds
\begin{equation}
\label{06191517}
\sum_{i=1}^s\lambda_i\|(b+h)_{I_i}\|_2+\sum_{i=s+1}^m\lambda_i\big(\I{b+h}\big)_{(i)}\geq\sum_{i=1}^s\lambda_i\|b_{I_i}\|_2 + \sum_{i=1}^sg_{I_i}\T{T}h_{I_i}+(g^c)\T{T}h^c,
\end{equation}
for $g^c:=(g_{I_{s+1}}\T{T},\ldots,g_{I_m}\T{T})\T{T}$ and $h^c:=(h_{I_{s+1}}\T{T},\ldots,h_{I_m}\T{T})\T{T}$. Define $\widetilde{I}:=\big\{\widetilde{I}_1,\ldots, \widetilde{I}_{m-s}\big\}$, with set $\widetilde{I}_i:=\big\{(i-1)\cdot l +1,\ldots, i\cdot l\big\}$. Then $\llbracket g^c\rrbracket_{\widetilde{I}}= \llbracket g\rrbracket_{I^c}$. Consider first case, when $h$ belongs to the set $H^c:=\{h\in H:\ h_{I_i}\equiv0,\ i\leq s\}$. This yields
\begin{equation}
\label{06191400}
\sum_{i=1}^{m-s}\lambda_{s+i}\big(\llbracket h^c\rrbracket_{\widetilde{I}}\big)_{(i)}\geq (g^c)\T{T}h^c.
\end{equation}
Since $\{h^c:\ h\in H^c\}$ is open in $\mathbb{R}^{l(m-s)}$ and contains zero, from Proposition \ref{06182058} we have that $g^c\in \partial J_{\lambda^c,\widetilde{I}}(0)$ and the inequality (\ref{06191400}) is true for any $h^c\in\mathbb{R}^{l(m-s)}$ yielding  
\begin{equation}
0\geq \sup_{h^c}\Big\{(g^c)\T{T}h^c -J_{\lambda^c,\widetilde{I}}(h^c)\Big\} = J_{\lambda^c,\widetilde{I}}^*(g^c) = \left\{\begin{array}{cl}
0,&\llbracket g^c\rrbracket_{\widetilde{I}}\in C_{\lambda^c}\\
\infty,&\textrm{otherwise}
\end{array}\right.,
\end{equation}
see Proposition \ref{07071055}. This result immediately gives condition $\llbracket g^c\rrbracket_{\widetilde{I}}\in C_{\lambda^c}$, which is equivalent with $\llbracket g\rrbracket_{I^c}\in C_{\lambda^c}$. To find conditions for $g_{I_i}$ with $i\leq s$, define sets $H_i:=\{h\in H:\ h_{I_j}\equiv0,\ j\neq i\}$. For $h\in H_i$, (\ref{06191517}) reduces to $\lambda_i\|b_{I_i}+h_{I_i}\|_2\geq \lambda_i\|b_{I_i}\|_2+g_{I_i}\T{T}h_{I_i}$. Since the set $\{h_{I_i}:\ h\in H_i\}$ is open in $\mathbb{R}^l$ and contains zero, from Proposition \ref{06182058} we have $g_{I_i}\in \partial f_i(b_{I_i})$ for $f_i:\mathbb{R}^l\longrightarrow\mathbb{R}$, $f_i(x):=\lambda_i\|x\|_2$. Since $f_i$ is convex and differentiable in $b_{I_i}$, it holds $g_{I_i}=\lambda_i\frac{b_{I_i}}{\|b_{I_i}\|_2}$, which finishes the proof.
\end{proof}
\noindent The above theorem allows to write the optimality condition for $\hat{\beta}$ in form 
\begin{equation}
\label{24031126}
\left\{
\begin{array}{l}
X_{I_i}\T{T}(y-X\hat{\beta})=w\lambda_i\frac{\hat{\beta}_{I_i}}{\|\hat{\beta}_{I_i}\|_2},\ i=1,\ldots, s\\
\llbracket X\T{T}(y-X\hat{\beta})\rrbracket_{I^c}\in C_{w\lambda^c}
\end{array}
\right..
\end{equation}
Since $X_{I_i}\T{T}X_{I_i}\approx\mathbf{I}_l$, for $i\leq s$ we get $X_{I_i}\T{T}\left(y - X_{\backslash I_i}\hat{\beta}_{\backslash I_i}\right) \approx \hat{\beta}_{I_i}\left(1+\frac{w\lambda_i}{\|\hat{\beta}_{I_i}\|_2}\right)$, where $X_{\backslash I_i}$ is matrix $X$ without columns from $I_i$ and $\hat{\beta}_{\backslash I_i}$ denotes vector $\hat{\beta}$ with removed coefficients indexed by $I_i$. This means that, for $i=1,\ldots, s$, vector $v_{I_i}:=X_{I_i}\T{T}\left(y - X_{\backslash I_i}\hat{\beta}_{\backslash I_i}\right)$ is approximately collinear with $\hat{\beta}_{I_i}$. Since $1+\frac{w\lambda_i}{\|\hat{\beta}_{I_i}\|_2}>0$, we have $\frac{v_{I_i}}{\|v_{I_i}\|_2} \approx \frac{\hat{\beta}_{I_i}}{\|\hat{\beta}_{I_i}\|_2}$.
This yields $\hat{\beta}_{I_i}\approx\left(1-\frac{w\lambda_i}{\|v_{I_i}\|_2}\right)v_{I_i}$ and consequently $\|\hat{\beta}_{I_i}\|_2\approx\Big|\|v_{I_i}\|_2-w\lambda_i\Big|$. Therefore (\ref{24031126}) can be written as 
\begin{equation}
\label{23061949}
\left\{
\begin{array}{l}
\Big|\|v_{I_i}\|_2-w\lambda_i\Big|\approx\|\hat{\beta_{I_i}}\|_2,\ i=1,\ldots, s\\
\llbracket v \rrbracket_{I^c}\in C_{w\lambda^c}
\end{array}
\right.,
\end{equation}
for $v:=(v_{I_1}\T{T},\ldots,v_{I_m}\T{T})\T{T}$.

The task now is to select $\lambda_i$'s such that condition $\llbracket v\rrbracket_{I^c}\in C_{w\lambda^c}$ regulates the rate of false discoveries. Denote $I_S:=\bigcup_{i=1}^sI_i$. Putting $y=X_{I_S}\beta_{I_S}+z$, we obtain
\begin{equation}
\label{24061549}
v_{I_i} = X_{I_i}\T{T}X_{I_S}(\beta_{I_S}- \hat{\beta}_{I_S}) + X_{I_i}\T{T}z,
\end{equation}
for $i>s$ (irrelevant groups). Under orthogonal design this expression reduces only to the term $X_{I_i}\T{T}z$, and in such situation $\|v_{I_i}\|_2$ has $\Chi$ distribution with $l$ degrees of freedom which was used in subsection \ref{subsec:06232149} to define the sequence $\lambda$. In the considered near-orthogonal situation, the term $X_{I_i}\T{T}X_{I_S}(\beta_{I_S}- \hat{\beta}_{I_S})$ should be also taken into account. The following two assumptions will be important to derive the appropriate approximation of $v_{I_i}$ distribution:
\begin{itemize}
\setlength\itemsep{1pt}
\item the distribution of $v_{I_i}$ could be well approximated by multivariate normal distribution, 
\item for relatively strong effects it occurs $\frac{\hat{\beta}_{I_i}}{\|\hat{\beta}_{I_i}\|_2}\approx\frac{\beta_{I_i}}{\|\beta_{I_i}\|_2}$ for $i=1,\ldots,s$.
\end{itemize}
The first assumption is justified when one works with large data scenario, based on the Central Limit Theorem. In discussion concerning the second assumption it is important to clarify the effect of penalty imposed on entire groups. The magnitudes of coefficients in $\hat{\beta}_{I_i}$, for truly relevant group $i$, are generally significantly smaller than in $\beta_{I_i}$. This, a so-called shrinking effect, is typical for penalized methods. It turns out, however, that under assumed conditions estimates of coefficients of nonzero $\beta_{I_i}$ are pulled to zero proportionally and after normalizing, $\hat{\beta}_{I_i}$ and $\beta_{I_i}$ are comparable. 

From the upper equation in (\ref{24031126}), we have that $X_{I_S}\T{T}(X_{I_S}\beta_{I_S}-X_{I_S}\hat{\beta}_{I_S}) + X_{I_S}\T{T}z \approx wH_{\lambda,\beta}$, for 
\begin{equation}
\label{27091918}
H_{\lambda,\beta}: = \Big(\lambda_1\frac{\beta_{I_1}\T{T}}{\|\beta_{I_1}\|_2},\ldots, \lambda_s\frac{\beta_{I_s}\T{T}}{\|\beta_{I_s}\|_2}\Big)\T{T}, 
\end{equation}
which gives $X_{I_i}\T{T}X_{I_S}(\beta_{I_S}-\hat{\beta}_{I_S}) \approx X_{I_i}\T{T}X_{I_S}(X_{I_S}\T{T}X_{I_S})^{-1}(wH_{\lambda,\beta} - X_{I_S}\T{T}z)$. Combining the last expression with (\ref{24061549}) yields
\begin{equation}
\label{03251704}
v_{I_i}\approx  X_{I_i}\T{T}X_{I_S}(X_{I_S}\T{T}X_{I_S})^{-1}\Big(wH_{\lambda,\beta}-X_{I_S}\T{T}z\Big)+X_{I_i}\T{T}z.
\end{equation}
To determine the parameters of multivariate normal distribution, which best describes the distribution of $v_{I_i}$, we will derive the exact values of the mean and the covariance matrix of the distribution of the right-hand side expression in (\ref{03251704}) for $i>s$. Since $I_i\cap I_S=\emptyset$ and entries of $X$ matrix are randomized independently with $\mathcal{N}\big(0,\frac 1n\big)$ distribution, the expected value of the random vector in (\ref{03251704}) is $0$ and its covariance matrix is provided by the following Lemma.
\begin{lemma}
\label{06241642}
The covariance matrix of $\hat{v}_{I_i}:=X_{I_i}\T{T}X_{I_S}(X_{I_S}\T{T}X_{I_S})^{-1}\Big(wH_{\lambda,\beta}-X_{I_S}\T{T}z\Big)+X_{I_i}\T{T}z$, for $i>s$, is given by the formula
$$Cov(\hat{v}_{I_i}) = \left(\frac{n-ls}n+w^2\frac{\|\lambda^S\|^2_2}{n-ls-1}\right)\mathbf{I}_l,$$ where $\lambda^S:=(\lambda_1,\ldots, \lambda_s)\T{T}.$
\end{lemma}
\noindent Before proving Lemma \ref{06241642}, we will introduce two auxiliary results, proofs of which can be found at the end of this section.
\begin{lemma}
\label{25031835}
Suppose that entries of a random matrix $X\in M(n,r)$, with $r\leq n$, are independently and identically distributed and have a normal distribution with zero mean. Then, there exists the expected value of a random matrix $A_X = X(X\T{T}X)^{-1}X\T{T}$ and $\mathbb{E}\left(A_X\right) = \frac rn \mathbf{I}_n$.
\end{lemma}
\begin{lemma}
\label{31032152}
Suppose that $X\in M(n,r)$, with $r+1<n$, and entries of $X$ are independent and identically distributed, $X_{ij}\sim\mathcal{N}(0,1/n)$ for all $i$ and $j$. Then, there exists expected value of random matrix, $M_{X,\lambda}: = B_XH_{\lambda,\beta}H_{\lambda,\beta}\T{T}B_X\T{T}$, for $B_X=X(X\T{T}X)^{-1}$ and $H_{\lambda,\beta}$ defined in \eqref{27091918}. Moreover, it holds $\mathbb{E}\left(M_{X,\lambda}\right) = \frac {\|\lambda_S\|_2^2}{n-r-1}\, \mathbf{I}_n$.
\end{lemma}
\begin{proof}[Proof of Lemma \ref{06241642}]
We have $\hat{v}_{I_i} =\xi_{X,z} + \zeta_X,$ for $\xi_{X,z}: = X_{I_i}\T{T}\left(\mathbf{I}_n-A_X\right)z$, $\zeta_X:= wX_{I_i}\T{T}B_XH_{\lambda,\beta}$, $A_X: = X_{I_S}(X_{I_S}\T{T}X_{I_S})^{-1}X_{I_S}\T{T}$, $B_X: = X_{I_S}(X_{I_S}\T{T}X_{I_S})^{-1}$. Since $\mathbb{E}(\xi_{X,z}\zeta_X\T{T}) = 0$ and mean of $\hat{v}_{I_i}$ is equal to $0$, it holds $Cov(\hat{v}_{I_i})=Cov(\xi_{X,z})+Cov(\zeta_X)$. Now thanks to Lemma \ref{25031835} and Lemma \ref{31032152}
\begin{equation}
\begin{split}
Cov(\xi_{X,z}) = \ &\mathbb{E}\left[X_{I_i}\T{T}\big(\mathbf{I}_n-A_X\big)zz\T{T}\big(\mathbf{I}_n-A_X\big)\T{T}X_{I_i}\right]=\\
&\mathbb{E}\left[X_{I_i}\T{T}\big(\mathbf{I}_n-A_X\big)\big(\mathbf{I}_n-A_X\big)\T{T}X_{I_i}\right]= \mathbb{E}\left[X_{I_i}\T{T}\big(\mathbf{I}_n-A_X\big)X_{I_i}\right] =\\
&\frac1n\big(n-ls\big)\cdot\mathbb{E}\left[X_{I_i}\T{T}X_{I_i}\right]=\frac1n\big(n-ls\big)\cdot\mathbf{I}_l,
\end{split}
\end{equation}
\begin{equation}
\begin{split}
Cov(\zeta_X) = \ & w^2\mathbb{E}\left[X_{I_i}\T{T}B_XH_{\lambda,\beta}H_{\lambda,\beta}\T{T}B_X\T{T}X_{I_i}\right] = w^2\frac{\|\lambda^S\|_2^2}{n-sl-1}\mathbb{E}\left[X_{I_i}\T{T}X_{I_i}\right] =\\
& w^2\frac{\|\lambda^S\|_2^2}{n-sl-1}\mathbf{I}_l,
\end{split}
\end{equation}
which finishes the proof.
\end{proof}
We have shown that for $i>s$ the distribution of $\|v_{I_i}\|_2$ could be approximated by scaled $\Chi$ distribution with $l$ degrees of freedom and scale parameter $\mathcal{S} = \sqrt{\frac{n-ls}n+\frac{w^2\|\lambda_S\|^2_2}{n-sl-1}}$. Now, analogously to the orthogonal situation, lambdas could be defined as $\lambda_i:=\frac{1}{w_i}F_{\mathcal{S}\Chi_l}^{-1}\left (1-\frac{q\cdot i}m\right )= \frac{\mathcal{S}}{w_i}F_{\Chi_l}^{-1}\left (1-\frac{q\cdot i}m\right )$. Since $s$ is unknown, we will apply the strategy used in \cite{SLOPE}: define $\lambda_1$ as in orthogonal case and for $j\geq 2$ define $\lambda_i$ basing on already generated sequence, according to following procedure.
\begin{algorithm}[H]
{\fontsize{9pt}{5pt}\selectfont
  \caption{Selecting lambdas in near-orthogonal situation: equal groups sizes}
	\label{24021313}
  \begin{algorithmic}
		\State \textbf{input:} $q\in (0,1)$,\ \ $w>0$,\ \ $p,\ n,\ m,\ l \in\mathbb{N}$
		\State $\lambda_1:=\frac1wF_{\chi_{\mathnormal{l}}}^{-1}\left (1-\frac{q}m\right )$;
		\State  \textbf{For} $i\in\{2,\ldots,m\}$:
		\State \indent  $\lambda^S: = (\lambda_1,\ldots, \lambda_{i-1})^\mathsf{T}$;
		\State \indent  $\mathcal{S}: = \sqrt{\frac{n-l(i-1)}n+\frac{w^2\|\lambda^S\|^2_2}{n-l(i-1)-1}}$;
		\State \indent  $\lambda^*_i:=\frac{\mathcal{S}}{w}F_{\chi_l}^{-1}\left (1-\frac{q\cdot i}m\right )$; 
		\State \indent  if $\lambda^*_i\leq \lambda_{i-1}$, then put $\lambda_i:=\lambda^*_i$. Otherwise, stop the procedure and put $\lambda_j:=\lambda_{i-1}$ for $j\geq i$;
		\State  \textbf{end for}
  \end{algorithmic}\label{21091256}
	}
\end{algorithm}

Consider now the Gaussian design with arbitrary group sizes and sequence of positive weights $w_1,\ldots,w_m$. One possible approach is to construct consecutive $\lambda_i$ as the largest scaled quantiles among all distributions, i.e. as $\max\limits_{j=1,\ldots,m}\left\{\frac{\mathcal{S}_j}{w_j}F^{-1}_{\chi_{l_j}}\left (1-\frac{q\cdot i}{m}\right )\right\}$ for corrections $\mathcal{S}_j$'s adjusted to different $l_i$ values (the conservative strategy). In this article, however, we will stick to the more liberal strategy based on $\lambda^{mean}$, which leads to the modified sequence of tuning parameters presented in Procedure \ref{11091537}.



\begin{proposition}
\label{06182058}
For any open set $H$ containing zero the subgradient of convex function $f$ at $b$ could be equivalently defined as a vector $g$ satisfying $f(b+h)\geq f(b)+g^\mathsf{T}h,\ \textrm{for all}\ h\in H.$
\end{proposition}
\begin{proof}
 Suppose that $f$ is convex function and for some $b, g\in \mathbb{R}^p$ it occurs $f(b+h)\geq f(b)+g^\mathsf{T}h$ for $h\in H$, where $H$ is open set containing zero. Let $h_0\in \mathbb{R}^p$ be arbitrary vector. Function $F:\mathbb{R}\rightarrow\mathbb{R}$, defined as ${F(t): = f(b+th_0)-tg^\mathsf{T}h_0}$, is convex. There exists $t_0\in(0,1)$ such that $t_0h_0\in H$, what gives
\begin{equation}
f(b)\leq F(t_0)=F\big((1-t_0)\cdot 0+t_0\cdot 1\big)\leq (1-t_0)f(b)+t_0F(1)
\end{equation}
and $f(b+h_0)\geq f(b)+g^\mathsf{T}h_0$ as a result. 
\end{proof}

\subsection{The proof of Lemma \ref{25031835}}

The claim is obvious for $n=1$ and we will assume that $n>1$. First, we will list some basic properties of $A_X$. It could be easily noticed that: $A_X$ is symmetric matrix, $A_X$ is idempotent matrix (meaning that $A_XA_X=A_X$) and that $\operatorname{trace}(A_X)=\operatorname{trace}(X^\mathsf{T}X(X^\mathsf{T}X)^{-1})=r$. We will now show that for each $i\in\{1,\ldots,n\}$, $j\in\{1,\ldots,r\}$ the support of a $A_X(i,j)$ distribution is bounded, which will give us the existence of the expected value. Let $\|A\|_F$ be the Frobenius norm. Then 
\begin{equation}
\big|(A_X)_{i,j}\big| \leq \|A\|_F = \sqrt{\operatorname{trace}(A_X^\mathsf{T}A_X)}=\sqrt{\operatorname{trace}(A_X)}=\sqrt{r}.
\end{equation}
We will use notation $E_X: = \mathbb{E}\left(A_X\right)$. Since entries of matrix $X$ are randomized independently with the same distribution, $E_X$ is invariant under permutation applied to rows, i.e. $E_X=E_{PX}$ for any permutation matrix $P$. This gives $E_X = PE_XP^\mathsf{T}$, which means that applying the same permutation to rows and columns has no impact on expected value. We will show that
\begin{equation}
\label{24032202}
(E_X)_{i,j}=(E_X)_{1,n},\textrm{ for }i<j.
\end{equation}
Consider first the case when $i=1$ and $1<j<n$. Denoting by $P_{j\leftrightarrow n}$ matrix corresponding to transposition which replaces elements $j$ and $n$, we have $(E_X)_{1,j} = \big(P_{j\leftrightarrow n}E_XP_{j\leftrightarrow n}^\mathsf{T}\big)_{1,j}= (E_X)_{1,n}$. When $j=n$ and $1<i<n$, the same reasoning works with $P_{1\leftrightarrow i}$. Suppose now, that $1<i<n$ and $1<j<n$. We get $(E_X)_{i,j}=(E_X)_{1,n}$ analogously by using arbitrary permutation matrix $P$ which replaces element $j$ with $n$ and element $i$ with $1$. Since $A_X$ is symmetric, (\ref{24032202}) is true also for $i>j$. On the other hand, for all $i,j \in\{1,\ldots, n\}$, we have $(E_X)_{i,i} = \big(P_{j\leftrightarrow i}E_XP_{j\leftrightarrow i}^\mathsf{T}\big)_{i,i}= (E_X)_{j,j}$. Consequently, all off-diagonal entries of $E_X$ are equal to some $t$ and all diagonal entries have the same value $d$. Since
\begin{equation}
nd = \operatorname{trace}(E_X) = \sum_{i=1}^n\mathbb{E}(A_X(i,i)) = \mathbb{E}\left(\sum_{i=1}^n(A_X)_{i,i}\right) = r,
\end{equation}
we have $d=\frac rn$ and it remains to show that $t=0$. Define $\Sigma: = \left[\begin{BMAT}{c0c}{c0c}-1&\mathbf{0}^\mathsf{T}\\ \mathbf{0}&\mathbf{I}_{n-1}\end{BMAT}\right]$. Then $\Sigma X_S$ differs from $X_S$ only by signs of the first row. Since entries of matrix $X_S$ have zero-symmetric distribution, we have $E_X = E_{\Sigma X}$. Now
\begin{equation}
\left[\begin{BMAT}{c0c}{c0c}d&\mathbf{1}_{n-1}^\mathsf{T}t\\ \mathbf{1}_{n-1}t& \ddots\end{BMAT}\right] = E_X = \Sigma E_X \Sigma = \left[\begin{BMAT}{c0c}{c0c}d&-\mathbf{1}_{n-1}^\mathsf{T}t\\ -\mathbf{1}_{n-1}t& \ddots\end{BMAT}\right], 
\end{equation}
which implies $t=0$ and proves the statement.


\subsection{The proof of Lemma \ref{31032152}}
It is easy to see that $M_{X,\lambda}$ is symmetric, positive semi-definite matrix. Denote by $\|M_{X,\lambda}\|_*$ the nuclear (trace) norm of matrix $M_{X,\lambda}$. We have
\begin{equation}
\label{01040935}
\begin{split}
&\mathbb{E}\big|(M_{X,\lambda})_{i,j}\big|\leq \mathbb{E}\big(\|M_{X,\lambda}\|_*\big)=\mathbb{E}\big(\operatorname{trace}(M_{X,\lambda})\big)=\mathbb{E}\big(\operatorname{trace}(H_{\lambda,\beta}^\mathsf{T}B_X^\mathsf{T}B_XH_{\lambda,\beta})\big)=\\
&\mathbb{E}\big(H_{\lambda,\beta}^\mathsf{T}(X^\mathsf{T}X)^{-1}H_{\lambda,\beta}\big)=\frac {n} {(n-r-1)}\,H_{\lambda,\beta}^\mathsf{T}H_{\lambda,\beta}=\frac {n\,\|\lambda^S\|_2^2} {n-r-1},
\end{split}
\end{equation}
since $X^\mathsf{T}X$ follows an inverse Wishart distribution. This gives the existence of $E_X:=\mathbb{E}(M_{X,\lambda})$. Analogously to situation in Lemma \ref{25031835}, $E_X$ is invariant under permutation or signs changes applied to rows of $X$, i.e. $E_X=E_{PX}$ for any permutation matrix $P$, and $E_X = E_{\Sigma X}$ for diagonal matrix $\Sigma$ with entries on diagonal coming from set $\{-1,1\}$. Since $E_{PX} = PE_XP^\mathsf{T}$ and $E_{\Sigma X} = \Sigma E_X\Sigma$, as before we have that $E_X$ is diagonal matrix with all diagonal entries having the same value $d$. The value $d$ could be easy found using (\ref{01040935}) since we have
\begin{equation}
n d = \operatorname{trace}\big(E_X\big) = \frac {n\,\|\lambda^S\|_2^2} {n-r-1}. 
\end{equation}

\end{document}